\documentclass{article}

\usepackage{graphicx}
\usepackage{amsthm,amsmath,amssymb} 
\usepackage{color}
\usepackage[mathscr]{euscript}
\usepackage[noadjust]{cite}
\usepackage{multirow}
\usepackage{tabu,array,soul}
\usepackage[a4paper, total={5.5in, 7.5in}]{geometry}
\usepackage{subcaption}
\usepackage{caption}

\newtheorem{theorem}{Theorem}
\newtheorem{observation}{Observation}
\newtheorem{definition}{Definition}
\renewcommand*{\theobservation}{\Alph{observation}}

\newtheorem{corollary}[theorem]{Corollary}

\newenvironment{claim}{\medskip\noindent\textit{Claim. }\itshape}{\medskip}
\newtheorem{lemma}[theorem]{Lemma}
\newtheorem{question}{Question}

\newcommand{\com}[1]{}

\begin{document}
\bibliographystyle{plain}
\title{On the stab number of rectangle intersection graphs}
\date{}
\author{Dibyayan Chakraborty\thanks{Indian Statistical Institute, Kolkata, India. E-mail: \texttt{dibyayancg@gmail.com}}\and Mathew C. Francis\thanks{Indian Statistical Institute, Chennai, India. E-mail: \texttt{mathew@isichennai.res.in}. Partially supported by the DST-INSPIRE Faculty Award IFA12-ENG-21.}}

\maketitle

\begin{abstract}
We introduce the notion of \emph{stab number} and \emph{exact stab number} of rectangle intersection graphs, otherwise known as graphs of boxicity at most 2. A graph $G$ is said to be a \emph{$k$-stabbable rectangle intersection graph}, or \emph{$k$-SRIG} for short, if it has a rectangle intersection representation in which $k$ horizontal lines can be chosen such that each rectangle is intersected by at least one of them. If there exists such a representation with the additional property that each rectangle intersects exactly one of the $k$ horizontal lines, then the graph $G$ is said to be a \emph{$k$-exactly stabbable rectangle intersection graph}, or \emph{$k$-ESRIG} for short. The stab number of a graph $G$, denoted by $stab(G)$, is the minimum integer $k$ such that $G$ is a $k$-SRIG. Similarly, the exact stab number of a graph $G$, denoted by $estab(G)$, is the minimum integer $k$ such that $G$ is a $k$-ESRIG. In this work, we study the stab number and exact stab number of some subclasses of rectangle intersection graphs. A lower bound on the stab number of rectangle intersection graphs in terms of its pathwidth and clique number is shown. Tight upper bounds on the exact stab number of split graphs with boxicity at most 2 and block graphs are also given. We show that for $k\leq 3$, $k$-SRIG is equivalent to $k$-ESRIG and for any $k\geq 10$, there is a tree which is a $k$-SRIG but not a $k$-ESRIG. We also develop a forbidden structure characterization for block graphs that are 2-ESRIG and trees that are 3-ESRIG, which lead to polynomial-time recognition algorithms for these two classes of graphs. These forbidden structures are natural generalizations of asteroidal triples. Finally, we construct examples to show that these forbidden structures are not sufficient to characterize block graphs that are 3-SRIG or trees that are $k$-SRIG for any $k\geq 4$.
\medskip

\end{abstract}

\section{Introduction}
A \emph{rectangle intersection representation} of a graph is a collection of axis-parallel rectangles on the plane such that each rectangle in the collection represents a vertex of the graph and two rectangles intersect if and only if the vertices they represent are adjacent in the graph. The graphs that have rectangle intersection representation are called \emph{rectangle intersection graphs}.
The \emph{boxicity} $box(G)$ of a graph $G$ is the minimum $d$ such that $G$ is representable as a geometric intersection graph of $d$-dimensional (axis-parallel) hyper-rectangles. A graph $G$ is an interval graph if $box(G)=1$ and $G$ is a rectangle intersection graph if $box(G)\leq 2$.

A \emph{$k$-stabbed rectangle intersection representation} is a rectangle intersection representation, along with a collection of $k$ horizontal lines called \emph{stab lines}, such that every rectangle intersects at least one of the stab lines. A graph $G$ is a \emph{$k$-stabbable rectangle intersection} graph (\emph{$k$-SRIG}), if there exists a $k$-stabbed rectangle intersection representation of $G$. The \emph{stab number} of a rectangle intersection graph, denoted by $stab(G)$, is the minimum integer $k$ such that there exists a $k$-stabbed rectangle intersection representation of $G$. In other words $stab(G)$ is the minimum integer $k$ such that $G$ is $k$-SRIG. Clearly, if a graph $G$ has boxicity at most 2, then $stab(G)$ is finite. For graphs $G$ with boxicity at least three, we define $stab(G)=\infty$.  

A \emph{$k$-exactly stabbed rectangle intersection representation} is a $k$-stabbed rectangle intersection representation in which every rectangle intersects exactly one of the stab lines.
A graph $G$ is a \emph{$k$-exactly stabbable rectangle intersection} graph, or \emph{$k$-ESRIG} for short, if there exists a $k$-exactly stabbed rectangle intersection representation of $G$.
The \emph{exact stab number} of a rectangle intersection graph, denoted by $estab(G)$, is the minimum integer $k$ such that there exists a $k$-exactly stabbed rectangle intersection representation of $G$. In other words, $estab(G)$ is the minimum integer $k$ such that $G$ is $k$-ESRIG. When a graph $G$ has no $k$-exactly stabbed rectangle intersection representation for any integer $k$, we define $estab(G)=\infty$. A graph $G$ with $estab(G)<\infty$ is said to be an \emph{exactly stabbable rectangle intersection graph}. Note that for a graph $G$, $stab(G)\leq estab(G)$ and that a graph $G$ is an interval graph if and only if $stab(G)=estab(G)=1$, or in other words, the class of interval graphs, the class of 1-SRIGs, and the class of 1-ESRIGs are all the same.

For a subclass $\mathcal{C}$ of rectangle intersection graphs, $stab(\mathcal{C},n)$ is the minimum integer $k$ such that any graph $G\in\mathcal{C}$ with $n$ vertices satisfies $stab(G)\leq k$, and $estab(\mathcal{C},n)$ is the minimum integer $k$ such that for any graph $G\in\mathcal{C}$ with $n$ vertices satisfies $estab(G)\leq k$. A \emph{unit height rectangle intersection} graph $G$ is a graph that has a rectangle intersection representation in which all rectangles have equal height. It is well-known that all unit height rectangle intersection graphs are exactly stabbable rectangle intersection graphs (for the sake of completion, we prove this in Theorem~\ref{thm:notunitheight} in Section~\ref{sec:basic}).

\subsection{Motivation and related work}
Boxicity of a graph has been an active field of research for many decades~\cite{esperet2013,adiga2014,chandran2008,Chandran2016,chandran2007}. While recognizing graphs with boxicity at most $d$ is NP-complete for all $d\geq 2$~\cite{kratochvil1994,yannakakis1982}, there are efficient algorithms to recognize interval graphs, i.e. graphs with boxicity at most 1~\cite{lekkerkerkerboland,corneil2009}. There seems to be a ``jump in the difficulty level" of problems as the boxicity of the input graph increases from 1 to 2. For example, the {\sc Maximum Independent Set} and {\sc Chromatic Number} problems, while being linear-time solvable for interval graphs, become NP-complete for rectangle intersection graphs (even with the rectangle intersection representation given as input)~\cite{kratochvil1990,imai1983}. Our goal is to understand the reason of this jump by studying graph classes that lie ``in between'' interval graphs and rectangle intersection graphs. For this purpose, we introduce a parameter called stab number for rectangle intersection graphs. The concept of stab number is a generalization of the idea behind a class of graphs known as ``2SIG'', which was introduced in an earlier paper~\cite{bhore2015}. Even though our definitions of 2-SRIG and 2-ESRIG are both slightly different from that of ``2SIG'', all three classes of graphs turn out to be equivalent (Theorem~\ref{thm:equiv} shows that the classes $k$-SRIG and $k$-ESRIG are equivalent for any $k\leq 3$). A $k$-stabbed rectangle intersection representation of a graph involves rectangles and horizontal lines. Such combined arrangements of lines and rectangles have been popular topics of study in the geometric algorithms community. For example, such arrangements appear in the works of Agarwal et al.~\cite{agarwal1998} and Chan~\cite{chan2004}, who gave approximation algorithms for the {\sc Maximum Independent Set} problem in unit height rectangle intersection graphs, and also in a paper by Erlebach et al.~\cite{erlebach2010}, who proposed a PTAS for {\sc Minimum Weight Dominating Set} for unit square intersection graphs. Correa et al.~\cite{correa2014} have studied the problems of computing independent and hitting sets for families of rectangles intersecting a diagonal line.

\subsection{Contributions and organization of the paper}
In this paper, we introduce the notion of ``stab number'' of a rectangle intersection graph and study this parameter for various subclasses of rectangle intersection graphs. In Section~\ref{preliminaries}, we give some definitions and notation that will be used throughout the paper. We prove some basic results about $k$-SRIGs and $k$-ESRIGs in Section~\ref{sec:basic}. We first show a simple necessary and sufficient condition for a graph to be a $k$-ESRIG and also show why the classes $k$-SRIG and $k$-ESRIG are equivalent when $k\leq 3$ (Theorem~\ref{thm:equiv}). Then we prove that the class of unit height rectangle intersection graphs is a proper subset of the class of rectangle intersection graphs with finite exact stab number (Theorem~\ref{thm:notstabbable}), which in turn is a proper subset of rectangle intersection graphs (Theorem~\ref{thm:notunitheight}). This leads us to the natural question of finding exactly stabbable graphs whose exact stab number is strictly greater than the stab number. We show that for each $k\geq 10$, there exist trees which are $k$-SRIG but not $k$-ESRIG (Theorem~\ref{thm:srig-parameter-diff}). Therefore, even for graphs that are exactly stabbable, like trees (Theorem~\ref{thm:blockstab}), the stab number and the exact stab number may differ. We prove this result only in Section~\ref{sec:notesrig}, after the machinery required for the proof is developed in Section~\ref{sec:treesandblocks}. In Section~\ref{sec:classes}, we show a lower bound on the stab number of rectangle intersection graphs in terms of the clique number and the pathwidth, and then study upper bounds on the stab number of rectangle intersection graphs that are also (a) split graphs, or (b) block graphs. In particular, we show (a) that all rectangle intersection graphs that are also split graphs have exact stab number at most 3 and that this bound is tight, and (b) an upper bound of $\lceil\log m\rceil$ on the exact stab number of block graphs with $m$ blocks (this bound is shown to be asymptotically tight in Section~\ref{sec:gl}). Then in Section~\ref{sec:asteroidal}, we describe a forbidden structure for $k$-SRIG and $k$-ESRIG, which we call ``asteroidal-(non-($k-1$)-SRIG)'' subgraphs and ``asteroidal-(non-($k-1$)-ESRIG)'' subgraphs respectively. These obstructions are a natural generalization of the well-known ``asteroidal-triples'' of Lekkerkerker and Boland~\cite{lekkerkerkerboland}, which are obstructions for interval graphs. In Section~\ref{sec:colorblocktree}, we discuss some general properties possessed by the block-trees of graphs without these kinds of obstructions. In Section~\ref{sec:treesandblocks}, we show that the absence of these forbidden structures is enough to characterize block graphs that are 2-ESRIG (Theorem~\ref{thm:block2sig}) and trees that are 3-ESRIG (Theorem~\ref{thm:tree3sig}). These results lead to polynomial-time algorithms to recognize block graphs that are 2-SRIG and trees that are 3-SRIG. In Section~\ref{sec:notsuff}, we develop a geometric argument that allows us to show that this kind of forbidden structure is not sufficient to characterize block graphs that are 3-SRIG (Theorem~\ref{thm:blocknot3sig}) or trees that are $k$-SRIG, for any $k\geq 4$ (Theorem~\ref{thm:treenotsuff}). We conclude by listing some open problems and suggesting some possible directions for further research on this topic.

\section{Preliminaries}\label{preliminaries}
We present some definitions in this section. 
Let $G$ be a graph with vertex set $V(G)$ and edge set $E(G)$.
Let $N(v) = \{u \in V(G)\colon uv \in E(G)\}$ and $N[v] = N(v) \cup \{v\}$ denote the \emph{open neighbourhood} and the \emph{closed neighbourhood} of a vertex $v$, respectively. For $S\subseteq V(G)$, we denote by $G[S]$ the subgraph induced in $G$ by the vertices in $S$, and by $G-S$ the graph obtained by removing the vertices in $S$ from $G$. For an edge $e\in E(G)$, we denote by $G-e$ the graph on vertex set $V(G)$ having edge set $E(G)\setminus\{e\}$.

Let $G$ be a rectangle intersection graph with rectangle intersection representation $\mathcal{R}$. A rectangle in $\mathcal{R}$ corresponding to the vertex $v$ is denoted as $r_v$.
All rectangles considered in this article are closed rectangles. Denote by $x^+_v$ $(x^-_v)$, the $x-$coordinate of the right (left) bottom corner of $r_v$. Also $y^+_v$ $(y^-_v)$ is the $y-$coordinate of the left top (bottom) corner of $r_v$. In other words, $r_v=[x^-_v,x^+_v]\times [y^-_v,y^+_v]$. The \emph{span} of a vertex $u$, denoted as $span(u)$, is the projection of $r_u$ on the $X-$axis, i.e. $span(u)=[x^-_u,x^+_u]$. For two intervals $I_1=[a_1,b_1]$ and $I_2=[a_2,b_2]$, we write $I_1 < I_2$ to indicate that $b_1<a_2$. Clearly, $I_1\cap I_2=\emptyset$ if and only if $I_1<I_2$ or $I_2<I_1$.
For an edge $uv\in E(G)$, we define $span(uv)=span(u)\cap span(v)$.

Let $G$ be a $k$-SRIG with a $k$-stabbed rectangle intersection representation $\mathcal{R}$ in which the stab lines are $y=a_1$, $y=a_2$, $\ldots$, $y=a_k$, where $a_1<a_2<\cdots<a_k$.
The \emph{top} (resp. \emph{bottom}) stab line of $\mathcal{R}$ is the stab line $y=a_k$ (resp. $y=a_1$). For $1\leq i<k$, we say that $y=a_{i+1}$ is the stab line ``just above'' the stab line $y=a_i$ and that $y=a_i$ is the stab line ``just below'' the stab line $y=a_{i+1}$. We also say that the stab lines $y=a_i$ and $y=a_{i+1}$ are ``consecutive''. A vertex $u\in V(G)$ is said to be ``on'' a stab line if $r_u$ intersects that stab line.
Two vertices $u,v$ of $G$ ``have a common stab'' if there is some stab line that intersects both $r_u$ and $r_v$. Similarly, a set of vertices is said to have a common stab if there is one stab line that intersects the rectangles corresponding to each of them.
It is easy to see that if $uv\in E(G)$, then there must be either a stab line such that $u$ and $v$ are on it or two consecutive stab lines such that $u$ is on one of them and $v$ is on the other. Whenever the $k$-stabbed rectangle intersection representation of a graph $G$ under consideration is clear from the context, the terms $r_u$, $x^-_u$, $x^+_u$, $y^-_u$, $y^+_u$, for every vertex $u\in V(G)$ and usages such as ``on a stab line'', ``have a common stab'', ``span'' etc. are considered to be defined with respect to this representation.
Clearly, both the classes $k$-SRIG and $k$-ESRIG are closed under taking induced subgraphs. We say that a graph is a non-$k$-SRIG (resp. non-$k$-ESRIG) if it is not a $k$-SRIG (resp. $k$-ESRIG). Similarly, we say that a graph is a non-interval graph if it is not an interval graph.

\section{Basic Results}\label{sec:basic}
Given a collection $\mathcal{I}$ of intervals, a \emph{hitting set} $X$ of $\mathcal{I}$ is a subset of $\mathbb{R}$ such that each interval in $\mathcal{I}$ contains at least one element of $X$. The set $X$ is an \emph{exact hitting set} of $\mathcal{I}$ if each interval in $\mathcal{I}$ contains exactly one element of $X$. An interval graph $G$ is said to have an exact hitting set \emph{of size $k$} if there exists an interval representation $\mathcal{I}$ of $G$ that has an exact hitting set of cardinality $k$. Note that some collections of intervals may not have an exact hitting set of any cardinality. Also, there are interval graphs (for example, $K_{1,4}$) that have no exact hitting set.

\begin{theorem}
A graph $G$ is a $k$-ESRIG if and only if there exists two interval graphs $I_1$ and $I_2$ such that $V(G)=V(I_1)=V(I_2)$ and $E(G)=E(I_1)\cap E(I_2)$ and at least one of $I_1,I_2$ has an exact hitting set of size $k$.
\end{theorem}

\begin{proof}
First we prove that if $G$ has a $k$-ESRIG representation, then there exist two interval graphs $I_1$ and $I_2$ such that $V(G)=V(I_1)=V(I_2)$ and $E(G)=E(I_1)\cap E(I_2)$ and at least one of them has an exact hitting set of size $k$. Let $\mathcal{R}$ be a $k$-exactly stabbed rectangle intersection representation of $G$ and $\{y=a_1,y=a_2,\ldots,y=a_k\}$ be the set of stab lines in $\mathcal{R}$. Let $I_x,I_y$ be the interval graphs formed by taking the projections of the rectangles in $\mathcal{R}$ on the $X$ and $Y$ axes, respectively. In other words, $I_x$ is the interval graph given by the interval representation $\{[x^-_u,x^+_u]\}_{u\in V(G)}$ and $I_y$ is the interval graph given by the interval representation $\{[y^-_u,y^+_u]\}_{u\in V(G)}$. It is clear that $V(G)=V(I_x)=V(I_y)$ and $E(G)=E(I_x)\cap E(I_y)$. Furthermore, the set $S=\{a_1,a_2,\ldots,a_k\}$ is an exact hitting set of the interval representation $\{[y^-_u,y^+_u]\}_{u\in V(G)}$ of $I_y$. Hence, $I_y$ has an exact hitting set of size $k$.

Now assume that there exist two interval graphs $I_1$ and $I_2$ such that $V(G)=V(I_1)=V(I_2)$ and $E(G)=E(I_1)\cap E(I_2)$ and at least one of them, say $I_1$, has an exact hitting set of size $k$. Let $S=\{a_1,a_2,\ldots,a_k\}$ be an exact hitting set of an interval representation $\{[c_u,d_u]\}_{u\in V(G)}$ of $I_1$. Also, let $\{[c'_u,d'_u]\}_{u\in V(G)}$ be an interval representation of $I_2$. For each $u\in V(G)$, define $r_u=[c'_u,d'_u]\times [c_u,d_u]$. It is easy to see that $\mathcal{R}=\{r_u\}_{u\in V(G)}$ is a rectangle intersection representation of $G$. Further, the lines $y=a_1$, $y=a_2$, $\ldots$, $y=a_k$  are horizontal lines such that each rectangle in $\mathcal{R}$ intersects exactly one of them. Hence, $\mathcal{R}$, together with these lines, is a $k$-exactly stabbed rectangle intersection representation of $G$ and therefore, $G$ is a $k$-ESRIG. This completes the proof.
\end{proof}

\begin{theorem}\label{thm:equiv}
When $k\leq 3$, the classes $k$-SRIG and $k$-ESRIG are equivalent.
\end{theorem}
\begin{proof}
If a graph $G$ is $k$-ESRIG for some $k$, then $G$ is also $k$-SRIG. Therefore it suffices to prove that if a graph $G$ has a $k$-stabbed rectangle intersection representation for some $k\leq 3$, then $G$ also has a $k$-exactly stabbed rectangle intersection representation. If $k=1$, then there is nothing to prove. So we shall assume that $k\in\{2,3\}$. Let $\mathcal{R}$ be a $k$-stabbed rectangle intersection representation of a graph $G$ with $k\leq 3$ with stab lines $y=0$, $y=1$, $\ldots$, $y=k-1$. We can assume without loss of generality that for any two distinct vertices $u,v\in V(G)$, we have $\{y^+_u,y^-_u\}\cap \{y^+_v,y^-_v\}=\emptyset$ and that for any vertex $v\in V(G)$, we have $ \{y^+_v,y^-_v\}\cap \{0,1,2\}=\emptyset$ (note that if this is not the case, then the rectangles in $\mathcal{R}$ can be perturbed slightly so that these conditions are satisfied). Let $S=\{y^+_v,y^-_v\}_{v\in V(G)} \cup \{0,1,2\}$ and $\epsilon$ be a positive real number such that $\epsilon < \min\{|a - b|\colon a,b\in S,a\neq b\}$. Let $M=\{u\in V(G)\colon r_u$ intersects the stab line $y=1\}$. For each vertex $u\in M$, define $r'_u=[x^-_u,x^+_u]\times [y'^-_u,y'^+_u]$, where $y'^-_u=\max\{\epsilon,y^-_u\}$ and $y'^+_u=\min\{2-\epsilon,y^+_u\}$. Let $\mathcal{R}'$ be the rectangle intersection representation given by the collection of rectangles $(\mathcal{R}\setminus\{r_u\colon u\in M\})\cup\{r'_u\colon u\in M\}$. It is now easy to verify that $\mathcal{R}'$ is a $k$-exactly stabbed rectangle intersection representation of $G$. Indeed, $\mathcal{R}'$ is obtained from $\mathcal{R}$ by the vertical shortening of some of the rectangles intersecting the stab line $y=1$, and we only need to show that every rectangle that is so shortened still intersects with all the rectangles with which it originally has an intersection. The definition of $\epsilon$ guarantees that in $\mathcal{R}$, the bottom edge of any rectangle is no higher than $2-\epsilon$ and the top edge of any rectangle is no lower than $\epsilon$. So when a rectangle is shortened in the manner described above, it does not become disjoint from a rectangle with which it previously had a nonempty intersection. Therefore $\mathcal{R}$ is a valid rectangle intersection representation of $G$. It is clear that any rectangle that intersects the stab line $y=1$ in $\mathcal{R}$ intersects only the stab line $y=1$ in $\mathcal{R}'$. This implies that $\mathcal{R}'$ is a $k$-exactly stabbed rectangle intersection representation of $G$.
\end{proof}

In the following theorem, we show that for $k=4$, the classes $k$-SRIG and $k$-ESRIG differ.

\begin{theorem}\label{thm:notstabbable}
There is a graph $G$ such that $stab(G)\leq 4$ and $estab(G)=\infty$.
\end{theorem}

\begin{proof}
We let $G=K_{4,4}$, i.e. the complete bipartite graph in which each partite set contains four vertices each. Clearly, $G$ is a rectangle intersection graph with $stab(G)\leq 4$ (see Figure~\ref{fig:boxrep}(a)). We shall prove that $estab(G)=\infty$, or in other words, $G$ is not an exactly stabbable rectangle intersection graph. First we prove the following claim.
	
\begin{claim}
Let $C$ be a cycle of length four and $E(C)=\{ab,bc,cd,da\}$. There is no $k$-exactly stabbed rectangle intersection representation of $C$, for any integer $k$, in which $a,c$ have a common stab and $b,d$ have a common stab.
\end{claim}
	
\noindent
Assume for the sake of contradiction that there is a $k$-exactly stabbed rectangle intersection representation $\mathcal{R}$ of $C$, for some integer $k$, in which $a,c$ have a common stab and $b,d$ have a common stab. Clearly, $a,b,c,d$ cannot all be on one stab line (as $C$ is not an interval graph). Since every vertex is on exactly one stab line and because $ab\in E(C)$, we can assume without loss of generality that $a,c$ are on the stab line just below the stab line on which $b,d$ are. Since $a,c$ and $b,d$ are nonadjacent in $C$, again without loss of generality we can assume that $span(a)<span(c)$. Since $b\in N(a)\cap N(c)$, we can infer that $[x^+_a,x^-_c] \subset span(b)$. Similarly, we can show that $[x^+_a,x^-_c] \subset span(d)$. But this implies that $[x^+_a,x^-_c]\subset span(b) \cap span(d)$. Since $b,d$ are on the same stab line, this means that $r_b\cap r_d\neq \emptyset$. As $bd\notin E(C)$, this contradicts the fact that $\mathcal{R}$ is a rectangle intersection representation of $C$. This proves the claim.
	
\medskip
Now suppose that $G$ has a $k$-exactly stabbed rectangle intersection representation $\mathcal{R}$ for some $k$. Let $V_1,V_2$ be the two partite sets of $G$ (recall that $G$ is isomorphic to $K_{4,4}$) and $v\in V_1$ be a vertex on some stab line $\ell$. Since each vertex is on exactly one stab line, and all vertices of $V_2$ are adjacent to $v$, we know that each vertex of $V_2$ must be on the stab line $\ell$, on the stab line just above $\ell$, or on the stab line just below $\ell$. By Pigeon Hole Principle, there exists $u,w\in V_2$ such that $u$ and $w$ are both on one of these stab lines, say $\ell_1$. Now, for the same reason as before, each vertex of $V_1$ must be on the stab line $\ell_1$, on the stab line just above $\ell_1$, or on the stab line just below $\ell_1$. Again by Pigeon Hole Principle, there are two vertices $u',w'\in V_1$ such that $u'$ and $w'$ are both on one of these stab lines. Now, consider the cycle $C$ of length four with $E(C)=\{u'u,uw',w'w,wu'\}$, that is an induced subgraph of $G$. It can be seen that the rectangles in $\mathcal{R}$ corresponding to the vertices of $C$ form a $k$-exactly stabbed rectangle intersection representation of $C$ in which $u',w'$ have a common stab and $u,w$ have a common stab. This contradicts the claim proved above. Therefore, $G$ cannot have a $k$-exactly stabbed rectangle intersection representation for any $k$.
\end{proof}

\begin{corollary}
The class of exactly stabbable rectangle intersection graphs is a proper subset of the class of rectangle intersection graphs.
\end{corollary}

The above theorem shows that there are graphs whose stab number is a constant but their exact stab number is infinite. Later on, in Theorem~\ref{thm:srig-parameter-diff}, we shall show that there are even trees whose stab number and exact number differ, even though both these parameters are finite for trees.

\begin{figure}
	\centering
	\begin{tabular}{cc} 
		\includegraphics[page=1]{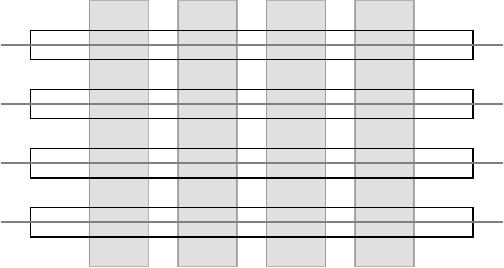}& \includegraphics[page=2]{figures.pdf}\\
		(a)&(b)
	\end{tabular}
	\caption{(a) A 4-stabbed rectangle intersection representation of $K_{4,4}$, (b) a 3-exactly stabbed rectangle intersection representation of $K_{3,3}$.}
	\label{fig:boxrep}
\end{figure}

\begin{theorem}\label{thm:notunitheight}
The class of unit height rectangle intersection graphs is a proper subset of the class of exactly stabbable rectangle intersection graphs.
\end{theorem}

\begin{proof}
We shall first give a proof for the well-known fact that every unit height rectangle intersection graph is an exactly stabbable rectangle intersection graph. We shall prove the following stronger claim.

\begin{claim}
Given a unit height rectangle intersection representation $\mathcal{R}$ for a graph $G$, there exists a set of horizontal lines $y=a_1$, $y=a_2$, $\ldots$, $y=a_k$ (for some integer $k$), where $a_1<a_2<\cdots<a_k$, such that each rectangle in $\mathcal{R}$ intersects exactly one of them and $a_1=\min_{u\in V(G)}\{y^+_u\}$.
\end{claim}

\noindent Let $a=\min_{u\in V(G)}\{y^+_u\}$ and let $S=\{u\colon u\in V(G)$ and $a\in [y^-_u,y^+_u]\}$. Now consider the unit height rectangle intersection representation $\mathcal{R}'=\mathcal{R}\setminus \{r_u\}_{u\in S}$ of $G'=G-S$. By the induction hypothesis, there exists a set of horizontal lines $y=a'_1$, $y=a'_2$, $\ldots$, $y=a'_{k'}$, for some integer $k'$, where $a'_1<a'_2<\cdots<a'_{k'}$, such that each rectangle in $\mathcal{R}'$ intersects exactly one of them and $a'_1=\min_{u\in V(G')}\{y^+_u\}$. Since every rectangle in $\mathcal{R}'$ lies completely above the horizontal line $y=a$, we have that $\min_{u\in V(G')}\{y^+_u\}>a+1$. Therefore, we have $a'_1-a>1$. Since $a'_1<a'_2<\cdots<a'_{k'}$, this means that for $1\leq i\leq k'$, no rectangle of $\mathcal{R}$ intersects both the horizontal lines $y=a'_i$ and $y=a$. Since every rectangle in $\{r_u\}_{u\in S}$ intersects the horizontal line $y=a$, and every rectangle in $\{r_u\}_{u\in V(G')}$ intersects exactly one of the horizontal lines $y=a'_1$, $y=a'_2$, $\ldots$, $y=a'_{k'}$, it follows that each rectangle of $\mathcal{R}$ intersects exactly one of the horizontal lines $y=a$, $y=a'_1$, $y=a'_2$, $\ldots$, $y=a'_{k'}$. This proves the claim.
\medskip

We shall now show the existence of an exactly stabbable rectangle intersection graph that is not a unit height rectangle intersection graph.
Consider the graph $K_{3,3}$, i.e. the complete bipartite graph in which each partite set contains three vertices each. Clearly, $K_{3,3}$ is an exactly stabbable rectangle intersection graph (see Figure~\ref{fig:boxrep}(b)). We shall prove that $K_{3,3}$ is not a unit height rectangle intersection graph. 
	
A rectangle intersection representation $\mathcal{R}$ of a graph $G$ is \emph{crossing-free} if for any two rectangles $r_u$ and $r_v$ in $\mathcal{R}$, the regions $r_u\setminus r_v$ and $r_v\setminus r_u$ are both arc-connected. Note that a unit height rectangle intersection representation of a graph is crossing-free. We shall show that if a triangle-free graph $G$ has a crossing-free rectangle intersection representation, then $G$ must be a planar graph. It then follows directly that $K_{3,3}$ is not a unit height rectangle intersection graph.
	
Let $\mathcal{R}$ be a crossing-free rectangle intersection representation of a triangle-free graph $G$ and let $S\subseteq V(G)$ be the set of vertices of $G$ having degree one. Let $H=G-S$. Clearly, $G$ is planar if and only if $H$ is planar. Let $\mathcal{R'}$ be obtained from $\mathcal{R}$ by removing all the rectangles corresponding to the vertices in $S$. Note that $H$ is a triangle-free graph and $\mathcal{R'}$ is crossing-free. 
	
\begin{claim}
There is no rectangle in $\mathcal{R'}$ which is contained in some other rectangle of $\mathcal{R'}$.
\end{claim}
	
\noindent 
Assume for the sake of contradiction that for vertices $u,v\in V(H)$ we have $r_u\subseteq r_v$ in $\mathcal{R'}$. Since $u$ is a vertex of $H$, we know that $u$ must have degree at least two in $G$. Let $w$ be a neighbour of $u$ other than $v$ in $G$. Then in $\mathcal{R}$, we have $r_w\cap r_u\neq\emptyset$. Since $r_u\subseteq r_v$, this implies that $r_w\cap r_v\neq\emptyset$. But now $u,v,w$ form a triangle in $G$, contradicting the fact that $G$ is triangle-free. This proves the claim.
	
\medskip

Since $H$ is triangle-free, we have that in $H$, for any vertex $u\in V(H)$ and any two vertices in $v,w\in N(u)$, $r_v\cap r_w=\emptyset$. This, together with the fact that $\mathcal{R'}$ is crossing free, implies that the region $r_u\setminus \bigcup_{w\in N(u)} r_w$ is arc-connected and non-empty. (To see this, observe that if $r_u\setminus \bigcup_{w\in N(u)} r_w$ is non-empty, but is not arc-connected, then there exists two points $x,y\in r_u$ and a simple curve $\mathbf{c}\subseteq\bigcup_{w\in N(u)} r_w$ such that $x$ and $y$ are in different arc-connected components of $r_u\setminus\mathbf{c}$. Since for any two vertices in $v,w\in N(u)$, we have $r_v\cap r_w=\emptyset$, we know that there exists some $z\in N(u)$ such that $\mathbf{c}\subseteq r_z$. But this means that $x$ and $y$ are in different arc-connected components of $r_u\setminus r_z$, contradicting the fact that $\mathcal{R}'$ is crossing-free. If $r_u\setminus \bigcup_{w\in N(u)} r_w$ is empty, then $r_u\subseteq \bigcup_{w\in N(u)} r_w$. Again, since for any two vertices in $v,w\in N(u)$, we have $r_v\cap r_w=\emptyset$, it must be the case that there exists some $z\in N(u)$ such that $r_u\subseteq r_z$. But this contradicts the claim proved above.) Now choose for every vertex $u\in V(H)$, a point $p_u$ in $r_u\setminus \bigcup_{w\in N(u)} r_w$. In other words, $p_u$ is a point in $r_u$ which is not contained in any rectangle other than $r_u$. For every edge $uv\in E(H)$, choose a point $p_{uv}$ that is contained in the rectangular region $r_u\cap r_v$. Further, for each edge $uv\in E(H)$, choose a simple curve $\mathbf{s_{u,v}}$ between $p_u$ and $p_{uv}$ that is completely contained in $r_u$ and a simple curve $\mathbf{s_{v,u}}$ between $p_v$ and $p_{uv}$ that is completely contained in $r_v$ such that for any curve in the collection $\{\mathbf{s_{u,v}},\mathbf{s_{v,u}}\}_{uv\in E(H)}$, none of its interior points are contained in any other curve in the collection. Now the set of simple curves $\{\mathbf{s_{u,v}}\cup\mathbf{s_{v,u}}\}_{uv\in E(H)}$ corresponds to the edges of $H$ and gives a planar embedding of $H$ (please see Figure~\ref{fig:planar} for an example). Hence, $G$ is a planar graph.
\begin{figure}
	\centering
	\includegraphics[page=3]{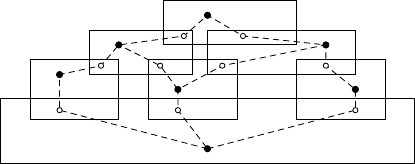}
	\caption{The dotted curves along with the solid points endpoints, give a planar embedding of the intersection graph of the rectangles in the figure. The hollow circle contained in the intersection region of two rectangles, say $r_u$ and $r_v$, represents the point $p_{uv}$.}\label{fig:planar}
\end{figure}
\end{proof}

\section{Bounds on the stab number for some graph classes}\label{sec:classes}
In this section, we study the stab number of some subclasses of rectangle intersection graphs. We show a lower bound on $stab(G)$ for any rectangle intersection graph $G$, which is used to derive an asymptotically tight lower bound for the stab number of grids. We also derive upper bounds on $estab(G)$ when $G$ is a split graph or a block graph.

\subsection{Lower bounds}
It is clear that given a $k$-stabbed rectangle intersection representation of a graph $G$, a set of $\omega(G)$ colours can be used to properly colour the vertices whose rectangles have a common stab (since the subgraph induced in $G$ by these vertices is an interval graph). This means that if $G$ is exactly stabbable, we can use two sets of $\omega(G)$ colours each to colour the vertices on alternate stab lines of a $k$-exactly stabbed representation of $G$ (for some $k$) to obtain a proper colouring of $G$. Thus, if $G$ is an exactly stabbable rectangle intersection graph, then $\chi(G)\leq 2\omega(G)$. For general rectangle intersection graphs, we can adapt the same colouring strategy to get the following observation.

\begin{observation}
For a rectangle intersection graph $G$, we have $\chi(G)\leq stab(G)\cdot \omega(G)$, or in other words, $stab(G)\geq\frac{\chi(G)}{\omega(G)}$.
\end{observation}

\noindent\textit{Remarks.} Even though for a 3-SRIG $G$, the above observation gives only $\chi(G)\leq 3\omega(G)$, we can use Theorem~\ref{thm:equiv} to infer that $G$ is actually $3$-ESRIG, and therefore $\chi(G)\leq 2\omega(G)$. Note that for any rectangle intersection graph $G$, $\chi(G)\leq 8\omega(G)^2$~\cite{asplundgrunbaum}. The question of whether there exists an upper bound on $\chi(G)$ for rectangle intersection graphs that is linear in $\omega(G)$ is open.
\medskip

We now strengthen the above observation and show that the $\chi(G)$ in the lower bound can be replaced by $pw(G)+1$, where $pw(G)$ is the ``pathwidth'' of $G$. A \emph{path decomposition} of a graph $G$ is a collection $X_1,X_2,\ldots,X_t$ of subsets of $V(G)$, where $t$ is some positive integer, such that for each edge $uv\in E(G)$, there exists $i\in\{1,2,\ldots,t\}$ such that $u,v\in X_i$ and for each vertex $u\in V(G)$, if $u\in X_i\cap X_j$, where $i<j$, then $u\in X_k$ for $i\leq k\leq j$. The \emph{width} of a path decomposition $X_1,X_2,\ldots,X_t$ of $G$ is defined to be $\max_{1\leq i\leq t}\{|X_i|\}-1$. The \emph{pathwidth} of a graph $G$, denoted by $pw(G)$, is the width of a path decomposition of $G$ of minimum width.

We adapt a proof by Suderman~\cite{suderman2004} to show that if a graph $G$ is $k$-SRIG then $G$ has pathwidth at most $k\cdot\omega(G)-1$.

\begin{theorem}\label{thm:pathwidth}
	Let $G$ be a rectangle intersection graph. Then $pw(G)\leq \omega(G)\cdot stab(G)-1$, or in other words, $stab(G)\geq\frac{pw(G)+1}{\omega(G)}$.
\end{theorem}
\begin{proof}
	Let $G$ be a rectangle intersection graph with $stab(G)=k$ . We shall show that $pw(G)\leq k\cdot\omega(G)-1$. Let $\mathcal{R}$ be a $k$-stabbed rectangle intersection representation of $G$. Let $V(G)=\{u_1,u_2,\ldots,u_n\}$ such that $x^+_{u_1}\leq x^+_{u_2}\leq\cdots\leq x^+_{u_n}$. For $i\in\{1,2,\ldots,n\}$, let us define the subset $X_i=\{v\in V(G)\colon x^+_{u_i}\in span(v)\}$. We claim that $X_1,X_2,\ldots,X_n$ is a path decomposition of $G$. To see this, note that for any edge $u_iu_j\in E(G)$, where $i<j$, $u_i,u_j\in X_i$. Also, if some vertex $v\in X_i\cap X_j$, where $i<j$, then $span(v)$ contains both $x^+_{u_i}$ and $x^+_{u_j}$, implying that it also contains $x^+_{u_k}$, for $i\leq k\leq j$. Therefore, $v\in X_k$, for $i\leq k\leq j$. To complete the proof, we only need to show that $\max_{1\leq i\leq n}\{|X_i|\}\leq k\cdot\omega(G)$. Suppose that for some $i\in\{1,2,\ldots,n\}$, there exists $S\subseteq X_i$ such that $|S|\geq \omega(G)+1$ and all the vertices of $S$ have a common stab. Since $x^+_{u_i}\in \bigcap\limits_{u\in S} span(u)$ and the rectangles corresponding to the vertices of $S$ all intersect a common stab line, we have that the vertices of $S$ form a clique in $G$, which is a contradiction to the fact that $\omega(G)$ is the clique number of $G$. Therefore, for any $i\in\{1,2,\ldots,n\}$, there exists at most $\omega(G)$ vertices in $X_i$ that have a common stab. Since there are only $k$ stab lines in $\mathcal{R}$, we now have that $|X_i|\leq k\cdot \omega(G)$ for each $i\in\{1,2,\ldots,n\}$.
\end{proof}
\medskip

The \emph{$(h,w)$-grid} is the undirected graph $G$ with $V(G)=\{(x,y)\colon x,y\in\mathbb{Z}, 1\leq x\leq h, 1\leq y\leq w\}$ and $E(G)=\{(u,v)(x,y)\colon |u-x|+|v-y|=1\}$.
 
\begin{corollary}
Let $G$ be the $(h,w)$-grid. Then $\frac{1}{2}(\min\{h,w\}+1)\leq stab(G)\leq estab(G)\leq\min\{h,w\}$.
\end{corollary}
\begin{proof}
	It is clear that $\omega(G)\leq 2 $ and from a result of~\cite{ellis2008} we know that the pathwidth of the $(h,w)$-grid is $\min\{h,w\}$. From these facts and Theorem~\ref{thm:pathwidth}, we can infer that, $\frac{1}{2}(\min\{h,w\}+1)\leq stab(G)$. It is easy to see that the $(h,w)$-grid graph has a $\min\{h,w\}$-exactly stabbed rectangle intersection representation as shown in Figure~\ref{fig:gridrep}, and therefore $estab(G)\leq \min\{h,w\}$. The statement of the corollary now follows from the fact that $stab(G)\leq estab(G)$.
\end{proof}

\begin{figure}
	\centering
	\begin{tabular}{cc}
		\includegraphics[page=8]{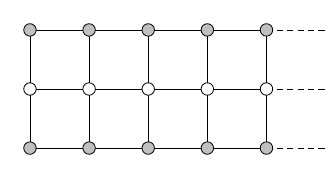}&\includegraphics[page=9]{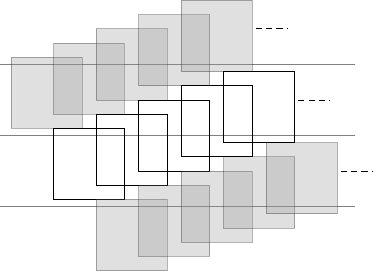}\\
		(a)&(b)
	\end{tabular}
	\caption{Illustration of $\min\{h,w\}$-exactly stabbed rectangle intersection representation of the $(h,w)$-grid: (a) The $(3,n)$-grid with $n\geq 3$; (b) a 3-exactly stabbed rectangle intersection representation of the $(3,n)$-grid. }
	\label{fig:gridrep}
\end{figure}

The above corollary shows that $stab(${\sc Grids}$,n)=\Theta(\sqrt{n})$. This also shows that there are triangle-free rectangle intersection graphs on $n$ vertices whose stab number can be $\Omega(\sqrt{n})$. Moreover, these triangle-free rectangle intersection graphs are exactly stabbable.

\subsection{Split graphs}
A split graph is a graph whose vertex set can be partitioned into a clique and an independent set.
It is known that split graphs can have arbitrarily high boxicity~\cite{cozzens1983}. So it is natural to ask whether the split graphs within rectangle intersection graphs are all exactly stabbable rectangle intersection graphs. We show that any split graph with boxicity at most 2 is 3-ESRIG (Theorem~\ref{thm:split3sig}) and that there exists a split graph with boxicity at most 2 which is not 2-ESRIG (Theorem~\ref{thm:splitnot2-SIG}). From Theorem~\ref{thm:equiv}, it then follows that the stab number and exact stab number are equal for any split graph that has boxicity at most 2. Adiga et al. showed that deciding whether a split graph has boxicity at most 3 is NP-complete~\cite{adiga2010}. But as far as we know, the problem of deciding whether the boxicity of a split graph is at most 2 is not known to be polynomial-time solvable or NP-complete. By our observations below, it follows that this problem is equivalent to deciding whether a given split graph is 3-ESRIG (or equivalently, 3-SRIG).

\begin{theorem}\label{thm:split3sig}
A split graph $G$ is a rectangle intersection graph if and only if $G$ is a 3-ESRIG.
\end{theorem}
\begin{proof}
As $G$ is a split graph, there exists a partition of $V(G)$ into sets $C$ and $I$ such that $C$ is a clique and $I$ is an independent set.
If $G$ is a 3-ESRIG then $G$ is a rectangle intersection graph. Now let $G$ be a split graph having a rectangle intersection representation $\mathcal{R}$ such that for any two vertices $u,v\in V(G)$, $\{x^-_u,x^+_u,y^-_u,y^+_u\}\cap\{x^-_v,x^+_v,y^-_v,y^+_v\}=\emptyset$ (note that such a rectangle intersection representation exists for any rectangle intersection graph). We shall assume without loss of generality that in this representation, the origin is contained in $\bigcap_{v\in C} r_v$. For every vertex $u\in I$, define the region $A_u=\bigcap_{v\in N[u]} r_v$. It is easy to see that $A_u\subseteq r_u$. It follows that for vertices $u,v\in V(G)$ such that $u\in I$ and $v\notin N[u]$, $A_u\cap r_v=\emptyset$. Also, $A_u$ is a rectangle (by the Helly property of rectangles) with non-zero height and width. This means that we can choose a point $p_u$ in $A_u$ that is not on the $X$-axis for each vertex $u\in I$, while satisfying the additional property that no two points in $\{p_u\}_{u\in I}$ have the same $x$-coordinate. Consider $u\in I$. Since the degenerate rectangle given by the point $p_u$ intersects all the rectangles in $\{r_v\}_{v\in N(u)}$, we can replace the rectangle $r_u$ with the degenerate rectangle given by the point $p_u$ to obtain a new rectangle intersection representation of $G$. Let $\mathcal{R}'$ be the rectangle intersection representation of $G$ obtained in this fasion, i.e. $\mathcal{R}'=(\mathcal{R}\setminus\{r_u\}_{u\in I})\cup\{p_u\}_{u\in I}$ (see Figure~\ref{fig:splitrep}(a)).

\begin{figure}
	\centering
	\begin{tabular}{cc}
		\includegraphics[page=4]{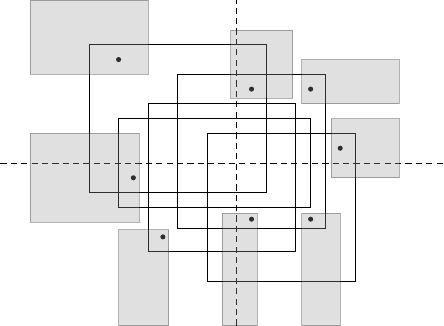}&
		\includegraphics[page=5]{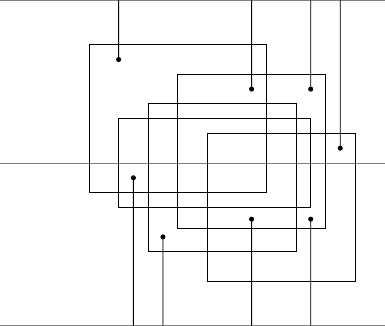}\\
		(a)&(b)
	\end{tabular}
	\caption{Representation of split graphs with boxicity at most 2. (a) The shaded rectangles represent vertices of the independent set of the split graph and the dots indicate the points $p_u$, for each vertex $u$ in the independent set. (b) The 3-ESRIG representation derived from the rectangle intersection representation given in (a).}
	\label{fig:splitrep}
\end{figure}

Let $I^+$ (respectively $I^-$) be the set of vertices $\{u\in I\colon p_u$ is above (respectively, below) the $X$-axis $\}$. Let $y_{max}=\max\{y^+_v\colon v\in C\}$ and $y_{min}=\min\{y^-_v\colon v\in C\}$. For each vertex $u\in I^+$, we define $s_u$ to be the degenerate rectangle given by the vertical line segment whose bottom end point is $p_u$ and top end point has $y$-coordinate $y_{max}+1$. Similarly, for each vertex $u\in I^-$, we define $s_u$ to be the degenerate rectangle given by the vertical line segment whose top end point is $p_u$ and bottom end point has $y$-coordinate $y_{min}-1$. As each rectangle in $\mathcal{R}'$ corresponding to a vertex in $C$ contains the origin, we have that for any $u,v\in V(G)$ such that $u\in I$ and $v\in C$, the rectangle $r_v$ intersects $s_u$ if and only if $r_v$ contains $p_u$. Therefore, the collection of rectangles given by $(\mathcal{R}'\setminus \{p_u\}_{u\in I}) \cup \{s_u\}_{u\in I}$ is a rectangle intersection representation of $G$. It is easy to see that this rectangle intersection representation, together with the horizontal lines $y=y_{min}-1$, $y=0$, and $y=y_{max}+1$, forms a 3-ESRIG representation of $G$ (see Figure~\ref{fig:splitrep}(b)).
\end{proof}

\begin{theorem}\label{thm:splitnot2-SIG}
There is a split graph $G$ which is a rectangle intersection graph but not a 2-ESRIG.
\end{theorem}
\begin{proof}
Let $G$ be the split graph whose vertex set is partitioned into a clique $C$ on four vertices and an in independent set $I$ of 14 vertices, and whose edges are defined as follows. Let $\mathcal{X}$ be the set of all subsets $X$ of $C$ with $1\leq |X|\leq 3$. For every $X\in \mathcal{X}$, there is exactly one vertex $u_X\in I$ such that $N(u_X)=X$. See Figure~\ref{fig:not2sig}(a) for a drawing of the graph $G$. Clearly, $G$ has a rectangle intersection representation as shown in Figure~\ref{fig:not2sig}(b). 

\begin{figure}
	\centering
	\begin{tabular}{cc}
		\includegraphics[page=6]{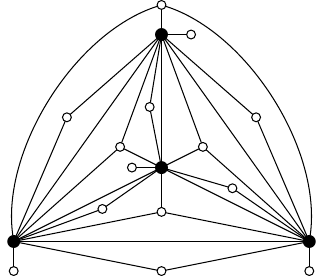}&	\includegraphics[page=7]{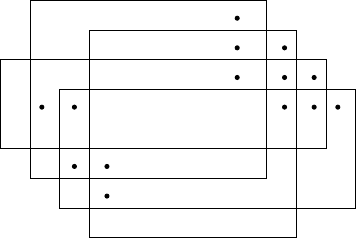}\\
		(a)&(b)
	\end{tabular}
	\caption{(a) A planar split graph which is 3-ESRIG but not 2-ESRIG. The clique vertices are coloured black and the remaining vertices are independent vertices. (b) A rectangle intersection representation of the graph shown in (a). The vertices corresponding to the independent set are represented as points.}
	\label{fig:not2sig}
\end{figure}

Now assume for the sake of contradiction that $G$ has a 2-ESRIG representation $\mathcal{R}$. We can assume that the stab lines are $y=0$ and $y=1$. We shall further assume that all the rectangles are contained in the strip of the plane between the two stab lines, i.e. for each $v\in V(G)$, we have $y^-_v\geq 0$ and $y^+_v\leq 1$ (it is easy to see that every 2-ESRIG representation can be converted to such a 2-ESRIG representation by ``trimming'' the parts of the rectangles that lie above the top stab line and below the bottom stab line).

Observe that for each $X\in\mathcal{X}$, the rectangle $r_{u_X}$ intersects all the rectangles in $\{r_v\}_{v\in X}$ and is disjoint from each rectangle in $\{r_v\}_{v\in C\setminus X}$. Now choose a point $p_X\in r_{u_X}\cap \bigcap_{v\in X} r_v$. Clearly, $p_X\in \bigcap_{v\in X} r_v$ and $p_X\notin\bigcup_{v\in C\setminus X} r_v$.

Let $a,b\in C$ (not necessarily distinct) such that $x^-_a=\max\{x^-_v\}_{v\in C}$ and $x^+_b=\min\{x^+_v\}_{v\in C}$. Let $c,d$ be two distinct vertices in $C\setminus\{a,b\}$. By our choice of $a$ and $b$, we have $[x^-_a,x^+_b]\subseteq span(c)$ and $[x^-_a,x^+_b]\subseteq span(d)$, or in other words $[x^-_a,x^+_b]\subseteq span(c)\cap span(d)$.

\begin{claim}
The vertices $c$ and $d$ have a common stab.
\end{claim}

\noindent Suppose for the sake of contradiction that $c$ and $d$ do not have a common stab. Then, since $[x^-_a,x^+_b]\subseteq span(c)\cap span(d)$ and $r_c\cap r_d\neq\emptyset$, it follows that the rectangle $[x^-_a,x^+_b]\times [0,1]\subseteq r_c\cup r_d$. We thus have $r_a\cap r_b\subseteq [x^-_a,x^+_b]\times [0,1]\subseteq r_c\cup r_d$. But this contradicts the fact that there exists a point $p_{\{a,b\}}$ such that $p_{\{a,b\}}\in r_a\cap r_b$ and $p_{\{a,b\}}\notin r_c\cup r_d$. This proves the claim.
\medskip

By the above claim, we shall assume without loss of generality that $c$ and $d$ are on the stab line $y=0$ and that $y^+_c\leq y^+_d$. This implies that $[x^-_a,x^+_b]\times [0,y^+_c]\subseteq [x^-_a,x^+_b]\times [0,y^+_d]\subseteq r_d$ (recall that $[x^-_a,x^+_b]\subseteq span(d)$). Note that $r_a\cap r_b\cap r_c\subseteq [x^-_a,x^+_b]\times [0,y^+_c]$, implying that $r_a\cap r_b\cap r_c\subseteq r_d$. But this contradicts the fact that there exists a point $p_{\{a,b,c\}}$ such that $p_{\{a,b,c\}}\in r_a\cap r_b\cap r_c$ and $p_{\{a,b,c\}}\notin r_d$.
\end{proof}

\subsection{Block graphs}\label{sec:upperbound}
A graph $G$ is a block graph if every block (i.e 2-connected component) of $G$ is a clique. Note that all trees are block graphs. It is not hard to see that all trees, and indeed all block graphs, are rectangle intersection graphs. We show that all block graphs are exactly stabbable rectangle intersection graphs and give an upper bound of $\lceil\log m\rceil$ for the exact stab number of block graphs with $m$ blocks, where $m\geq 2$. Note that this implies an upper bound of $\lceil\log n\rceil$ for the exact stab number of trees on $n$ vertices. We shall show in Section~\ref{sec:gl} that this bound is asymptotically tight, by constructing trees whose stab number is $\Omega(\log n)$.

Let $G$ be a block graph.
Given a $k$-exactly stabbed rectangle intersection representation $\mathcal{R}$ of $G$, we say that a set of vertices $S\subseteq B$, where $B$ is a block in $G$, is \emph{accessible} if all vertices in $S$ are on the bottom stab line of $\mathcal{R}$ and for any vertex $v\notin S$ either $v$ is not on the bottom stab line or $x^-_u<x^-_v$ for every vertex $u\in S$. 

\begin{theorem}\label{thm:blockstab}
For any block graph $G$ with $m$ blocks, $estab(G)\leq\max\{1,\lceil \log m\rceil\}$.
\end{theorem}
\begin{proof}
Note that we only need the statement of the theorem to be proved for connected graphs. In fact, we shall prove the following stronger claim for connected graphs.
	
\begin{claim}
Let $G$ be any connected block graph with $m$ blocks and let $k=\max\{1,\lceil\log m \rceil\}$. Then for any block $B$ of $G$, any subset $S$ of $B$, any $a,b\in\mathbb{R}$ such that $a<b$, and any $h\in\mathbb{R}$ such that $0\leq h<1$, there is a $k$-exactly stabbed rectangle intersection representation $\mathcal{R}(S,a,b,h)$ of $G$ with stab lines $y=0$, $y=1$, $y=2$, $\ldots$, $y=k-1$ such that:
\begin{itemize}
\vspace{-0.05in}
\itemsep -0.025in
\item $S$ is accessible,
\item for every vertex $u\in V(G)$, $span(u)\subseteq (a,b)$,
\item for every vertex $u\in V(G)$ that is on the bottom stab line, we have $y^+_u>h$, and
\item for every vertex $u\in V(G)$ that is not on the bottom stab line, we have $y^-_u>h$.
\end{itemize}
\end{claim}
\vspace{-0.1in}
	
\noindent \textit{Proof.} 
We prove the claim by induction on $m$. When $m\leq 2$, $G$ is an interval graph. It is not hard to see that the statement of the claim is true in this case. From here onwards, we shall assume that $m\geq 3$, and that the statement of the claim is true when the number of blocks in the graph is lesser than $m$.
	
Let $\mathcal{H}$ be the set of components of $G-B$. It is easy to see that each graph $H\in\mathcal{H}$ is a block graph and at most one of them can have greater than $\frac{m}{2}$ blocks. We shall denote the graph in $\mathcal{H}$ that has greater than $\frac{m}{2}$ blocks, if it exists, as $H^*$. For a vertex $u\in B$, let $\mathcal{H}_u=\{H\in\mathcal{H}\colon N(u)\cap V(H)\neq\emptyset\}$. Note that for $u,v\in B$ such that $u\neq v$, $\mathcal{H}_u\cap\mathcal{H}_v=\emptyset$. Also, since $G$ is connected, $\{\mathcal{H}_u\}_{u\in B}$ is a partition of $\mathcal{H}$. If $H^*$ exists, let $u^*\in B$ be the vertex such that $H^*\in\mathcal{H}_{u^*}$. 
	
Let $\mathcal{I}_B=\{[c_u,d_u]\}_{u\in B}$ be an interval representation of $G[B]$ (which is a complete graph) such that all endpoints of intervals are distinct, $[c_u,d_u]\subseteq (a,\frac{a+b}{2})$ for any $u\in V(G)$, and for any $u\in S$ and $v\in B\setminus S$, we have $c_u<c_v$. Let $B=\{u_1,u_2,\ldots,u_{|B|}\}$, where $c_{u_1}<c_{u_2}<\cdots<c_{u_{|B|}}$. We shall define $c_{u_{|B|+1}}=d_{u_{|B|}}$ (this shall be used later on). Choose $|B|$ real numbers $h_1,h_2,\ldots,h_{|B|}$ such that $h<h_1<h_2<\cdots<h_{|B|}<1$. We define $r_u$ for every vertex $u\in B$ other than $u^*$ as follows: $r_u=[c_u,d_u]\times [0,h_i]$, where $i\in\{1,2,\ldots,|B|\}$ is such that $u=u_i$. We shall show how to define $r_{u^*}$, in case $u^*$ exists, later.
	
For each $i\in\{1,2,\ldots,|B|\}$, let $\mathcal{H}'_i=\mathcal{H}_{u_i}\setminus\{H^*\}$, if $H^*$ exists, and $\mathcal{H}'_i=\mathcal{H}_{u_i}$ otherwise. Let $t_i=|\mathcal{H}'_i|$. For each $i\in\{1,2,\ldots,|B|\}$, let $\mathcal{H}'_i=\{H_{i,1},H_{i,2},\ldots,H_{i,|t_i|}\}$ and for each  $j\in\{1,2,\ldots,t_i\}$, let $S_{i,j}=N(u_i)\cap V(H_{i,j})$ (which is nonempty by the definition of $\mathcal{H}_{u_i}$). For each $i\in\{1,2,\ldots,|B|\}$, choose $t_i+1$ real numbers $c_{u_i}<q_{i,1}<q_{i,2}<\cdots<q_{i,t_i+1}<c_{u_{i+1}}$ (recall that $c_{u_{|B|+1}}=d_{u_{|B|}}$). Now consider any $i\in\{1,2,\ldots,|B|\}$ and any $j\in\{1,2,\ldots,t_i\}$. As the number of blocks in $H_{i,j}$ is at most $\frac{m}{2}$, we can apply the induction hypothesis on $H_{i,j}$ to conclude that there is a $\max\{1,\lceil\log m\rceil-1\}$-exactly stabbed rectangle intersection representation $\mathcal{R}'_{i,j}=\mathcal{R}(S_{i,j},q_{i,j},q_{i,j+1},0)$ of $H_{i,j}$. Since $m\geq 3$, we know that $k\geq 2$ and that $\max\{1,\lceil\log m\rceil-1\}=k-1$. Thus, $\mathcal{R}'_{i,j}$ uses the stab lines $y=0$, $y=1$, $\ldots$, $y=k-2$. For each vertex $v\in V(H_{i,j})$, let $r'_v=[x'^-_v,x'^+_v]\times [y'^-_v,y'^+_v]$ be the rectangle corresponding to $v$ in $\mathcal{R}'_{i,j}$. We now define $r_v$ for each vertex $v\in V(H_{i,j})$ for all $i\in\{1,2,\ldots,|B|\}$ and $j\in\{1,2,\ldots,t_i\}$ as follows. If $v\in S_{i,j}$, then $r_v=[x'^-_v,x'^+_v]\times [h_i,y'^+_v+1]$. If $v\in V(H_{i,j})\setminus S_{i,j}$ and $v$ is on the bottom stab line of $\mathcal{R}'_{i,j}$, then we define $r_v=[x'^-_v,x'^+_v]\times [1,y'^+_v+1]$. Lastly, if $v\in V(H_{i,j})\setminus S_{i,j}$, but $v$ is not on the bottom stab line of $\mathcal{R}'$, then we define $r_v=[x'^-_v,x'^+_v]\times [y'^-_v+1,y'^+_v+1]$.
	
\medskip
(*) For an integer $i\in \{1,2,\ldots,|B|\}$ and a vertex $v$ of some $H\in\mathcal{H}'_i$, we have $[x^-_v,x^+_v]\subset [c_{u_i},c_{u_{i+1}}]$.
	
\medskip
(+) For an integer $i\in \{1,2,\ldots,|B|\}$ and for any two distinct integers $j,k\in\{1,2,\ldots,t_i\}$ let $u$ be a vertex in $V(H_{i,j})$ and $v$ be a vertex in $V(H_{i,k})$. Then $r_u\cap r_v=\emptyset$ (since $[x^-_u,x^+_u]$, $[x^-_v,x^+_v]$ belong respectively to the intervals $(q_{i,j},q_{i,j+1}),(q_{i,k},q_{i,k+1})$ which are disjoint). 
	
\medskip
(++) Let $i,j$ be two distinct integers in $\{1,2,\ldots,|B|\}$. Let $u$ be a vertex in some graph in $\mathcal{H}'_i$ and $v$ be a vertex in some graph in $\mathcal{H}'_j$. Then $r_u\cap r_v=\emptyset$ (since $[x^-_u,x^+_u]$, $[x^-_v,x^+_v]$ belong respectively to the intervals $(c_{u_i},c_{u_{i+1}}),(c_{u_j},c_{u_{j+1}})$ which are disjoint). 
	
\medskip
We now define a rectangle $r_v$ for each vertex $v\in V(H^*)$ and the rectangle $r_{u^*}$ for $u^*$, in case $H^*$ exists. Let $S^*=N[u^*]\cap V(H^*)$. Since $H^*$ contains less than $m$ blocks, and recalling that $k=\max\{1,\lceil\log m\rceil\}$, we have by the induction hypothesis that $H^*$ has a $k$-exactly stabbed rectangle intersection representation $\mathcal{R}^*=\mathcal{R}(S^*,\frac{a+b}{2},b,h_{|B|})$ that uses the stab lines $y=0$, $y=1$, $\ldots$, $y=k-1$. Let the rectangle in $\mathcal{R}^*$ corresponding to a vertex $v\in V(H^*)$ be denoted by $r^*_v=[x^{*-}_v,x^{*+}_v]\times [y^{*-}_v,y^{*+}_v]$. We define $r_v=r^*_v$ for every vertex $v\in V(H^*)$.
We now let $r_{u^*}=[c_{u^*},\max\{x^{*-}_v\colon v\in S^*\}]\times [0,h_i]$, where $i\in\{1,2,\ldots,|B|\}$ is such that $u^*=u_i$.

\medskip
(+++) Let $i$ be any integer in $\{1,2,\ldots,|B|\}$. Let $u$ be a vertex of some graph in $\mathcal{H}'_i$ and $v$ be a vertex of $H^*$. Then $r_u\cap r_v=\emptyset$ (since $[x^-_u,x^+_u]$, $[x^-_v,x^+_v]$ belong respectively to the intervals $(a,\frac{a+b}{2}),(\frac{a+b}{2},b)$ which are disjoint). 
	
\medskip
We now verify that $\mathcal{R}=\{r_u\}_{u\in V(G)}$ forms a $\lceil\log m\rceil$-exactly stabbed rectangle intersection representation of $G$ that satisfies all the requirements to be $\mathcal{R}(S,a,b,h)$. For a vertex $u\in V(G)$, let $x^-_u,x^+_u,y^-_u,y^+_u$ be such that $r_u=[x^-_u,x^+_u]\times [y^-_u,y^+_u]$.
	
From the construction of $\mathcal{R}$, it is clear that all the vertices in $B$, and therefore all the vertices in $S$, are on the bottom stab line. It is also easy to see that the only vertices on the bottom stab line other than the vertices in $B$ are some vertices in $V(H^*)$. For any vertex $u\in B$ and $v\in V(H^*)$, we have $x^-_u<\frac{a+b}{2}<x^-_v$. Note that for any vertex $u\in B$, we have $x^-_u=c_u$. Therefore, for vertices $u,v\in B$ such that $u\in S$ and $v\in B\setminus S$, we have $x^-_u<x^-_v$ (recall that $c_u<c_v$ in this case). From this, we can infer that $S$ is accessible in $\mathcal{R}$.
	
It is clear that for each $u\in V(G)$, $r_u\subset (a,b)$. Now consider any vertex $v$ that is on the bottom stab line in $\mathcal{R}$. As explained before, $v$ is either in $B$ or in $V(H^*)$. If $v\in B$, then $v=u_i$ for some $i\in\{1,2,\ldots,|B|\}$, and $y^+_v=h_i>h$. On the other hand, if $v\in V(H^*)$, then $r_v=r^*_v$, the rectangle corresponding to $v$ in $\mathcal{R}^*$. Since $\mathcal{R}^*=\mathcal{R}(S^*,\frac{a+b}{2},b,h_{|B|})$, we know that $y^{*+}_v>h_{|B|}>h$, and therefore we have $y^+_v>h$. Therefore, for every vertex $v\in V(G)$ that is on the bottom stab line, we have $y^+_v>h$. Now consider a vertex $v\in V(G)$ that is not on the bottom stab line in $\mathcal{R}$. It is clear that $v\notin B$. If $v\in V(H)$, where $H\neq H^*$ and $H\in\mathcal{H}_{u_i}$, for some $i\in\{1,2,\ldots,|B|\}$, then by our construction, $y^-_v\geq h_i>h$. If $v\in V(H^*)$, then we know that since $v$ is not on the bottom stab line of $\mathcal{R}$, it is also not on the bottom stab line of $\mathcal{R}^*$. Since $\mathcal{R}^*=(S^*,\frac{a+b}{2},b,h_{|B|})$, this means that $y^{*-}_v>h_{|B|}>h$. As $y^-_v=y^{*-}_v$, we now have $y^-_v>h$. This shows that $\mathcal{R}$ satisfies the four conditions to be chosen as $\mathcal{R}(S,a,b,h)$.
	
As it can be easily verified that each rectangle in $\mathcal{R}$ is intersected by exactly one of the stab lines $y=0$, $y=1$, $\ldots$, $y=k-1$, it only remains to be shown that $\mathcal{R}$ is a rectangle intersection representation of $G$. Even though this is more or less clear from the construction, we give a proof for the sake of completeness. Consider $u,v\in V(G)$. We shall show that $uv\in E(G)$ if and only if $r_u\cap r_v\neq\emptyset$.
		
\begin{enumerate}
\vspace{-0.075in}
\itemsep 0in
\renewcommand{\theenumi}{(\roman{enumi})}
\renewcommand{\labelenumi}{\theenumi}
\item First, let us consider the case when $u,v\in V(H^*)$. Since we have $r_u=r^*_u$ and $r_v=r^*_v$, $r_u\cap r_v\neq\emptyset\Leftrightarrow r^*_u\cap r^*_v\neq\emptyset$. Since $\mathcal{R}^*$ is a valid representation of $H^*$, we have $r_u\cap r_v\neq\emptyset\Leftrightarrow uv\in E(H^*)\Leftrightarrow uv\in E(G)$.

\item Next, let us consider the case when $u\in B$ and $v\in V(H^*)$. If $u\neq u^*$ then $H^*\notin \mathcal{H}_u$ and thus $uv\notin E(G)$. Also, we have $[x^-_u,x^+_u]\subseteq (a,\frac{a+b}{2})$ (since $u\neq u^*$) and $[x^-_v,x^+_v]\subseteq (\frac{a+b}{2},b)$. Hence $r_u\cap r_v=\emptyset$. Now assume that $u=u^*=u_i$ (for some $i\in\{1,2,\ldots,|B|\}$). Recall that $S^*=N(u)\cap V(H^*)$. Suppose first that $v\in S^*$. Then $uv\in E(G)$. Now from the definition of $\mathcal{R}^*$ and $r_{u^*}=r_u$, we have that both $r_v$ and $r_u$ intersect the stab line $y=0$, $x^-_v=x^{*-}_v$ and that $x^+_u=\max\{x^{*-}_w\colon w\in S^*\}$. Combining these, we have $x^-_v\leq x^+_u$. This gives us $x^-_u<\frac{a+b}{2}<x^-_v\leq x^+_u$, implying that $r_u\cap r_v\neq\emptyset$. Now assume that $v\notin S^*$, from which it follows that $uv\notin E(G)$. If $r_v=r^*_v$ intersects the stab line $y=0$, then since $\mathcal{R}^*=\mathcal{R}(S^*,\frac{a+b}{2},b,h_{|B|})$, we have that $\max\{x^-_w\colon w\in S^*\}<x^-_v$, implying that $x^+_u<x^-_v$ (recall that $u=u^*$). Therefore, $r_u\cap r_v=\emptyset$. The only remaining case is if $r_v$ does not intersect the bottom stab line. Then, since $r_v=r^*_v$ and $\mathcal{R}^*=\mathcal{R}(S^*,\frac{a+b}{2},b,h_{|B|})$, we have $y^{*-}_v>h_{|B|}\geq h_i=y^+_u$, where $i\in\{1,2,\ldots,|B|\}$ is such that $u=u^*=u_i$. Therefore $r_u\cap r_v=\emptyset$.
		
\item Next, let $u$ be a vertex of some graph in $\mathcal{H}'_i$ for some $i\in\{1,2,\ldots,|B|\}$ and $v$ be a vertex in $H^*$. Then clearly $uv\notin E(G)$ and by (+++) we have that $r_u\cap r_v=\emptyset$.
		
\item Next, suppose that $u,v\in B$. Note that for every vertex $u\in B\setminus\{u^*\}$, we have $x^-_u=c_u$ and $x^+_u=d_u$. Since we have $x^-_{u^*}=c_{u^*}$ and $x^+_{u^*}=\max\{x^{*-}_v\colon v\in S^*\}>\frac{a+b}{2}>d_{u^*}$, we can conclude that for every vertex $u\in B$, $[c_u,d_u]\subseteq [x^-_u,x^+_u]$. As $G[B]$ is a clique, we have $uv\in E(G)$. By our construction, both $u$ and $v$ are on the bottom stab line, and since $[c_u,d_u]\cap [c_v,d_v]\neq\emptyset$, we have $[x^-_u,x^+_u]\cap [x^-_v,x^+_v]\neq\emptyset$. We thus have $r_u\cap r_v\neq\emptyset$.
		
\item Next, let us consider the case when $u\in B$ and $v$ is a vertex of some graph in $\mathcal{H}'_i$. First, let us consider the case when $u=u_i$. Let $v$ be a vertex in $H_{i,j}$ for some $j\in\{1,2,\ldots,t_i\}$. If $uv\in E(G)$, then $v\in S_{i,j}$ (recall that $S_{i,j}=N(u_i)\cap V(H_{i,j})$). In this case, we have by (*) that $[x^-_{v},x^+_v]\subset [c_{u_i},c_{u_{i+1}}]$ and thus $[x^-_{v},x^+_v]\subset [x^-_u,x^+_u]$. Furthermore, we have by construction that $y^-_v=h_i=y^+_u$, allowing us to conclude that $r_u\cap r_v\neq \emptyset$. If $uv\notin E(G)$, then $v\notin S_{i,j}$, and therefore by construction, we know that $y^-_v\geq 1$ whereas $y^+_u=h_i<1$. Therefore the two rectangles $r_u$ and $r_v$ do not intersect. Now let us consider the case when $u\neq u_i$. In this case, we have $uv\notin E(G)$. Let $u=u_j$ and assume $j<i$. Then from our construction, we have that $y^-_v\geq h_i>h_j=y^+_{u}$ and therefore $r_u\cap r_v=\emptyset$. Now assume $j>i$. Then from (*), we know that $x^+_v<c_{u_{i+1}}\leq x^-_u$ and therefore conclude that $r_u\cap r_v=\emptyset$.
		
\item Next, let $i,j$ be two distinct integers in $\{1,2,\ldots,|B|\}$. Let $u$ be a vertex of some graph in $\mathcal{H}'_i$ and $v$ be a vertex of some graph in $\mathcal{H}'_j$. Then clearly $uv\notin E(G)$ and by (++) we have that $r_u\cap r_v=\emptyset$.
		
\item Next, let $i$ be an integer in $\{1,2,\ldots,|B|\}$ and $j,k$ be two distinct integers in $\{1,2,\ldots,t_i\}$. Let $u$ be a vertex in $H_{i,j}$ and $v$ be a vertex of $H_{i,k}$. Then clearly $uv\notin E(G)$ and by (+) we have that $r_u\cap r_v=\emptyset$.
		
\item Finally, let $i$ be an integer in $\{1,2,\ldots,|B|\}$ and $j$ be an integer in $\{1,2,\ldots,t_i\}$. Let $u,v\in V(H_{i,j})$. Let $\{r'_w\}_{w\in V(H_{i,j})}=\mathcal{R}'_{i,j}$. Also, let $r'_w=[x'^-_w,x'^+_w]\times [y'^-_w,y'^+_w]$.
Then we have $[x^-_u,x^+_u]=[x'^-_u,x'^+_u]$, $[x^-_v,x^+_v]=[x'^-_v,x'^+_v]$, $y^+_u=y'^+_u+1$, $y^+_v=y'^+_v+1$, $y^-_u\in\{1,h_i,y'^-_u+1\}$, and $y^-_v\in\{1,h_i,y'^-_v+1\}$. Let us assume without loss of generality that $y^-_u\leq y^-_v$.
We now have $[y^-_u,y^+_u]\cap [y^-_v,y^+_v]=\emptyset\Leftrightarrow y^+_u<y^-_v\Leftrightarrow y'^+_u+1<y^-_v$. Recall that $y^-_v\in\{1,h_i,y'^-_v+1\}$. If $y'^+_u+1<y^-_v$ and $y^-_v\in\{1,h_i\}$, then we have $y'^+_u<0$, which is not possible (as no stab line of $\mathcal{R}'_{i,j}$ could have intersected $r'_u$). We can thus continue the derivation as $y'^+_u+1<y^-_v\Leftrightarrow y'^+_u+1<y'^-_v+1\Leftrightarrow y'^+_u<y'^-_v\Leftrightarrow [y'^-_u,y'^+_u]\cap [y'^-_v,y'^+_v]=\emptyset$. Since we have $[x^-_u,x^+_u]=[x'^-_u,x'^+_u]$ and $[x^-_v,x^+_v]=[x'^-_v,x'^+_v]$, it is clear that $[x^-_u,x^+_u]\cap [x^-_v,x^+_v]=\emptyset\Leftrightarrow [x'^-_u,x'^+_u]\cap [x'^-_v,x'^+_v]=\emptyset$. We can thus conclude that $r_u\cap r_v=\emptyset\Leftrightarrow r'_u\cap r'_v=\emptyset$. Since $\mathcal{R}'_{i,j}$ is a valid representation of $H_{i,j}$, we have $r_u\cap r_v=\emptyset\Leftrightarrow uv\notin E(H_{i,j})\Leftrightarrow uv\notin E(G)$.
\end{enumerate}
This completes the proof.
\end{proof}

\section{Asteroidal subgraphs in a graph}\label{sec:asteroidal}

In this section, we present a forbidden structure for $k$-SRIGs and $k$-ESRIGs that generalizes the ``asteroidal triples'' of Lekkerkerker and Boland~\cite{lekkerkerkerboland}. We then study the block-trees of graphs in the context of these forbidden structures, to derive some preliminary observations which shall be used in the proofs in Section~\ref{sec:treesandblocks}. First, we give some basic definitions.

We say that two subgraphs $G_1,G_2$ of a graph $G$ are \emph{neighbour-disjoint} if for any vertex $v\in V(G_1)$, $N[v]\cap V(G_2)=\emptyset$. In other words, $V(G_1)$ and $V(G_2)$ are disjoint and there is no edge between a vertex in $V(G_1)$ and a vertex in $V(G_2)$.
\medskip

Let $G=(V,E)$ be any graph. Given a vertex $v\in V(G)$, we say that a path $P$ \emph{misses} $v$, if no vertex in $P$ is a neighbour of $v$. Similarly, given a subgraph $H$ of $G$ we say that $P$ \emph{misses} $H$ if $P$ misses each vertex in $V(H)$; in other words, $P$ misses $H$ exactly when $P$ and $H$ are neighbour-disjoint.

\begin{definition}
Given a graph $G$, three vertices $a,b,c\in V(G)$ are said to form an \emph{asteroidal triple}, or AT for short, in $G$ if there exists a path between any two vertices in $\{a,b,c\}$ that misses the third.
\end{definition}

A graph is said to be \emph{AT-free} if it contains no asteroidal triple.
A graph is \emph{chordal} if it contains no induced subgraph isomorphic to a cycle on 4 or more vertices.

\begin{theorem}[\cite{lekkerkerkerboland}]\label{thm:lb}
A graph $G$ is an interval graph if and only if $G$ is chordal and AT-free.
\end{theorem}

\begin{definition}
Three connected induced subgraphs $G_1,G_2,G_3$ of $G$ that are pairwise neighbour-disjoint are said to be \emph{asteroidal} in $G$ if for each $i\in\{1,2,3\}$, for any $i,j,k$ such that $\{i,j,k\}=\{1,2,3\}$, there is a path from some vertex of $G_i$ to some vertex of $G_j$ that misses $G_k$.
\end{definition}

Suppose $G_1,G_2,G_3$ are asteroidal in a graph $G$. Then from the above definition, they are pairwise neighbour-disjoint and each of them is connected. This implies that for any $i,j,k$ such that $\{i,j,k\}=\{1,2,3\}$, and for any $u\in V(G_i)$ and any $v\in V(G_j)$, there is some path between $u$ and $v$ that misses $G_k$.

\begin{definition}
Let $\mathcal{C}$ be a class of graphs and let $G$ be any graph. Let $G_1,G_2,G_3$ be asteroidal in $G$ and let $G_i\in\mathcal{C}$ for $i\in\{1,2,3\}$. Then we say that $G_1,G_2,G_3$ are \emph{asteroidal-$\mathcal{C}$} in $G$.
\end{definition}

\begin{definition}
We say that a graph $G$ is \emph{asteroidal-$\mathcal{C}$-free} if there does not exist three  subgraphs that are asteroidal-$\mathcal{C}$ in $G$.
\end{definition}

\subsection{A forbidden structure for $k$-SRIGs and $k$-ESRIGs}\label{sec:forbidden}
We now show that no $k$-SRIG can contain three subgraphs that are asteroidal-(non-$(k-1)$-SRIG) in it. The same technique can be used to show that a $k$-ESRIG cannot contain three subgraphs that are asteroidal-(non-$(k-1)$-ESRIG) in it. The intuition is that if a $k$-SRIG $G$ contains subgraphs $G_1,G_2,G_3$ which are asteroidal-(non-$(k-1)$-SRIG) in $G$, then in any $k$-stabbed rectangle intersection representation of $G$, the rectangles corresponding to vertices in $G_i$, for each $i\in\{1,2,3\}$, together occupy all the stab lines (as each $G_i$ is a non-$(k-1)$-SRIG). Coupled with the fact that the three subgraphs are pairwise neighbour-disjoint, this enforces a kind of ``left-to-right'' order on the subgraphs: that is, in the $k$-SRIG representation, for distinct $i,j\in\{1,2,3\}$, the collection of rectangles corresponding to vertices of $G_i$ can be thought of as being ``to the left of'' or ``to the right of'' the collection of rectangles corresponding to the vertices of $G_j$. If we take this left-to-right order of subgraphs to be $G_1,G_2,G_3$, then it can be shown that any path from a vertex of $G_1$ to a vertex of $G_3$ must contain a vertex whose rectangle intersects a rectangle belonging to a vertex of $G_2$, thus contradicting the fact that $G_1,G_2,G_3$ are asteroidal in $G$. We give the formal proof below.

\begin{theorem}\label{thm:asteroidal_k-1_sig_free}
$k$-SRIGs are asteroidal-(non-$(k-1)$-SRIG)-free.
\end{theorem}
\begin{proof}
Assume for the sake of contradiction that $G$ is a $k$-SRIG with a $k$-stabbed rectangle intersection representation $\mathcal{R}$ and has three connected induced non-$(k-1)$-SRIG subgraphs $G_1,G_2,G_3$ that are asteroidal in $G$. As each of $G_1,G_2,G_3$ are non-$(k-1)$-SRIGs, but are $k$-SRIGs (as they are induced subgraphs of $G$), for each $i\in\{1,2,3\}$, there exists a walk $W_i$ in $G_i$ such that $W_i$ contains at least one vertex on each stab line of $\mathcal{R}$ (for example, $W_i$ can be chosen to be any path in $G_i$ between a vertex on the top stab line and a vertex on the bottom stab line). This further implies that for each $i\in\{1,2,3\}$, there exists a vertex $v_i$ in $W_i$ that is on the bottom stab line. As $G_1,G_2,G_3$ are pairwise neighbour-disjoint, we know that $span(v_1)$, $span(v_2),span(v_3)$ are pairwise disjoint. Therefore we can assume without loss of generality that $span(v_1) < span(v_2) < span(v_3)$. Now consider the set of vertices $S=\{w\colon w\in N[w']$ for some $w'\in W_2\}$.	
	
Consider the region $X$ of the plane defined by $X=\bigcup_{u\in W_2} r_u$. Since $W_2$ is connected and has a vertex on each stab line, $X$ is an arc-connected region that intersects all the stab lines. Clearly, for any vertex $x$ such that $r_x\cap X\neq \emptyset$ we can conclude that $x\in S$. Now let $B$ be the rectangle with diagonally opposite corners $(x_1,y_1)$ and $(x_2,y_2)$ where $x_1=\min\{x^-_v\colon v\in V(G)\}$, $x_2=\max\{x^+_v\colon v\in V(G)\}$, $y=y_1$ is the bottom stab line and $y=y_2$ is the top stab line of $\mathcal{R}$.
	
\begin{claim}
The rectangles $B\cap r_{v_1}$ and $B\cap r_{v_3}$ are completely contained in different arc-connected regions of $B\setminus X$.
\end{claim}
	
\noindent
Since $v_1$ and $v_3$ have no neighbours in $W_2$, and therefore are not in $S$, we can infer from our earlier observation that the rectangles $B\cap r_{v_1}$ and $B\cap r_{v_3}$ are disjoint from $X$. This means that each of these rectangles are completely contained in some arc-connected region of $B\setminus X$.
Assume for the sake of contradiction that the rectangles $B\cap r_{v_1}$ and $B\cap r_{v_3}$ are completely contained in the same connected region of $B\setminus X$. This implies that there exists a curve $\mathbf{s}$ in $B\setminus X$ that connects some point in $B\cap r_{v_1}$ that is on the bottom stab line to some point in $B\cap r_{v_3}$ that is also on the bottom stab line. Now consider the points $p,q\in X$ such that $p$ is on the top stab line and $q$ is a point in $r_{v_2}$ that is on the bottom stab line. Since $X$ is connected, there is a curve $\mathbf{s'}$ in $X$ that connects $p,q$. Since $span(v_1) < span(v_2) < span(v_3)$ and $\mathbf{s},\mathbf{s'}$ are curves that are completely contained in $B$, we can conclude that the curves $\mathbf{s}$ and $\mathbf{s'}$ intersect. But this is a contradiction, as $\mathbf{s}$ is a curve in $B\setminus X$ and hence cannot contain any point in $\mathbf{s'}\subseteq X$. This completes the proof of the claim.
\medskip

As $G_1,G_2,G_3$ are asteroidal in $G$, there is a path $P$ between $v_1$ and $v_3$ that misses $G_2$. This means that the path $P$ does not contain any vertex from $S$, and therefore the rectangle corresponding to no vertex in $P$ intersects $X$. Since every rectangle in the representation intersects $B$, this means that $\bigcup_{w\in V(P)} B\cap r_w$ is an arc-connected set in $B\setminus X$ that contains both $B\cap r_{v_1}$ and $B\cap r_{v_3}$. This is a contradiction to the above claim.
\end{proof}

The following theorem can be proved using the similar arguments, and hence we omit the proof.

\begin{theorem}\label{thm:asteroidal_k-1_esrig_free}
$k$-ESRIGs are asteroidal-(non-$(k-1)$-ESRIG)-free.
\end{theorem}

\subsection{The coloured block-tree of a graph}\label{sec:colorblocktree}

A \emph{hereditary} class of graphs is a class of graphs that is closed under taking induced subgraphs.
A class of graphs is said to be \emph{closed under vertex addition} if adding a vertex (and an arbitrary set of edges incident on it) to any graph in the class results in another graph that is in the class. It can be seen that a class of graphs is closed under vertex addition if and only if its complement class (the set of graphs that are not in the class) is hereditary. Therefore, the class of non-$k$-SRIGs and the class of non-$k$-ESRIGs, for any positive integer $k$, are both closed under vertex addition.
In this section, we study the block-tree (defined below) of an asteroidal-$\mathcal{C}$-free graph, where $\mathcal{C}$ is some graph class that is closed under vertex addition. The lemmas derived in this section will be useful in the next section.

For any graph $G$, let $\mathcal{B}(G)$ be the set of blocks in it and $C(G)$ the set of cut-vertices in it. The \emph{block-tree} of $G$ (denoted as $T_G$) is the graph with $V(T_G)=\mathcal{B}(G)\cup C(G)$ and $E(T_G)=\{Bc\colon B\in\mathcal{B}(G)$, $c\in C(G)$, and $c\in B\}$. For any graph $G$, the graph $T_G$ turns out to be a tree, justifying the name ``block-tree of $G$''~\cite{diestel}.

For $e=Bc\in E(T_G)$, where $c\in C(G)$ and $B\in\mathcal{B}(G)$, we denote by $T_G(e)$ the connected component of $T_G-e$ containing $B$. Also, let us define
$$G_e=G[\bigcup_{B\in T_G(e)} B\setminus \{c\}]$$
In other words, $G_e$ is the component of $G-\{c\}$ that contains the vertices of $B$ other than $c$. Note that $G_e$ is a connected induced subgraph of $G$.
The following observation is a direct consequence of the structure of the block-tree.

\begin{observation}\label{obs:blocktree}
	The vertices of $G$ other than $c$ that belong to blocks not in $T_G(e)$ are neither in $G_e$ nor are adjacent to any vertex in $G_e$.
\end{observation}

Let $\mathcal{C}$ be a class of graphs.
Let us now colour red those edges $e$ of $T_G$ such that $G_e\in\mathcal{C}$.
Further, let us colour red those cut-vertices in $T_G$ that have at least two red edges incident on them. Note that if two red edges $e_1$ and $e_2$ are incident on a cut-vertex $u$ in $T_G$, then $G_{e_1}$ and $G_{e_2}$ are two components of $G-\{u\}$.
As the final step of colouring, we colour red those block-vertices of $T_G$ that are adjacent to at least two cut-vertices that are red. We now say that the tree $T_G$ \emph{is coloured with respect to $\mathcal{C}$}.

\begin{lemma}\label{lem:redconnected}
	Let $\mathcal{C}$ be a class of graphs that is closed under vertex addition. Let $G$ be any graph and let $T_G$ be coloured with respect to $\mathcal{C}$. Then the subgraph of $T_G$ induced by the set of red vertices is connected.
\end{lemma}
\begin{proof}
	We only need to prove that for any $u,v\in V(T_G)$ that are coloured red, any vertex $w\in V(T_G)$ that lies on the path in $T_G$ between $u$ and $v$ is also red.
	Let $P$ be the path between $u$ and $v$ in $T_G$.
	If $u$ is a cut-vertex, then let $u'=u$ and if $u$ is a block-vertex, then let $u'$ be a red cut-vertex that is adjacent to $u$ but is not on $P$. Similarly, if $v$ is a cut-vertex, then we let $v'=v$ and if $v$ is a block-vertex, we let $v'$ be a red cut-vertex that is adjacent to $v$ but is not on $P$. Clearly, the path $P'$ in $T_G$ between $u'$ and $v'$ also contains $w$. It can be seen that there is a red edge $e_u$ that is incident on $u'$ but does not belong to $P'$ and a red edge $e_v$ that is incident on $v'$ but does not belong to $P'$. As $e_u$ and $e_v$ are red edges, we know that $G_{e_u},G_{e_v}\in\mathcal{C}$. Now consider any edge $e$ that is in $P'$. From the structure of the block-tree, it follows that either $V(G_{e_u})\subseteq V(G_e)$ or $V(G_{e_v})\subseteq V(G_e)$. (To see this, let $z$ be the cut-vertex in $e$ and assume that $u'$ is closer to $z$ than $v'$ in $T_G$. Then, $T_G(e_v)$ is a subtree of $T_G(e)$. Note that $z$ is not adjacent to any block-vertex of $T_G(e_v)$, implying that $z$ is not contained in any block that appears as a block-vertex in $T_G(e_v)$. We now have that $V(G_{e_v})\subseteq V(G_e)$.) Since $\mathcal{C}$ is closed under vertex addition, we now have that $G_e\in\mathcal{C}$, which implies that $e$ is red. Therefore, every edge in $P'$ is red. It now follows that every cut-vertex in $P'$ other than $u'$ and $v'$ are incident with at least two red edges. Therefore every cut-vertex in $P'$ is red (recall that $u'$ and $v'$ are red by definition). This tells us that every block-vertex in $P'$ is adjacent to two red cut-vertices, and is therefore red. This proves that $w$ is red.
\end{proof}

\begin{lemma}\label{lem:removespecialblocks}
	Let $G$ be a graph and $\mathcal{C}$ a class of graphs closed under vertex addition. Let $T_G$ be coloured with respect to $\mathcal{C}$ and let $\mathcal{B}$ be the set of block-vertices of $T_G$ that have at least one red neighbour (or equivalently, the blocks of $G$ that contain at least one cut-vertex that is red in $T_G$). Furthermore, assume that $T_G$ has at least one red vertex. Let $H$ be any component of $G-\bigcup\limits_{B\in\mathcal{B}} B$. Then:
	\begin{enumerate}
		\vspace{-0.075in}
		\itemsep 0in
		\renewcommand{\theenumi}{(\alph{enumi})}
		\renewcommand{\labelenumi}{\theenumi}
		\item\label{it:oneneighbour} there exists exactly one vertex $u\in V(G)\setminus V(H)$ such that $N(u)\cap H\neq\emptyset$, and
		\item\label{it:notinC} $H\notin\mathcal{C}$.
	\end{enumerate}
\end{lemma}
\begin{proof}
	Let us mark the block-vertices in $T_G$ corresponding to blocks of $G$ that contain at least one vertex of $H$ and also mark the cut-vertices in $T_G$ corresponding to cut-vertices of $G$ that are in $H$. Clearly, the block-vertices that are marked are not in $\mathcal{B}$. Since $H$ is connected, it follows from the structure of the block-tree that the marked vertices of $T_G$ form a subtree of $T_G$ whose leaves are all marked block-vertices. Further, it is clear that any unmarked cut-vertex that is adjacent to a marked block-vertex belongs to some block in $\mathcal{B}$ (otherwise, that cut-vertex would have been in $H$ and therefore marked). Now suppose there exist two distinct edges $e=uX$ and $e'=u'X'$ of $T_G$ where $X,X'$ are marked block-vertices and $u,u'$ are unmarked cut-vertices. Let $B,B'$ be the blocks in $\mathcal{B}$ that contain $u,u'$ respectively. As $B,B'\in\mathcal{B}$, there exist red cut-vertices $v,v'$ adjacent to $B,B'$ respectively where $u\neq v$ and $u'\neq v'$. From Lemma~\ref{lem:redconnected}, we know that the red vertices in $T_G$ induce a connected subtree of $T_G$. Therefore, every vertex in the path in $T_G$ between $v$ and $v'$ has to be red. This implies that $u$ is red, which further implies that $X\in\mathcal{B}$. But this contradicts the fact that $X$ is a marked block-vertex. We can therefore conclude that there exist at most one marked block-vertex $X$ that has an unmarked neighbour in $T_G$. Since $T_G$ contains at least one marked vertex and at least one unmarked vertex (as $V(H)\neq\emptyset$ and $\mathcal{B}\neq\emptyset$), we have that there is exactly one marked block-vertex $X$ such that it has an unmarked neighbour $u$ in $T_G$. It now follows from the structure of the block-tree that $H=G_{uX}$. This implies that no vertex in $H$ can have a neighbour in $V(G)\setminus V(H)$ other than $u$. This proves~\ref{it:oneneighbour}.
	
	We shall now prove~\ref{it:notinC}. Suppose for the sake of contradiction that $H\in\mathcal{C}$, or in other words, $G_{uX}\in\mathcal{C}$. So, the edge $uX$ is red in $T_G$.
	
	\begin{claim}
		The cut-vertex $u$ of $T_G$ is red.
	\end{claim}
	
	\noindent As observed earlier, $u$ is in some block that is in $\mathcal{B}$. Let $B\in\mathcal{B}$ be a block containing $u$. So $uB$ is an edge of $T_G$. Since $B\in\mathcal{B}$, there must be some red cut-vertex $u'$ in $T_G$ that is adjacent to $B$. Clearly, $u'\neq u$, as otherwise, $X$ would have been adjacent to a red cut-vertex, and hence it would have been in $\mathcal{B}$. But this cannot happen as $X$ contains vertices from $H$. Since $u'$ is a red cut-vertex, it has at least two red edges incident on it and therefore there is a red edge $e$ incident on $u'$ that is different from $u'B$. From the definition of red edges, we have that $G_e\in\mathcal{C}$. It follows from the structure of the block-tree that $G_e$ is an induced subgraph of $G_{uB}$. As $\mathcal{C}$ is closed under vertex addition, we have that $G_{uB}\in\mathcal{C}$, implying that the edge $uB$ is red in $T_G$. We now have two red edges, $uX$ and $uB$, incident on $u$, which means that $u$ is a red cut-vertex of $T_G$.
	\medskip
	
	From the above claim, it follows that $X$ is a block-vertex of $T_G$ that is incident to a red cut-vertex $u$, and hence it is in $\mathcal{B}$. But this is a contradiction as $B$ contains vertices of $H$. This proves~\ref{it:notinC}.
\end{proof}

\begin{lemma}\label{lem:redpath}
	Let $\mathcal{C}$ be a class of graphs that is closed under vertex addition. Let $G$ be an asteroidal-$\mathcal{C}$-free graph and let $T_G$ be coloured with respect to $\mathcal{C}$. Then the subgraph $T_r$ of $T_G$ induced by the set of red vertices is either empty or is a path.
\end{lemma}

\begin{proof}
	If there are no red vertices in $T_G$, then there is nothing to prove. So let us suppose that $T_r$ is not empty. From Lemma~\ref{lem:redconnected}, it follows that $T_r$ is connected.
	It only remains to be shown that every vertex has degree at most two in $T_r$. Suppose for the sake of contradiction that $u$ is a red vertex that has three red neighbours $u_1,u_2,u_3$.
	
	Let us first consider the case when $u$ is a block-vertex. Then, clearly $u_1,u_2,u_3$ are all cut-vertices. Since each $u_i$, for $i\in\{1,2,3\}$ is red, we know that there are two red edges incident on each of them. This means that for each $i\in\{1,2,3\}$ there is a red edge $e_i$ different from $uu_i$ that is incident on $u_i$. It is clear from Observation~\ref{obs:blocktree} that $G_{e_1},G_{e_2},G_{e_3}$ are pairwise neighbour-disjoint connected induced subgraphs of $G$. Because $e_1,e_2,e_3$ are red, we know that $G_{e_1},G_{e_2},G_{e_3}\in\mathcal{C}$. For each $i\in\{1,2,3\}$, let $v_i$ be a neighbour of $u_i$ in $G_{e_i}$. Let the block-vertex $u$ in $T_G$ correspond to a block $B$ in $G$. From the definition of the block-tree, we know that $u_1,u_2,u_3\in B$. Since $B$ is a 2-connected subgraph of $G$, for any $i,j,k$ such that $\{i,j,k\}=\{1,2,3\}$, there exists a path $P_{ij}$ in $B$ between $u_i$ and $u_j$ that does not contain $u_k$. Let $P'_{ij}=P_{ij}\cup\{u_iv_i,u_jv_j\}$. From Observation~\ref{obs:blocktree}, it follows that $P'_{ij}$ misses $G_{e_k}$. This means that $G_{e_1},G_{e_2},G_{e_3}$ are asteroidal-$\mathcal{C}$ in $G$, contradicting the fact that $G$ is asteroidal-$\mathcal{C}$-free.
	
	Next, let us consider the case when $u$ is a cut-vertex. Then, $u_1,u_2,u_3$ are block-vertices that are coloured red. Since each of them have to be adjacent to at least two red cut-vertices, we know that for each $i\in\{1,2,3\}$, there is a red cut-vertex $u'_i$ different from $u$ that is adjacent to $u_i$. Then again, as for each $i\in\{1,2,3\}$, $u'_i$ is red, we can infer that there is a red edge $e_i$ different from $u'_iu_i$ that is incident on $u'_i$. As before, $G_{e_1},G_{e_2},G_{e_3}$ form neighbour-disjoint connected induced subgraphs of $G$ that all belong to $\mathcal{C}$. For each $i\in\{1,2,3\}$, let $v_i$ be a neighbour of $u'_i$ in $G_{e_i}$. It is now clear from the structure of the block-tree that for any $i,j,k$ such that $\{i,j,k\}=\{1,2,3\}$, there is a path $P_{ij}$ in $G$ between $u'_i$ and $u'_j$ that does not contain $u'_k$. We can now infer using Observation~\ref{obs:blocktree} that the path $P'_{ij}=P_{ij}\cup\{v_iu'_i,v_ju'_j\}$ misses $G_{e_k}$. So we again have that $G_{e_1},G_{e_2},G_{e_3}$ are asteroidal-$\mathcal{C}$ in $G$, contradicting the fact that $G$ is asteroidal-$\mathcal{C}$-free.
\end{proof}

\begin{lemma}\label{lem:noredvertex}
	Let $\mathcal{C}$ be a class of graphs that is closed under vertex addition. Let $G$ be a graph and let $T_G$ be coloured with respect to $\mathcal{C}$. If there are no red vertices in $T_G$, then there exists a block $B$ in $G$ such that no component of $G-B$ is in $\mathcal{C}$.
\end{lemma}
\begin{proof}
	Note that if there exists a cut-vertex $u$ in $G$ such that each component of $G-\{u\}$ is not in $\mathcal{C}$, then clearly, removal of any block that contains $u$ from $G$ will result in a graph whose components are not in $\mathcal{C}$ (recall that $\mathcal{C}$ is closed under vertex addition). Therefore, we shall assume that for any cut-vertex $u$ of $G$, there is some component of $G-\{u\}$ that is in $\mathcal{C}$. Since $\{G_e\colon e$ incident on $u\}$ are the components of $G-\{u\}$, this implies that in $T_G$, every cut-vertex has at least one edge $e$ incident on it such that $G_e\in\mathcal{C}$. In other words, every cut-vertex of $T_G$ has at least one red edge incident on it. Since $T_G$ contains no red vertices, we can now conclude that every cut-vertex in $T_G$ has exactly one red edge incident on it.
	
	For a cut-vertex $u$ in $G$, let us define $f(u)$ to be the only red edge incident on $u$ in $T_G$. Let $v$ be the cut-vertex in $G$ that minimizes $|V(G_{f(v)})|$. Let $f(v)=vB$, where $B$ is a block-vertex of $T_G$. Recall that in $T_G$, every edge incident on $v$ other than $vB$ is a non-red edge. In other words, none of the components of $G-v$ other than $G_{f(v)}$ belong to $\mathcal{C}$. We now claim that every edge in $T_G$ incident on $B$ is red. Suppose that there is a non-red edge $wB$ in $T_G$. Since $w$ is a cut-vertex, there is a red edge $f(w)$ incident on $w$. Since $wB$ is non-red, $f(w)$ is different from $wB$. From the structure of the block-tree, it is evident that $V(G_{f(w)})\subset V(G_{f(v)})$ ($w\in V(G_{f(v)}\setminus V(G_{f(w)})$). But this contradicts our choice of $v$ as we now have $(|V(G_{f(w)})|<|V(G_{f(v)})|$. Therefore, every edge that is incident on $B$ in $T_G$ is red.
	
	For any block-vertex $X$ in $T_G$, we shall define $F_X=\{wY\colon wX\in E(T_G)$ and $X\neq Y\}$. In other words, $F_X$ consists of exactly those edges of $T_G$ that are not incident on $X$ but are incident on some cut-vertex adjacent to $X$. Note that $\{G_e\colon e\in F_X\}$ are exactly the components of $G-X$. Since for the block-vertex $B$ under consideration, we know that every edge incident on it is red, we can infer that every edge in $F_B$ is non-red (as every cut-vertex has exactly one red edge incident on it). This means that each of $\{G_e\colon e\in F_B\}$ is a graph that is not in $\mathcal{C}$; in other words, no component of $G-B$ belongs to $\mathcal{C}$. We have thus found the required block.
\end{proof}

\section{Trees and block graphs}\label{sec:treesandblocks}

A question asked in Babu et al.~\cite{babu2014} is whether it can be determined in polynomial-time if an input tree has a rectangle intersection representation in which each rectangle is a square of unit height and width. Instead of restricting the rectangles to be unit squares, we study a different restriction. In particular, we ask if, given a tree and an integer $k$, it can be determined in polynomial-time whether the tree has a $k$-SRIG or $k$-ESRIG representation. We show that the problem is polynomial-time solvable if $k\leq 3$. In fact, we show that we can determine in polynomial-time if the input graph $G$ is 2-ESRIG (equivalently 2-SRIG, by Theorem~\ref{thm:equiv}) if $G$ is guaranteed to be a block graph. We also show that it can be determined in polynomial-time if an input tree is 3-ESRIG (equivalently 3-SRIG, by Theorem~\ref{thm:equiv}). Our algorithms depend on a forbidden structure characterization for block graphs that are 2-ESRIG and trees that are 3-ESRIG. In fact, in both cases, the algorithm is a search for the presence of these forbidden structures in the input graph, and therefore it is a ``certifying algorithm'', in the sense that the algorithm outputs a representation whenever the answer is ``Yes'' and a forbidden structure in the graph whenever the answer is ``No''.

The forbidden structure characterizations of block graphs that are 2-ESRIG and trees that are 3-ESRIG are obtained as follows. In the previous section, we showed that a necessary condition for a graph to be a 2-ESRIG is that it has to be asteroidal-(non-interval)-free. We show in this section that for block graphs, this necessary condition is also sufficient. We later on show that for trees that are 3-ESRIG, the necessary condition of being asteroidal-(non-2-ESRIG)-free is again a sufficient condition.
First, we need the following lemma.

\begin{lemma}\label{lem:gs}
Let $\mathcal{C}$ be a class of graphs that is closed under vertex addition. Let $G$ be a block graph that is asteroidal-$\mathcal{C}$-free and let $T_G$ be coloured with respect to $\mathcal{C}$. Then there exists a set $S\subseteq V(G)$ such that $G[S]$ is an interval graph and no component of $G-S$ is in $\mathcal{C}$.
\end{lemma}
\begin{proof}
When $T_G$ contains at least one red vertex, let $\mathcal{B}$ be the set of block-vertices of $T_G$ that have at least one red neighbour. If $T_G$ contains no red vertices, then by Lemma~\ref{lem:noredvertex}, there is a block $B$ in $G$ whose removal gives us components, none of which are in $\mathcal{C}$. In this case, let $\mathcal{B}=\{B\}$. We shall let $S$ be the set of vertices which are contained in some block in $\mathcal{B}$, or in other words, $S=\bigcup_{B\in\mathcal{B}} B$. By the above observation and Lemma~\ref{lem:removespecialblocks}, we can assume from here onwards that no component of $G-S$ is in $\mathcal{C}$. If there are no red vertices in $T_G$, then $G[S]$ is a complete graph, and therefore an interval graph. To complete the proof, we only need to show that if $T_G$ contains at least one red vertex, then $G[S]$ is an interval graph.

Suppose that $T_G$ contains at least one red vertex. Then from Lemma~\ref{lem:redpath}, we know that the red vertices in $T_G$ form a path. Since block graphs are chordal, by Theorem~\ref{thm:lb}, we need to only show that $G[S]$ is AT-free in order to prove that $G[S]$ is an interval graph. Suppose for the sake of contradiction that there exists an asteroidal triple $\{a,b,c\}\subseteq S$ in $G[S]$. Since $\{a,b,c\}$ has to be an independent set in $G$, we know that there is no block that contains any two of them. We shall say that a cut-vertex in $G$ is red if that cut-vertex is coloured red in $T_G$. Note that from the definition of $S$, every vertex in $S$ is adjacent to at least one red cut-vertex (since each vertex of $S$ is in some block that also contains a cut-vertex that is coloured red in $T_G$, and each block is a complete graph). Let $a',b',c'$ denote red cut-vertices that are adjacent to $a,b,c$ respectively. Suppose that $a'=b'=x$. Then, it is clear from the structure of the block-tree that either every path between $a$ and $c$ contains $x$ or every path between $b$ and $c$ contains $x$. But this contradicts the fact that $a,b,c$ form an AT in $G[S]$, since $x$ is a neighbour of both $a$ and $b$. We can therefore assume that $a',b',c'$ are distinct red cut-vertices. Since the red vertices form a path in $T_G$, the vertices $a',b',c'$ must lie on a path in $T_G$. Let us assume without loss of generality that $b'$ lies on the path in $T_G$ between $a'$ and $c'$. This means that every path between $a'$ and $c'$ in $G[S]$ contains $b'$. We now claim that every path in $G[S]$ between $a$ and $c$ goes through $b'$. Suppose for the sake of contradiction that there exists a path $P$ between $a$ and $c$ in $G[S]$ that does not contain $b'$. Then the path $a'a\cup P\cup cc'$ is a path between $a'$ and $c'$ in $G[S]$ that does not contain $b'$, contradicting the fact that every path in $G[S]$ between $a'$ and $c'$ contains $b'$. So, we have that every path between $a$ and $c$ in $G[S]$ contains $b'$, which is a neighbour of $b$. This contradicts the fact that $a,b,c$ forms an AT in $G[S]$.
\end{proof}

\begin{theorem}\label{thm:block2sig}
A block graph $G$ is 2-ESRIG if and only if $G$ is asteroidal-(non-interval)-free.
\end{theorem}
\begin{proof}
Let $G$ be a block graph.
We know by Theorem~\ref{thm:asteroidal_k-1_esrig_free} that if $G$ is a 2-ESRIG then $G$ is asteroidal-(non-interval)-free.
Now we prove that if $G$ is asteroidal-(non-interval)-free then $G$ is a 2-ESRIG.
	
By letting $\mathcal{C}$ be the class of non-interval graphs, we have by Lemma~\ref{lem:gs} that there exists a set $S\subseteq V(G)$ such that $G[S]$ is an interval graph and each component of $G-S$ is also an interval graph.

Let $\mathcal{R}=\{[c_u,d_u]\}_{u\in S}$ be an interval representation of $G[S]$ such that all endpoints of intervals are distinct. Let $\epsilon\in\mathbb{R}^+$ be such that $\epsilon<\min\{|c_u-c_v|\colon u,v\in S, u\neq v\}$. Also, let $L,R\in\mathbb{R}$ such that $L<\min_{u\in S} c_u$ and $R>\max_{u\in S} d_u$. For each vertex $u\in S$, define $t_u=\frac{c_u-L}{R-L}$.
Let $\mathcal{H}$ be the set of components of $G-S$.
For a vertex $u\in S$, let $\mathcal{H}_u=\{H\in\mathcal{H}\colon N(u)\cap H\neq\emptyset\}$. From Lemma~\ref{lem:removespecialblocks}\ref{it:oneneighbour}, it is clear that for each component $H\in\mathcal{H}$, there is a exactly one vertex in $S$ that has neighbours in $H$. Therefore, it follows that $\{\mathcal{H}_u\}_{u\in S}$ is a partition of $\mathcal{H}$ (recall that $G$ is connected). Since each component of $\mathcal{H}$ is an interval graph, and because disjoint unions of interval graphs are again interval graphs, we know that for $u\in S$, the graph $I_u$ formed by the disjoint union of the components in $\mathcal{H}_u$ is an interval graph. It is easy to see that $\{I_u\}_{u\in S}$ is a collection of neighbour-disjoint interval graphs. For each $u\in S$, let $\mathcal{R}_u$ be an interval representation $\{[c'_v,d'_v]\}_{v\in V(I_u)}$ for the interval graph $I_u$ such that every interval in it is contained in the interval $[c_u,c_u+\epsilon]$. Note that for distinct $a,b\in S$, no interval of $\mathcal{R}_a$ intersects with any interval of $\mathcal{R}_b$. Also let $b'_v=1$ if $v\notin N(u)$ and $b'_v=t_u$ if $v\in N(u)$. From here onwards, we shall assume that for every vertex $v\in V(G)\setminus S$, the interval $[c'_v,d'_v]$ and the value $b'_v$ are defined. 
	
We shall now define a rectangle $r_u=[x^-_u,x^+_u]\times [y^-_u,y^+_u]$ for each vertex $u\in V(G)$. For a vertex $u\in S$, we let $x^-_u=c_u$, $x^+_u=d_u$, $y^-_u=0$ and $y^+_u=t_u$. For a vertex $u\in V(G)\setminus S$, we let $x^-_u=c'_u$, $x^+_u=d'_u$, $y^-_u=b'_u$ and $y^+_u=1$. 
We leave it to the reader to verify that the rectangles $\{r_u\}_{u\in V(G)}$ form a 2-exactly stabbed rectangle intersection representation of $G$.
\end{proof}

\noindent\textit{Remarks.} Let $\mathcal{C}$ be the class of non-interval graphs and $G$ be a block graph with $n$ vertices and $m$ edges. Since checking whether $G$ is in $\mathcal{C}$ or not is possible in $O(n+m)$ time~\cite{corneil2009}, we can infer that coloring the edges of $T_G$ with respect to $\mathcal{C}$ is possible in $O(n^2+nm)$ time. The construction procedure described in the above proof can also be performed in $O(n^2+nm)$ time, thus giving a polynomial time algorithm to recognize block graphs that are 2-ESRIG. 

\begin{theorem}\label{thm:tree3sig}
A tree $G$ is 3-ESRIG if and only if $G$ is asteroidal-(non-2-ESRIG)-free.
\end{theorem}
\begin{proof}
Let $G$ be a tree.
We know by Theorem~\ref{thm:asteroidal_k-1_esrig_free} that if $G$ is a 3-ESRIG then $G$ is asteroidal-(non-2-ESRIG)-free.
Now we prove that if $G$ is asteroidal-(non-2-ESRIG)-free then $G$ is a 3-ESRIG.
	
By letting $\mathcal{C}$ be the class of non-2-ESRIGs, we have by Lemma~\ref{lem:gs} that there exists a set $S\subseteq V(G)$ such that $G[S]$ is an interval graph and each component of $G-S$ is a 2-ESRIG.
	
Let $\mathcal{R}=\{[c_u,d_u]\}_{u\in S}$ be an interval representation of $G[S]$ such that all endpoints of intervals are distinct. Let $\epsilon\in\mathbb{R}^+$ be such that $\epsilon<\min\{|c_u-c_v|\colon u,v\in S, u\neq v\}$. Also, let $L,R\in\mathbb{R}$ such that $L<\min_{u\in S} c_u$ and $R>\max_{u\in S} d_u$. For each vertex $u\in S$, define $t_u=\frac{c_u-L}{R-L}$. Let $\mathcal{H}$ be the set of components of $G-S$. For a vertex $u\in S$, let $\mathcal{H}_u=\{H\in\mathcal{H}\colon N(u)\cap H\neq\emptyset\}$. From Lemma~\ref{lem:removespecialblocks}\ref{it:oneneighbour}, it is clear that for each component $H\in\mathcal{H}$, there is exactly one vertex in $S$ that has neighbours in $H$. Therefore, it follows that $\{\mathcal{H}_u\}_{u\in S}$ is a partition of $\mathcal{H}$ (recall that $G$ is connected). 
Now let $H$ be a component of $\mathcal{H}_u$. Since $G$ is a tree, there is exactly one vertex $w$ of $H$ which is adjacent to $u$ in $G$.
It is easy to see that there is a 2-exactly stabbed rectangle intersection representation of $H$ such that $w$ is on the bottom stab line (take any 2-exactly stabbed rectangle intersection representation of $H$, and if the rectangle corresponding to $w$ does not intersect the bottom stab line, then reflect the whole representation about the $X$-axis).
	
Since each component of $\mathcal{H}$ is a 2-ESRIG, and because disjoint unions of 2-ESRIGs are again 2-ESRIG, we know that for $u\in S$, the graph $I_u$ formed by the disjoint union of the components in $\mathcal{H}_u$ is a 2-ESRIG. Let $\mathcal{R}_u=\{r'_v\}_{v\in I_u}$ be a 2-exactly stabbed rectangle intersection representation of $I_u$ with the stab lines $y=1$ and $y=2$ such that for any vertex $v$ of $I_u$, $span(v)\subset [c_u,c_u+\epsilon]$, and for each vertex $w\in N(u)\cap V(I_u)$ the rectangle $r'_{w}$ intersects the stab line $y=1$. Let $I_u^1$ be the subgraph induced in $I_u$ by the vertices that are on the stab line $y=1$ in $\mathcal{R}_u$. Similarly, $I_u^2$ be the subgraph induced in $I_u$ by the vertices that are on the stab line $y=2$ in $\mathcal{R}_u$. For any vertex $v\in I_u$, let $c'_v,d'_v,t'_v,b'_v$ be such that $r'_v=[c'_v,d'_v]\times [b'_v,t'_v]$.

We shall now define a rectangle $r_u$
for each vertex $u\in V(G)$ as follows.
For a vertex $u\in S$, we let $r_u=[c_u,d_u]\times [0,t_u]$. 
Consider a vertex $v\in V(G)\setminus S$. Let $u$ be the vertex in $S$ such that $v\in V(I_u)$. If $v\in V(I^2_u)$, then we let $r_v=r'_v$. If $v\in V(I^1_u)$ and $v\notin N(u)$, then we let $r_v=[c'_v,d'_v]\times [1,t'_v]$. If $v\in V(I^1_u)$ and $v\in N(u)$, then we let $r_v=[c'_v,d'_v]\times [t_u,t'_v]$.
We leave it to the reader to verify that the rectangles $\{r_u\}_{u\in V(G)}$ form a 3-exactly stabbed rectangle intersection representation of $G$.
\end{proof}

\noindent\textit{Remarks.} Let $\mathcal{C}$ be the class of non-2-ESRIG graphs and $T$ be a tree with $n$ vertices. Since checking whether $T$ is in $\mathcal{C}$ or not is possible in $O(n^2)$ time, we can infer that coloring the edges of block-tree of $T$ with respect to $\mathcal{C}$ is possible in $O(n^3)$ time. The construction procedure described in the above proof can also be performed in $O(n^3)$ time, thus giving a polynomial time algorithm to recognize trees that are 3-ESRIG. 
\medskip

We show in Section~\ref{sec:notsuff} that the forbidden structure characterizations of Theorems~\ref{thm:block2sig} and~\ref{thm:tree3sig} do not extend to block graphs that are 3-ESRIG (equivalently 3-SRIG, by Theorem~\ref{thm:equiv}) or trees that are $k$-SRIG for any $k\geq 4$. First, we explore the natural question of whether there exists a constant $c$ such that every tree is a $c$-SRIG. We give a negative answer to this question in the following section. The construction used will come in handy in Sections~\ref{sec:notsuff} and~\ref{sec:notesrig}.

\subsection{Constructing trees with high stab number}\label{sec:gl}

For a rooted tree $T$, let $root(T)$ be the root vertex of $T$. The following observation is easy to see.

\begin{observation}\label{obs:subtree}
	Let $T$ be a tree and $T'$ be a subtree of $T$ such that $T-V(T')$ has only one component.
	\begin{enumerate}
		\vspace{-0.075in}
		\itemsep 0in
		\renewcommand{\theenumi}{(\roman{enumi})}
		\renewcommand{\labelenumi}{\theenumi}
		\item\label{it:edge} For any edge $e\in E(T')$, at least one component of $T-e$ is a proper subtree of $T'$.
		\item\label{it:vertex} For any vertex $v\in V(T')$, all but one component of $T-\{v\}$ are proper subtrees of $T'$. 
	\end{enumerate}
\end{observation}

First we describe a recursive procedure to construct a rooted tree $G_l$ for all $l\geq 1$. For $l=1$, let $G_1$ be the rooted tree containing only one vertex. For any integer $l$ greater than 1, we construct $G_l$ as follows. Let $T_1,T_2$ and $T_3$ be three rooted trees each isomorphic to $G_{l-1}$. Take a $K_{1,3}$ with vertex set $\{u,u_1,u_2,u_3\}$, where $u_1,u_2,u_3$ are the pendant vertices, and construct $G_l$ by adding edges between $u_i$ and $root(T_i)$ for each $i\in\{1,2,3\}$. Also let $root(G_l)=u$. For any rooted tree $T$ with root $r$, we can define the ``ancestor'' relation on $V(T)$ in the usual way: i.e., for $u,v\in V(T)$, $u$ is an \emph{ancestor} of $v$ if and only if the path in $T$ between $r$ and $v$ contains $u$. We prove the following lemma.

\begin{lemma}\label{lem:construct_G_l}
	~
	\begin{enumerate}
		\vspace{-0.075in}
		\itemsep 0in
		\renewcommand{\theenumi}{(\roman{enumi})}
		\renewcommand{\labelenumi}{\theenumi}
		\item\label{it:not_l-1-SIG} For $l>1$, $G_l$ is not $(l-1)$-SRIG.
		\item\label{it:l-SIG} For $l\geq 1$, there is an $l$-exactly stabbed rectangle intersection representation $\mathcal{R}$ of $G_l$ such that for $v,w\in V(G_l)$, $span(v)\subseteq span(w)$ if $w$ is an ancestor of $v$ and the vertices on the top stab line of $\mathcal{R}$ are exactly the vertices in $N[root(G_l)]$.
		\item\label{it:path_rep} Let $T$ and $T'$ be two trees each isomorphic to $G_l$, for some $l\geq 1$. Let $F_l$ be the tree obtained by taking a new vertex $u$ and joining it to the root vertices of $T,T'$ using paths of length two. 
		\begin{enumerate}
			\vspace{-0.075in}
			\itemsep 0in
			\renewcommand{\theenumii}{(\alph{enumii})}
			\renewcommand{\labelenumii}{\theenumii}
			\item\label{it:path_rep1} For $l\geq 1$, there is an $l$-exactly stabbed rectangle intersection representation $\mathcal{R'}$ of $F_l$ such that for $v,w\in V(F_l)$, $span(v)\subseteq span(w)$ if $w$ is an ancestor of $v$ in $T$ or $T'$, and all vertices in the path between $root(T)$ and $root(T')$ are on the top stab line of $\mathcal{R'}$.
			\item\label{it:path_rep2} For $l\geq 2$, there is an $l$-exactly stabbed rectangle intersection representation $\mathcal{R''}$ of $F_l$ such that for $v,w\in V(F_l)$, $span(v)\subseteq span(w)$ if $w$ is an ancestor of $v$ in $T$ or $T'$, and only the vertices in $N[root(T)]\cup N[root(T')]$ are on the top stab line of $\mathcal{R''}$.
		\end{enumerate}	
		\item\label{it:disjoint_nonexist} For $l\geq 2$, there does not exist two vertex-disjoint subtrees in $G_l$ such that they are both non-$(l-1)$-ESRIG.
		\item\label{it:stabno} For $l\geq 1$, $estab(G_l)=stab(G_l)=\log_3 (n+2)$, where $n=|V(G_l)|$. 
	\end{enumerate}	
\end{lemma}

\begin{proof}
	We prove each statement separately by induction on $l$. When $l=1$, $G_l$ consists of a single vertex and therefore all the statements are true. Now we assume that the above statements are true for all integers less than $l$. 
	
	Recall that $G_l$ is obtained by taking three rooted trees $T_1,T_2,T_3$, each isomorphic to $G_{l-1}$, and then making each root adjacent to a unique pendant vertex of a $K_{1,3}$. Let $u$ be the vertex of degree 3 and $u_1,u_2,u_3$ be the pendant vertices of the $K_{1,3}$. Also recall that $root(G_l)=u$.
	
	To prove~\ref{it:not_l-1-SIG}, note that as $T_i$ is isomorphic to $G_{l-1}$ for each $i\in\{1,2,3\}$, we have by our induction hypothesis that $T_i$ is not $(l-2)$-SRIG. Therefore, $T_1,T_2,T_3$ are asteroidal-(non-($l-2$)-SRIG) in $G_l$. Using Theorem~\ref{thm:asteroidal_k-1_sig_free}, we can conclude that $G_l$ is not $(l-1)$-SRIG.
	
	To prove~\ref{it:l-SIG}, note that by our induction hypothesis, for each $i\in\{1,2,3\}$, $T_i$ has an $(l-1)$-exactly stabbed rectangle intersection representation $\mathcal{R}_i$ such that for $v,w\in V(T_i)$, $span(v)\subseteq span(w)$ if $w$ is an ancestor of $v$ and only the vertices in $N[root(T_i)]$ are on the top stab line of $\mathcal{R}_i$. Since $T_1,T_2,T_3$ are vertex disjoint, it is easy to see that there is an $(l-1)$-exactly stabbed rectangle intersection representation $\mathcal{R}$ of the subgraph induced in $G_l$ by $\cup_{i=1}^3 V(T_i)$ such that only the vertices in $\cup_{i=1}^3 N[root(T_i)]$ are on the top stab line of $\mathcal{R}$: we can just place $\mathcal{R}_1$, $\mathcal{R}_2$ and $\mathcal{R}_3$ side by side as shown in Figure~\ref{fig:g_l}(a). Now by introducing a new stab line $\ell$ above the top stab line of $\mathcal{R}$ and new rectangles corresponding to the vertices in $N[root(G_l)]=\{u,u_1,u_2,u_3\}$ into the representation such that they all intersect $\ell$, and for each $i\in\{1,2,3\}$, the rectangle corresponding to $u_i$ intersects the rectangle corresponding to $root(T_i)$ as shown in Figure~\ref{fig:g_l}(a), we can get the desired $l$-exactly stabbed rectangle intersection representation of $G_l$.
	
	\begin{figure}
	\centering
	\begin{tabular}{cc}
		\includegraphics[page=10]{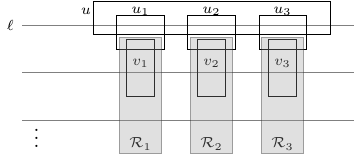}&\includegraphics[page=11]{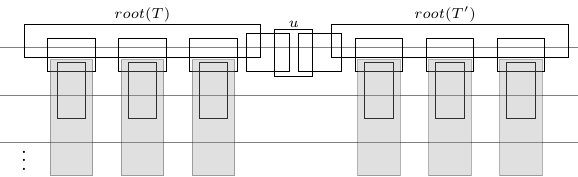}\\
		(a) & (b)
	\end{tabular}
	\medskip
	
	\begin{tabular}{c}
		\includegraphics[page=12]{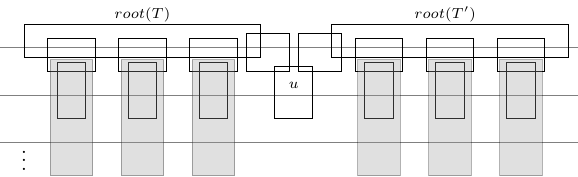}\\
		(c)
	\end{tabular}
	\caption{Construction of $G_l$ and $F_l$. The shaded region denotes a collection of rectangles. In (a), for $i\in\{1,2,3\}$, $v_i$ is the vertex $root(T_i)$. Figures (b) and (c) show different $l$-exactly stabbed rectangle intersection representations of $F_l$ as described in Lemma~\ref{lem:construct_G_l}\ref{it:path_rep1} and Lemma~\ref{lem:construct_G_l}\ref{it:path_rep2}.} 
	\label{fig:g_l}
	\end{figure}
	
	Now we prove~\ref{it:path_rep1} and~\ref{it:path_rep2}. Since $T$ and $T'$ are both isomorphic to $G_l$ and vertex disjoint, we can infer using~\ref{it:l-SIG} that there is an $l$-exactly stabbed rectangle intersection representation $\mathcal{R}$ of $F_l[V(T)\cup V(T')]$ such that for $v,w\in V(F_l)$, $span(v)\subseteq span(w)$ if $w$ is an ancestor of $v$ in $T$ or $T'$, and only the vertices in $N[root(T)]\cup N[root(T')]$ are on the top stab line of $\mathcal{R}$ (we can obtain $\mathcal{R}$ by placing the representations of $T$ and $T'$ as given by~\ref{it:l-SIG} side by side as shown in Figure~\ref{fig:g_l}(b)). Let $P$ be the path that joins $root(T)$ and $root(T')$ in $F_l$. As shown in Figure~\ref{fig:g_l}(b), we can represent $P$ such that all the vertices in $P$ are on the top stab line of $\mathcal{R}$. This proves~\ref{it:path_rep1}. Similarly, if $l\geq 2$, then as shown in Figure~\ref{fig:g_l}(c), we can represent $P$ such that only the vertices in $N[root(T)]\cup N[root(T')]$ are on the top stab line of $\mathcal{R}$. This proves~\ref{it:path_rep2}.
	
	Now we prove \ref{it:disjoint_nonexist}. Assume for the sake of contradiction that $X_1,X_2$ are two vertex-disjoint subtrees in $G_l$ such that they are both non-$(l-1)$-ESRIG. Since $G_l$ is connected, there exists an edge $e$ in $G_l$ such that if $X'_1$ and $X'_2$ are the two components in $G_l-e$, then for each $i\in\{1,2\}$, $X_i$ is a subtree of $X'_i$. This implies that both $X'_1$ and $X'_2$ are non-$(l-1)$-ESRIG.
	Suppose that $e\in E(T_i)$ for some $i\in\{1,2,3\}$. Note that $G_l-V(T_i)$ has only one component. Therefore, using Observation~\ref{obs:subtree}\ref{it:edge} we can infer that there exists $X\in\{X'_1,X'_2\}$ such that $X$ is a proper subtree of $T_i$. But as $T_i$, being isomorphic to $G_{l-1}$, is $(l-1)$-ESRIG by~\ref{it:l-SIG}, this implies that $X$ is $(l-1)$-ESRIG. This contradicts the fact that both $X'_1$ and $X'_2$ are non-$(l-1)$-ESRIG. Therefore, we can assume without loss of generality that $e$ is either $uu_1$ or the edge between $u_1$ and $root(T_1)$. If $e$ is the edge between $u_1$ and $root(T_1)$, then one of the components of $T-e$ is $T_1$, which is $(l-1)$-ESRIG by~\ref{it:l-SIG}, contradicting the fact that both components of $T-e$ are non-$(l-1)$-ESRIG. If $e$ is the edge $uu_1$, then one of the components of $T-e$ is isomorphic to $F_{l-1}$, and therefore by~\ref{it:path_rep}, is $(l-1)$-ESRIG. This again contradicts the fact that both components of $T-e$ are non-$(l-1)$-ESRIG.
	
	To prove~\ref{it:stabno}, we can solve the recurrence $|V(G_l)|=3|V(G_{l-1})|+4$ to obtain $n=|V(G_l)|=3^l-2$. Now, using~\ref{it:not_l-1-SIG} and~\ref{it:l-SIG}, we can conclude that $estab(G_l)=stab(G_l)=\log_3 (n+2)$.
\end{proof}

From Theorem~\ref{thm:blockstab}, we have that for any tree $T$ on $n$ vertices with $n\geq 3$, $estab(T)\leq\lceil\log (n-1)\rceil$. Also, using Theorem~\ref{thm:blockstab} and Lemma~\ref{lem:construct_G_l}\ref{it:stabno}, we have the following corollary.

\begin{corollary}
$estab(${\sc Trees}$,n)=\Theta(\log n)$, $stab(${\sc Trees}$,n)=\Theta(\log n)$, $estab(${\sc Block Graphs}$,n)=\Theta(\log n)$, and $stab(${\sc Block Graphs}$,n)=\Theta(\log n)$.
\end{corollary}

Although the stab number and exact stab number were equal for the trees that we constructed in this section, we shall show in Theorem~\ref{thm:srig-parameter-diff} there are trees for which these parameters differ. The graph $G_l$ and the observations in Lemma~\ref{lem:construct_G_l} will be used frequently in the remainder of the paper.

\subsection{Absence of asteroidal subgraphs is not sufficient}\label{sec:notsuff}

We showed in Theorem~\ref{thm:asteroidal_k-1_sig_free} that being asteroidal-(non-$(k-1)$-SRIG)-free is a necessary condition for a graph to be $k$-SRIG. Theorem~\ref{thm:block2sig} showed that this necessary condition is also sufficient for block graphs when $k\leq 2$ and Theorem~\ref{thm:tree3sig} demonstrated that this necessary condition is sufficient for trees when $k\leq 3$. In this section, we shall show that this necessary condition is not sufficient for block graphs for any $k\geq 3$ and it is not sufficient for trees for any $k\geq 4$. In particular, we shall prove the following two theorems.

\begin{theorem}\label{thm:blocknot3sig}
There exists a block graph that is asteroidal-(non-2-SRIG)-free, but is not 3-SRIG.
\end{theorem}

Note that by Theorem~\ref{thm:equiv}, the above theorem also means that there exists a block graph that is asteroidal-(non-2-ESRIG)-free, but is not 3-ESRIG.

\begin{theorem}\label{thm:treenotsuff}
For each integer $k\geq 4$, there exists a tree $T$ that is asteroidal-(non-($k-1$)-ESRIG)-free, but is not $k$-SRIG.
\end{theorem}

It is easy to see that Theorem~\ref{thm:treenotsuff} directly gives the following two corollaries, which tell us that the necessary conditions derived in Theorem~\ref{thm:asteroidal_k-1_sig_free} and Theorem~\ref{thm:asteroidal_k-1_esrig_free} for a tree to be a $k$-SRIG and $k$-ESRIG respectively, are not sufficient for any $k\geq 4$.

\begin{corollary}
For each integer $k\geq 4$, there exists a tree $T$ that is asteroidal-(non-($k-1$)-SRIG)-free, but is not $k$-SRIG.
\end{corollary}

\begin{corollary}
For each integer $k\geq 4$, there exists a tree $T$ that is asteroidal-(non-($k-1$)-ESRIG)-free, but is not $k$-ESRIG.
\end{corollary}

In order to prove these theorems, we develop some tools to study $k$-stabbed rectangle intersection representations using special kinds of curves in the representation that are derived from induced paths in the graph.

Consider a $k$-stabbed rectangle intersection representation $\mathcal{R}$ of a graph $G$.
In this representation, we say that a curve is \emph{rectilinear} if it consists of vertical and horizontal line segments and each horizontal line segment in it lies on a stab line.
Given an induced path $P=u_1u_2\ldots u_s$ in $G$ and two distinct points $p\in r_{u_1}$ and $p'\in r_{u_s}$ such that $p,p'$ lie on stab lines, a rectilinear curve \emph{through $P$ from $p$ to $p'$} is a simple rectilinear curve $\mathbf{p}$ that starts at $p$ and ends at $p'$, where $\mathbf{p}\subseteq\bigcup^s_{i=1} r_{u_i}$ and $\mathbf{p}\cap r_{u_i}$ is arc-connected (and nonempty) for each $i\in\{1,2,\ldots,s\}$. Note that such a curve always exists and that for each $i\in\{1,2,\ldots,s\}$, the curve contains some point in $r_{u_i}$ that is on a stab line.

Given a set $X$ of consecutive stab lines $y=a_1$, $y=a_2$, $\ldots$, $y=a_t$, such that $a_1<a_2<\cdots<a_t$, we say that $y=a_1$ is the bottom stab line in $X$ and $y=a_t$ is the top stab line in $X$. Further, we say that a connected induced subgraph $H$ of $G$ is \emph{$X$-spanning} if there is some vertex in $H$ on each stab line in $X$. An induced path in $G$ is said to be an \emph{$X$-spanning path} if its starting and ending vertices are on the top and bottom stab lines of $X$ respectively. Note that if a subgraph $H$ of $G$ is $X$-spanning, then there is an $X$-spanning path in $H$ (to see this, consider the shortest path between two vertices $u$ and $v$ in $H$ such that $u$ is on the top stab line in $X$ and $v$ is on the bottom stab line in $X$).

In the following, we use the term ``region'' to denote an arc-connected region of the plane that is bounded by a closed rectilinear curve which is the union of four simple rectilinear curves that satisfy some special properties (we assume that a region does not contain the points on its boundary).
Suppose $\mathbf{t}$, $\mathbf{l}$, $\mathbf{b}$, and $\mathbf{r}$ are four simple rectilinear curves such that $\mathbf{l}\cap\mathbf{r}=\emptyset$, $\mathbf{t}\cap\mathbf{b}=\emptyset$, and for each $(\mathbf{x},\mathbf{y})\in\{(\mathbf{t},\mathbf{l}),(\mathbf{l},\mathbf{b}),(\mathbf{b},\mathbf{r}),(\mathbf{r},\mathbf{t})\}$, the curves $\mathbf{x}$ and $\mathbf{y}$ have exactly one point in common which is also an end point of both of them. Then, the region $R=(\mathbf{t},\mathbf{l},\mathbf{b},\mathbf{r})$ is the bounded arc-connected component of $\mathbb{R}^2\setminus (\mathbf{t}\cup\mathbf{l}\cup\mathbf{b}\cup\mathbf{r})$. The closed rectilinear curve $\mathbf{t}\cup\mathbf{l}\cup\mathbf{b}\cup\mathbf{r}$ is called the ``boundary'' of $R$. For a region $R$, we let $\mathcal{L}_{\mathcal{R}}(R)$ denote the set of stab lines of $\mathcal{R}$ that intersect $R$. Also, let $G_R$ denote the subgraph induced in $G$ by the vertices whose rectangles lie completely inside $R$.

\begin{observation}\label{obs:ltlb}
Let $\ell_t$, $\ell_b$ be the stab lines just above and just below the top and bottom stab lines in $\mathcal{L}_{\mathcal{R}}(R)$ respectively. Then, no point on the boundary of $R$ lies above $\ell_t$ or below $\ell_b$.
\end{observation}
\begin{proof}
Suppose that the boundary of $R$ contains a point $p$ that is above $\ell_t$.
Let $p'$ be an arbitrary point in $R$ that is on the top stab line in $\mathcal{L}_{\mathcal{R}}(R)$. It is easy to see that there exists a simple curve from $p'$ to $p$ all of whose points except $p$ belong to $R$. Since $p'$ is below $\ell_t$ and $p$ above it, there must be a point on this curve that lies on $\ell_t$. But this would mean that $R$ intersects $\ell_t$, contradicting the fact that $\ell_t\notin\mathcal{L}_{\mathcal{R}}(R)$. Using similar arguments, we can prove that no point on the boundary of $R$ lies below $\ell_b$.
\end{proof}
\medskip

\begin{definition}\label{def:good}
	A region $R=(\mathbf{t},\mathbf{l},\mathbf{b},\mathbf{r})$ is said to be ``good'' if it has the following properties:
	\begin{enumerate}
		\vspace{-0.075in}
		\itemsep -0.05in
		\renewcommand{\theenumi}{(\roman{enumi})}
		\renewcommand{\labelenumi}{\theenumi}
		\item the parts of $\mathbf{l}$ and $\mathbf{r}$ that are above the top stab line in $\mathcal{L}_{\mathcal{R}}(R)$ and below the bottom stab line in $\mathcal{L}_{\mathcal{R}}(R)$ consist of just a vertical segment each, or in other words, every horizontal segment of $\mathbf{l}$ and $\mathbf{r}$ lies on a stab line in $\mathcal{L}_{\mathcal{R}}(R)$,
		\item no point of $\mathbf{t}$ lies below the bottom stab line of $\mathcal{L}_{\mathcal{R}}(R)$, and
		\item no point of $\mathbf{b}$ lies above the top stab line of $\mathcal{L}_{\mathcal{R}}(R)$.
	\end{enumerate} 
\end{definition}

For a good region $R=(\mathbf{t},\mathbf{l},\mathbf{b},\mathbf{r})$, we let $\mathbf{top}(R)=\mathbf{t}$ and $\mathbf{bottom}(R)=\mathbf{b}$.
\medskip

Let $R=(\mathbf{t},\mathbf{l},\mathbf{b},\mathbf{r})$ be a good region with $|\mathcal{L}_{\mathcal{R}}(R)|\geq 1$. Let $P_1$ and $P_2$ be two neighbour-disjoint $\mathcal{L}_{\mathcal{R}}(R)$-spanning paths in $G_R$. For $i\in\{1,2\}$, let $u_i,v_i$ be the endvertices of $P_i$ that are on the top and bottom stab lines in $\mathcal{L}_{\mathcal{R}}(R)$ respectively. For $i\in\{1,2\}$, let $\mathbf{p_i}$ be a rectilinear curve that starts at a point $(x_i,y_i)\in r_{u_i}$ on the top stab line in $\mathcal{L}_{\mathcal{R}}(R)$ and ends at a point $(x'_i,y'_i)\in r_{v_i}$ on the bottom stab line in $\mathcal{L}_{\mathcal{R}}(R)$, with the following additional properties:
\begin{enumerate}
\vspace{-0.075in}
\itemsep -0.05in
\renewcommand{\theenumi}{(\roman{enumi})}
\renewcommand{\labelenumi}{\theenumi}
\item The only point in $\mathbf{p_i}$ that is in $r_{u_i}$ and is also on the top stab line in $\mathcal{L}_{\mathcal{R}}(R)$ is $(x_i,y_i)$, and
\item The only point in $\mathbf{p_i}$ that is in $r_{v_i}$ and is also on the bottom stab line in $\mathcal{L}_{\mathcal{R}}(R)$ is $(x'_i,y'_i)$.
\end{enumerate}
It is not difficult to see that the curves $\mathbf{p_1},\mathbf{p_2}$ always exist. (Take any rectilinear curve $\mathbf{q}$ through $P_i$ between some point on the top stab line in $r_{u_i}$ and some point on the bottom stab line in $r_{v_i}$. Let $(x_i,y_i)$ be the last point in $\mathbf{q}$ that is both in $r_{u_i}$ and is on the top stab line and let $(x'_i,y'_i)$ be the first point in $\mathbf{q}$ that is both in $r_{v_i}$ and is on the bottom stab line. Then the subcurve of $\mathbf{q}$ between $(x_i,y_i)$ and $(x'_i,y'_i)$ can be taken as $\mathbf{p_i}$.)

Suppose that there is a path in $G_R$ between a vertex of $P_1$ and a vertex of $P_2$. Then, let $P$ be the induced path in $G_R$ between a vertex $w_1$ in $P_1$ and a vertex $w_2$ in $P_2$ such that all other vertices of $P$ belong to neither $P_1$ nor $P_2$. Let $p_1,p_2$ be points on stab lines where for $i\in\{1,2\}$, $p_i\in\mathbf{p_i}\cap r_{w_i}$, such that there exists a rectilinear curve $\mathbf{p}$ through $P$ from $p_1$ to $p_2$, whose interior points belong to neither $\mathbf{p_1}$ nor $\mathbf{p_2}$ (note that $p_1$, $p_2$ and $\mathbf{p}$ always exist --- take arbitrary points $p,p'$ on stab lines such that $p\in\mathbf{p_1}\cap r_{w_1}$, $p'\in\mathbf{p_2}\cap r_{w_2}$ and consider the rectilinear curve $\mathbf{p'}$ through $P$ between $p$ and $p'$; $p_1,p_2$ can be chosen to be the closest pair of points on $\mathbf{p'}$ such that $p_1\in\mathbf{p_1}$, $p_2\in\mathbf{p_2}$, and the part of $\mathbf{p'}$ between $p_1$ and $p_2$ can be chosen as $\mathbf{p}$). Please refer to Figure~\ref{fig:regions}(a) for an example showing the different curves in $R$.
For $i\in\{1,2\}$, let $\mathbf{s_i}$ be the shortest vertical line segment with its bottom endpoint being $(x_i,y_i)$ and top endpoint being a point on the boundary of $R$. Similarly, for $i\in\{1,2\}$, let $\mathbf{s'_i}$ be the shortest vertical line segment with its top endpoint being $(x'_i,y'_i)$ and bottom endpoint being a point on the boundary of $R$ (refer Figure~\ref{fig:regions}(b)).

\begin{figure}
\centering
\centering
\begin{tabular}{cc} 
	\includegraphics[page=24]{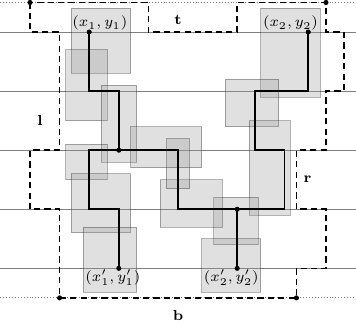}& \includegraphics[page=16]{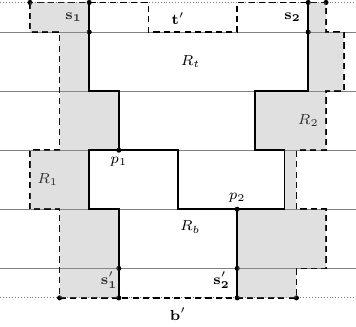}\\
	(a)&(b)
\end{tabular}
\caption{An example of a good region $R=(\mathbf{t},\mathbf{l},\mathbf{b},\mathbf{r})$ (whose boundary is shown using thick dashed lines) containing the rectangles corresponding to minimal spanning paths $P_1$ and $P_2$ and a path $P$ connecting them. (a) shows the rectilinear curves $\mathbf{p_1}$, $\mathbf{p_2}$ and $\mathbf{p}$ through these paths using thick solid lines.  (b) shows the partition of $R$ into the four regions $R_1$, $R_2$, $R_t$ and $R_b$.
}\label{fig:regions}
\end{figure}

\begin{observation}\label{obs:topsi}
For each $i\in\{1,2\}$, the top endpoint of $\mathbf{s_i}$ lies on the stab line just above the top stab line in $\mathcal{L}_{\mathcal{R}}(R)$ and on a horizontal segment of $\mathbf{t}$ and the bottom endpoint of $\mathbf{s'_i}$ lies on the stab line just below the bottom stab line in $\mathcal{L}_{\mathcal{R}}(R)$ and on a horizontal segment of $\mathbf{b}$. 
\end{observation}
\begin{proof}
For $i\in\{1,2\}$, we know that the top endpoint of $\mathbf{s_i}$ lies on the boundary of $R$, and hence on a horizontal segment of the boundary of $R$. This implies that the top endpoint of $\mathbf{s_i}$ lies on a stab line. Also, note that the bottom endpoint of $\mathbf{s_i}$ is a point in $R$ that is on the top stab line in $\mathcal{L}_{\mathcal{R}}(R)$. This means that the top endpoint of $\mathbf{s_i}$ lies above the top stab line in $\mathcal{L}_{\mathcal{R}}(R)$. Since the top endpoint of $\mathbf{s_i}$ lies on the boundary of $R$, we immediately have from Observation~\ref{obs:ltlb} that it lies on the stab line just above the top stab line in $\mathcal{L}_{\mathcal{R}}(R)$. Also, since it lies on a horizontal segment of the boundary of $R$, it lies on some horizontal segment that belongs to one of the curves $\mathbf{t},\mathbf{l},\mathbf{b},\mathbf{r}$. Since $R$ is good, we know that no horizontal segment of $\mathbf{l}$, $\mathbf{r}$ or $\mathbf{b}$ lies above the top stab line in $\mathcal{L}_{\mathcal{R}}(R)$. This means that the top endpoint of $\mathbf{s_i}$ lies on a horizontal segment of $\mathbf{t}$. Using similar reasoning, it can be seen that for $i\in\{1,2\}$, the bottom endpoint of $\mathbf{s'_i}$ lies on the stab line just below the bottom stab line in $\mathcal{L}_{\mathcal{R}}(R)$ and on a horizontal segment of $\mathbf{b}$.
\end{proof}

Let $\mathbf{t'}\subseteq\mathbf{t}$ be the portion of the curve $\mathbf{t}$ that starts at the top endpoint of $\mathbf{s_1}$ and ends at the top endpoint of $\mathbf{s_2}$. Similarly, let $\mathbf{b'}\subseteq\mathbf{b}$ be the portion of the curve $\mathbf{b}$ that starts at the bottom endpoint of $\mathbf{s'_1}$ and ends at the bottom endpoint of $\mathbf{s'_2}$.

For $i\in\{1,2\}$, let the curve $\mathbf{p^t_i}$ be the connected portion of $\mathbf{p_i}$ that starts at $(x_i,y_i)$ and ends at the common point of $\mathbf{p_i}$ and $\mathbf{p}$ (denoted as $p_i$ previously) and let the curve $\mathbf{p^b_i}$ be the connected portion of $\mathbf{p_i}$ that starts at the common point of $\mathbf{p_i}$ and $\mathbf{p}$ and ends at $(x'_i,y'_i)$.

Let $R_1,R_2,R_t,R_b$ be the regions into which the region $R$ gets split by the union of the curves $\mathbf{p_1},\mathbf{p_2},\mathbf{p},\mathbf{s_1},$ $\mathbf{s'_1},\mathbf{s_2},\mathbf{s'_2}$, where $R_i$, for $i\in\{1,2\}$, is the region whose boundary contains $\mathbf{p_i}$,
$R_t=(\mathbf{t'},\mathbf{s_1}\cup\mathbf{p^t_1},\mathbf{p},\mathbf{p^t_2}\cup\mathbf{s_2})$, and $R_b=(\mathbf{p},\mathbf{p^b_1}\cup\mathbf{s'_1},\mathbf{b'},\mathbf{s'_2}\cup\mathbf{p^b_2})$ (please refer to Figure~\ref{fig:regions}(b)).

\begin{observation}\label{obs:hitsP}
	From the definition of $R_t$ and $R_b$, we have:
	\begin{enumerate}
		\vspace{-0.075in}
		\itemsep 0in
		\renewcommand{\theenumi}{(\roman{enumi})}
		\renewcommand{\labelenumi}{\theenumi}
		\item\label{it:tt'bb'} $\mathbf{top}(R_t)\subseteq\mathbf{top}(R)$ and $\mathbf{bottom}(R_b)\subseteq\mathbf{bottom}(R)$.
		\item $\mathbf{bottom}(R_t)=\mathbf{top}(R_b)$.
		\item\label{it:xinP} If $x$ is a vertex in $P$, then $r_x$ intersects $\mathbf{bottom}(R_t)$ ($=\mathbf{top}(R_b)$).
		\item\label{it:xhitsP} Let $x\in V(G)$ such that $r_x$ intersects $\mathbf{bottom}(R_t)$ ($=\mathbf{top}(R_b)$). Then $x$ has a neighbour in $P$.
	\end{enumerate}
\end{observation}
\medskip

For the rest of this section, for a good region $R$ and paths $P_1,P_2,P$ such that:
\begin{itemize}
\vspace{-0.075in}
\itemsep 0in
\item $P_1$ and $P_2$ are two neighbour-disjoint $\mathcal{L}_{\mathcal{R}}(R)$-spanning paths in $G_R$, and
\item $P$ is an induced path in $G_R$ between a vertex in $P_1$ and a vertex in $P_2$ such that all vertices of $P$ other than its end vertices belong to neither $P_1$ nor $P_2$ (note that such a path will exist if there is some path in $G_R$ between a vertex of $P_1$ and a vertex of $P_2$),
\end{itemize}
we shall denote by $\Delta(\mathcal{R},R,P_1,P_2,P)$ the ordered pair $(R_t,R_b)$, where the regions $R_t$ and $R_b$ are obtained using the procedure described above. We shall now prove some observations about the regions $R_t$ and $R_b$.

\begin{lemma}\label{lem:bottomstab}
	~
	\begin{enumerate}
		\vspace{-0.075in}
		\itemsep 0in
		\renewcommand{\theenumi}{(\alph{enumi})}
		\renewcommand{\labelenumi}{\theenumi}
		\item\label{it:t'b'} The curve $\mathbf{t'}$ (resp. $\mathbf{b'}$) does not intersect the bottom (resp. top) stab line in $\mathcal{L}_{\mathcal{R}}(R)$.
		\item\label{it:decreasestab} $R_t$ does not intersect the bottom stab line in $\mathcal{L}_{\mathcal{R}}(R)$ and $R_b$ does not intersect the top stab line in $\mathcal{L}_{\mathcal{R}}(R)$.
	\end{enumerate}
\end{lemma}
\begin{proof}
	Let us first prove~\ref{it:t'b'}.
	We shall only show that the curve $\mathbf{t'}$ does not intersect the bottom stab line in $\mathcal{L}_{\mathcal{R}}(R)$ as the other case is similar.
	Let the rectilinear curve $\mathbf{q}$ be $\mathbf{p^t_1}\cup\mathbf{p}\cup\mathbf{p^t_2}$. Note that $\mathbf{q}$ is a simple rectilinear curve.
	Let $\ell$ be the stab line just above the top stab line of $\mathcal{L}_{\mathcal{R}}(R)$. From Observation~\ref{obs:topsi}, we have that the top endpoints of $\mathbf{s_1}$ and $\mathbf{s_2}$ lie on $\ell$. Let the horizontal line segment (that lies entirely on $\ell$) between these two points be denoted by $\mathbf{s}$. Let $R'$ be the region bounded by $\mathbf{s_1}\cup\mathbf{q}\cup\mathbf{s_2}\cup\mathbf{s}$. From Observation~\ref{obs:ltlb}, it is then clear that $\mathbf{t'}$ lies entirely in $R'\cup\mathbf{s}$ (recall that $R'$ consists only of the points in the interior of the region bounded by $\mathbf{s_1}\cup\mathbf{q}\cup\mathbf{s_2}\cup\mathbf{s}$). Since the points in $\mathbf{q}$ all belong to rectilinear curves through paths in $G_R$, every horizontal segment of $\mathbf{q}$ is on a stab line in $\mathcal{L}_{\mathcal{R}}(R)$. Since the endpoints of $\mathbf{q}$ lie on the top stab line in $\mathcal{L}_{\mathcal{R}}(R)$, and $\mathbf{q}$ is a simple rectilinear curve, it follows that every point in $\mathbf{q}$ is on or above the bottom stab line in $\mathcal{L}_{\mathcal{R}}(R)$. As the points in $\mathbf{s_1}\cup\mathbf{s_2}\cup\mathbf{s}$ lie on or above the top stab line in $\mathcal{L}_{\mathcal{R}}(R)$, this means that all the points on the boundary of $R'$ lie on or above the bottom stab line in $\mathcal{L}_{\mathcal{R}}(R)$, implying that $R'$ does not intersect the bottom stab line in $\mathcal{L}_{\mathcal{R}}(R)$. As $\mathbf{s}$ lies on the stab line just above the top stab line in $\mathcal{L}_{\mathcal{R}}(R)$, we now have that $R'\cup\mathbf{s}$ does not intersect the bottom stab line in $\mathcal{L}_{\mathcal{R}}(R)$. From our earlier observation that $\mathbf{t'}$ lies entirely in $R'\cup\mathbf{s}$, we now have that $\mathbf{t'}$ does not intersect the bottom stab line in $\mathcal{L}_{\mathcal{R}}(R)$.
	
	To prove~\ref{it:decreasestab}, we shall only prove that $R_t$ does not intersect the bottom stab line in $\mathcal{L}_{\mathcal{R}}(R)$ as the case for $R_b$ involves similar arguments.
	Note that the boundary of $R_t$ is $\mathbf{t'}\cup\mathbf{s_1}\cup\mathbf{q}\cup\mathbf{s_2}$. From the arguments in the previous paragraph, it is easy to see that all the points in $\mathbf{s_1}\cup\mathbf{q}\cup\mathbf{s_2}$ lie on or above the bottom stab line in $\mathcal{L}_{\mathcal{R}}(R)$. Combining this with~\ref{it:t'b'}, we now have that all the points on the boundary of $R_t$ lie on or above the bottom stab line in $\mathcal{L}_{\mathcal{R}}(R)$. Hence we can conclude that the bottom stab line in $\mathcal{L}_{\mathcal{R}}(R)$ does not intersect $R_t$.
\end{proof}

An $\mathcal{L}_{\mathcal{R}}(R)$-spanning path $P$ is said to be a \emph{minimal $\mathcal{L}_{\mathcal{R}}(R)$-spanning path} if there is no $\mathcal{L}_{\mathcal{R}}(R)$-spanning path $P'$ such that $V(P')\subset V(P)$. Note that the existence of an $\mathcal{L}_{\mathcal{R}}(R)$-spanning path in a graph implies the existence of a minimal $\mathcal{L}_{\mathcal{R}}(R)$-spanning path in the graph.

\begin{lemma}\label{lem:goodregion}
Suppose that $P_1$ and $P_2$ are minimal $\mathcal{L}_{\mathcal{R}}(R)$-spanning paths. Let $R'\in\{R_t,R_b\}$ such that $|\mathcal{L}_{\mathcal{R}}(R')|\geq |\mathcal{L}_{\mathcal{R}}(R)|-1$. Then $R'$ is good.
\end{lemma}
\begin{proof}
We shall prove this only for the case when $R'=R_t$ as the other case is similar.
As $|\mathcal{L}_{\mathcal{R}}(R_t)|\geq |\mathcal{L}_{\mathcal{R}}(R)|-1$, and by Lemma~\ref{lem:bottomstab}\ref{it:decreasestab}, $R_t$ does not intersect the bottom stab line, we know that $\mathcal{L}_{\mathcal{R}}(R_t)$ consists of all the stab lines in $\mathcal{L}_{\mathcal{R}}(R)$ other than the bottom stab line in $\mathcal{L}_{\mathcal{R}}(R)$.
	
Recall that $R_t=(\mathbf{t'},\mathbf{s_1}\cup\mathbf{p^t_1},\mathbf{p},\mathbf{p^t_2}\cup\mathbf{s_2})$.
Since the paths $P_1$ and $P_2$ are minimal, we know that for $i\in\{1,2\}$, $u_i$ is the only vertex on $P_i$ that is on the top stab line in $\mathcal{L}_{\mathcal{R}}(R)$ and $v_i$ is the only vertex on $P_i$ that is on the bottom stab line in $\mathcal{L}_{\mathcal{R}}(R)$. Therefore, from the definition of curves $\mathbf{p_1}$ and $\mathbf{p_2}$, we have that for $i\in\{1,2\}$, the only points of $\mathbf{p_i}$ that lie on the top and bottom stab lines in $\mathcal{L}_{\mathcal{R}}(R)$ are the endpoints of $\mathbf{p_i}$, which further implies that $\mathbf{p_i}$ does not contain any horizontal segment on the top or bottom stab lines in $\mathcal{L}_{\mathcal{R}}(R)$. It follows that for $i\in\{1,2\}$, $\mathbf{p^t_i}$, and therefore $\mathbf{s_i}\cup\mathbf{p^t_i}$, also does not contain any horizontal segment on the top or bottom stab lines in $\mathcal{L}_{\mathcal{R}}(R)$. As $\mathcal{L}_{\mathcal{R}}(R_t)$ consists of all the stab lines in $\mathcal{L}_{\mathcal{R}}(R)$ other than the bottom stab line in $\mathcal{L}_{\mathcal{R}}(R)$, we have that $\mathbf{s_1}\cup\mathbf{p^t_1}$ and $\mathbf{s_2}\cup\mathbf{p^t_2}$ do not contain any horizontal segment that lies above the top stab line in $\mathcal{L}_{\mathcal{R}}(R_t)$ or below the bottom stab line in $\mathcal{L}_{\mathcal{R}}(R_t)$. Therefore, $R_t$ satisfies property~(i) of Definition~\ref{def:good}. From Lemma~\ref{lem:bottomstab}\ref{it:t'b'}, we have that $\mathbf{t'}$ does not intersect the bottom stab line in $\mathcal{L}_{\mathcal{R}}(R)$. Since the endpoints of $\mathbf{t'}$ lie above the top stab line in $\mathcal{L}_{\mathcal{R}}(R)$, we can then conclude using the definition of rectilinear curves that no point of $\mathbf{t'}$ lies below the bottom stab line of $\mathcal{L}_{\mathcal{R}}(R_t)$. Thus, $R_t$ satisfies property~(ii) of Definition~\ref{def:good}. Since the points in $\mathbf{p}$ all belong to rectangles contained in $R$ and $\mathbf{p}$ is a simple rectilinear curve, we know that all of them are on or below the top stab line in $\mathcal{L}_{\mathcal{R}}(R)$ and hence on or below the top stab line in $\mathcal{L}_{\mathcal{R}}(R_t)$. Therefore, $R_t$ satisfies property~(iii) of Definition~\ref{def:good} as well. This completes the proof.
\end{proof}

\begin{observation}\label{obs:gap}
Let $v\in V(G)$. For $i\in\{1,2\}$, if $r_v\cap\mathbf{t'}=\emptyset$ (resp. $r_v\cap\mathbf{b'}=\emptyset$) and $r_v$ intersects $\mathbf{s_i}$ (resp. $\mathbf{s'_i}$), then $r_v$ contains the point $(x_i,y_i)$ (resp. $(x'_i,y'_i)$).
\end{observation}
\begin{proof}
Suppose that $r_v\cap\mathbf{t'}=\emptyset$, but $r_v$ intersects $\mathbf{s_i}$. As the top endpoint of $\mathbf{s_i}$ is contained in $\mathbf{t'}$, we can infer that $r_v$ does not contain the top endpoint of $\mathbf{s_i}$. If $r_v$ also does not contain the bottom endpoint of $\mathbf{s_i}$, then there is no stab line that intersects $r_v$, as the top and bottom endpoints of $\mathbf{s_i}$ are on consecutive stab lines. We can therefore conclude that the bottom endpoint of $\mathbf{s_i}$, which is $(x_i,y_i)$, is contained in $r_v$. The arguments for the other case are similar and are therefore omitted.
\end{proof}

\begin{lemma}\label{lem:crossboundary}
	Let $v\in V(G)$ such that $r_v$ intersects the boundary of $R_t$ (resp. $R_b$). Then either $r_v$ intersects $\mathbf{t'}=\mathbf{top}(R_t)$ (resp. $\mathbf{b'}=\mathbf{bottom}(R_b)$) or $v$ has a neighbour on at least one of the paths $P_1$, $P_2$, or $P$.
\end{lemma}
\begin{proof}
	We shall prove this lemma only for $R_t$ as the arguments for $R_b$ are similar. Suppose there exists a vertex $v\in V(G)$ such that $r_v$ intersects the boundary of $R_t$, but $r_v$ does not intersect $\mathbf{t'}$ and $v$ does not have a neighbour on any of the paths $P_1$, $P_2$, or $P$. Then $r_v$ does not intersect any of the curves $\mathbf{p^t_1}$, $\mathbf{p}$, or $\mathbf{p^t_2}$. From this, it follows that $r_v$ does not contain the points $(x_1,y_1)$ or $(x_2,y_2)$. By Observation~\ref{obs:gap}, we now have that $r_v$ does not intersect $\mathbf{s_1}$ or $\mathbf{s_2}$. Since the boundary of $R_t$ is $\mathbf{t'}\cup\mathbf{s_1}\cup\mathbf{p^t_1}\cup\mathbf{p}\cup\mathbf{p^t_2}\cup\mathbf{s_2}$, this means that $r_v$ does not intersect the boundary of $R_t$, which is a contradiction.
\end{proof}

\begin{lemma}\label{lem:rtorrb}
	Let $v\in V(G_R)$ such that $P$ misses $v$ and there is a path in $G_R$ from $v$ to a vertex in $P$ that misses both $P_1$ and $P_2$. Then $r_v$ is contained in $R_t$ or $R_b$.
\end{lemma}
\begin{proof}
	As $v$ is not adjacent to any vertex in $P$, $P_1$ or $P_2$, the rectangle $r_v$ does not intersect $\mathbf{p}$, $\mathbf{p_1}$ or $\mathbf{p_2}$. Also, as $v\in V(G_R)$, $r_v$ does not intersect $\mathbf{t'}$ or $\mathbf{b'}$. Then, by Observation~\ref{obs:gap}, we can further infer that $r_v$ does not intersect $\mathbf{s_1}$, $\mathbf{s_2}$, $\mathbf{s'_1}$ or $\mathbf{s'_2}$. This means that $r_v$ is contained in one of the regions $R_1$, $R_2$, $R_t$ or $R_b$. Now suppose for the sake of contradiction that $r_v$ is contained in $R_1$. We know from the statement of the lemma that there is at least one path in $G_R$ from $v$ to some vertex in $P$ that misses both $P_1$ and $P_2$. Let $Q$ be such a path of minimum length and let $u$ be the endpoint of $Q$ other than $v$. It is clear that $V(P)\cap V(Q)=\{u\}$. Let $p'$ be a point in $r_u\cap\mathbf{p}$ that is on a stab line (recall from the definition of rectilinear curves through paths that such a point exists). As $u$ has no neighbour on $P_1$ or $P_2$, it can be seen that $p'$ is not an endpoint of $\mathbf{p}$, i.e., $p'$ is an interior point of $\mathbf{p}$. Now consider the rectilinear path through $Q$ from some point in $r_v$ (that is on a stab line) to $p'$. As the point $p'$ is not inside or on the boundary of $R_1$, this rectilinear curve must cross the boundary of $R_1$ at some point $p''$. It is clear that there is a vertex $x$ in $Q$ such that $p''\in r_x$. Since $r_x$ is contained in $R$, we can infer that $p''$ is on $\mathbf{s_1}\cup\mathbf{p_1}\cup\mathbf{s'_1}$ and also that $r_x$ does not intersect $\mathbf{t'}$ or $\mathbf{b'}$. If $p''$ is on $\mathbf{s_1}$ or $\mathbf{s'_1}$, we have by Observation~\ref{obs:gap} that $r_x$ intersects $\mathbf{p_1}$. So we can conclude that in any case, $r_x$ intersects $\mathbf{p_1}$. Since from the definition of $\mathbf{p_1}$, every point of $\mathbf{p_1}$ belongs to the rectangle corresponding to some vertex of $P_1$, this implies that $x$ is adjacent to some vertex of $P_1$. This contradicts the fact that $Q$ misses $P_1$. We can thus conclude that $r_v$ is not contained in $R_1$. Using similar arguments, we can also infer that $r_v$ is not contained in $R_2$. This completes the proof.
\end{proof}

\begin{lemma}\label{lem:winrtorrb}
Let $v,w\in V(G_R)$ such that $r_v$ is contained in $R'\in\{R_t,R_b\}$ and there is a path in $G_R$ between $v$ and $w$ that misses $P_1$, $P_2$ and $P$. Then $r_w$ is contained in $R'$.
\end{lemma}
\begin{proof}
	We shall prove the statement of the lemma only for the case $R'=R_t$ as the proof for the case $R'=R_b$ is similar.
	Let $Q$ be the path between $v$ and $w$ in $G_R$ that misses $P_1$, $P_2$ and $P$. Let $x$ be any vertex on $Q$. Clearly, $x$ has no neighbour on $P_1$, $P_2$ or $P$. As $x\in V(G_R)$, the rectangle $r_x$ is contained in $R$, implying that $r_x$ does not intersect the boundary of $R$. As we have $\mathbf{top}(R_t)\subseteq\mathbf{top}(R)$ by Observation~\ref{obs:hitsP}\ref{it:tt'bb'}, this means that $r_x$ does not intersect $\mathbf{top}(R_t)$. By Lemma~\ref{lem:crossboundary}, we now have that $r_x$ does not intersect the boundary of $R_t$. Therefore, no rectangle corresponding to a vertex in $Q$ can intersect the boundary of $R_t$. Since $r_v$ is contained in $R_t$, this means that the rectangle corresponding to each vertex of $Q$, and hence $r_w$, is contained in $R_t$.
\end{proof}

We shall use the technical details about good regions and rectilinear curves only for the proof of Theorem~\ref{thm:treenotsuff}. We now give a lemma that shall be sufficient for most of the other proofs. Given a graph $G$ and a representation $\mathcal{R}$ of $G$, we shall define $\mathcal{L}_{\mathcal{R}}(H)$, for any connected induced subgraph $H$ of $G$, to be the set of stab lines of $\mathcal{R}$ that intersect the rectangle corresponding to some vertex in $V(H)$. Note that $\mathcal{L}_{\mathcal{R}}(H)$ will contain a consecutive set of stab lines of $\mathcal{R}$.

\begin{lemma}\label{lem:useful}
Let $G$ be a connected $k$-SRIG and $\mathcal{R}$ a $k$-stabbed rectangle intersection representation of it. Let $H_1$ and $H_2$ be two neighbour-disjoint connected induced subgraphs of $G$ such that $\mathcal{L}_{\mathcal{R}}(H_1)=\mathcal{L}_{\mathcal{R}}(H_2)=\mathcal{L}_{\mathcal{R}}(G)=k$. Let $P$ be an induced path in $G$ between some vertex in $V(H_1)$ and some vertex in $V(H_2)$ such that no internal vertex of $P$ is in $V(H_1)$ or $V(H_2)$. Let $H$ be a connected induced subgraph of $G$ that is neighbour-disjoint from $H_1$, $H_2$ and $P$ such that there is a vertex in $H$ from which there is a path to a vertex of $P$ that misses both $H_1$ and $H_2$. Then, $\mathcal{L}_{\mathcal{R}}(H)\subset\mathcal{L}_{\mathcal{R}}(G)$.
\end{lemma}
\begin{proof}
We shall augment $\mathcal{R}$ to a new representation $\mathcal{R}'$ by adding two new stab lines, one above the top stab line and the other below the bottom stab line of $\mathcal{R}$. Notice that for any connected induced subgraph $G'$ of $G$, we have $\mathcal{L}_{\mathcal{R}'}(G')=\mathcal{L}_{\mathcal{R}}(G')$. Let $A$ be a good region that contains all the rectangles of $\mathcal{R}'$, i.e., $G_A=G$ (note that such a region exists; we can consider a rectangle with top and bottom edges on the top and bottom stab lines such that it contains all the rectangles of $\mathcal{R}'$). As the only two stab lines that are not intersected by any rectangle in $\mathcal{R}'$ are the top and bottom stab lines (recall that $\mathcal{L}_{\mathcal{R}}(G)$ contains all the stab lines of $\mathcal{R}$), it follows that $\mathcal{L}_{\mathcal{R}'}(A)=\mathcal{L}_{\mathcal{R}'}(G)$. It is clear that for any induced subgraph $G'$ of $G$, $\mathcal{L}_{\mathcal{R}'}(G')=\mathcal{L}_{\mathcal{R}}(G')$. Therefore, we have $\mathcal{L}_{\mathcal{R}'}(H_1)=\mathcal{L}_{\mathcal{R}'}(H_2)=\mathcal{L}_{\mathcal{R}'}(G)$, which implies that there are $\mathcal{L}_{\mathcal{R}'}(A)$-spanning paths in each of them. Let $P_1$ and $P_2$ be minimal $\mathcal{L}_{\mathcal{R}'}(A)$-spanning paths in $H_1$ and $H_2$ respectively. As $H_1$ and $H_2$ are neighbour-disjoint, $P_1$ and $P_2$ are neighbour-disjoint. It is not hard to see that there exists an induced path $P'$ in $G[V(H_1)\cup V(P)\cup V(H_2)]$ that contains $P$ as a subpath, such that $P'$ connects some vertex of $P_1$ to some vertex of $P_2$ and no internal vertex of $P'$ belongs to either $P_1$ or $P_2$. Let $(A_t,A_b)=\Delta(\mathcal{R}',A,P_1,P_2,P')$.
	
We know that there exists a vertex, say $v$, in $H$ such that there is a path from $v$ to a vertex of $P$ that misses both $H_1$ and $H_2$. Clearly, this is also a path from $v$ to a vertex in $P'$ that misses both $P_1$ and $P_2$. As $H$ is neighbour-disjoint from $P$, we know that $P$ misses $v$. By Lemma~\ref{lem:rtorrb}, we know that $r_v$ is contained in $A_t$ or $A_b$. Let us assume without loss of generality that $r_v$ is contained in $A_t$. Since $H$ is a connected induced subgraph of $G$ that is neighbour-disjoint from $H_1$, $H_2$ and $P$, we know that there is a path from $v$ to each vertex of $H$ that misses $H_1$, $H_2$ and $P$. This means that there is a path from $v$ to each vertex of $H$ that misses $P_1$, $P_2$ and $P'$. Now, we can use Lemma~\ref{lem:winrtorrb} to conclude that the rectangles corresponding to the vertices of $H$ are all contained in $A_t$. Since by Lemma~\ref{lem:bottomstab}\ref{it:decreasestab}, we know that $\mathcal{L}_{\mathcal{R}'}(A_t)\subset\mathcal{L}_{\mathcal{R}'}(A)$, we can now conclude that $\mathcal{L}_{\mathcal{R}'}(H)\subset\mathcal{L}_{\mathcal{R}'}(G)$, and therefore $\mathcal{L}_{\mathcal{R}}(H)\subset\mathcal{L}_{\mathcal{R}}(G)$.
\end{proof}

\medskip

\noindent\textbf{\textit{Proof of Theorem~\ref{thm:blocknot3sig}.}}
\medskip

Let $T$ be the block graph obtained by taking a copy of the tree $G_2$ (defined in Section~\ref{sec:gl}) and then introducing a true twin for one of the leaves. Let $w,w'$ be the two true twins in $T$, $v$ be their common neighbour and $u$ the degree 3 vertex adjacent to $v$. See Figure~\ref{fig:drawtree}(a) for a drawing of $T$. Notice that the graph $G_2$ is non-interval (folklore, or by Lemma~\ref{lem:construct_G_l}\ref{it:not_l-1-SIG}).
\medskip


Let $T_1$ and $T_2$ be trees each isomorphic to $G_2$. Let $H$ be the graph obtained by taking the disjoint union of $T_1$, $T_2$ and $T$ and then doing the following: introduce a new vertex $a$, connect $a$ to a leaf $T_1$ and to a leaf of $T_2$ using paths of length 2 and then make $a$ adjacent to $w$ (see Figure~\ref{fig:drawtree}(b)).
\medskip

\begin{figure}
	\centering
	\begin{tabular}{ccc}
		\includegraphics[page=14]{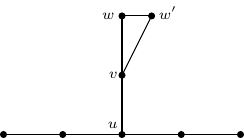}&\includegraphics[page=15]{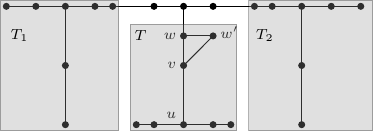}&\includegraphics[page=20]{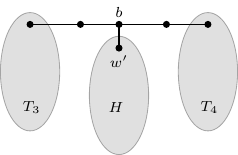}\\
		(a)&(b)&(c)
	\end{tabular}
	\caption{Construction of block graph for Proof of Theorem~\ref{thm:blocknot3sig}. (a) Construction of $T$. (b) Construction of $H$. $T_1$ and $T_2$ are isomorphic to $G_2$. (c) Construction of $G$. $T_3$ and $T_4$ are isomorphic to $G_3$ and $H$ is the block graph shown in (b). } 
	\label{fig:drawtree}
\end{figure}

\noindent\textit{Claim 1. $H$ is non-(2-SRIG).}
\smallskip

\noindent\textit{Proof.}
Note that $T-\{w\}$ is isomorphic to $G_2$, and hence is non-interval. As $T_1$, $T_2$, $T-\{w\}$ are asteroidal-(non-interval) in $H$, by Theorem~\ref{thm:asteroidal_k-1_sig_free}, we have that $H$ is non-(2-SRIG).

\medskip

It is easy to see that $H-\{w'\}$ is asteroidal-(non-interval)-free. Hence, by Theorem~\ref{thm:block2sig}, we have that $H-\{w'\}$ is 2-SRIG.
\medskip

\noindent\textit{Claim 2. The vertices $w$ and $v$ do not have a common stab in any 2-stabbed rectangle intersection representation of $H-\{w'\}$.}
\smallskip

\noindent\textit{Proof.}
Let $H'=H-\{w'\}$. Let $\mathcal{R}$ be any 2-stabbed rectangle intersection representation of $H'$. Since $T_1$ and $T_2$ are neighbour-disjoint connected induced subgraphs of $H'$ that are non-interval, we have that $|\mathcal{L}_{\mathcal{R}}(T_1)|=|\mathcal{L}_{\mathcal{R}}(T_2)|=2$. Let $P$ be the (induced) path between $T_1$ and $T_2$ in $H'$. Notice that $T-\{w,w'\}$ is a connected induced subgraph of $H'$ that is neighbour-disjoint from $T_1$, $T_2$ and $P$. Moreover, there is a path from the vertex $v$ of $T-\{w,w'\}$ to the vertex $a$ of $P$ that misses $T_1$ and $T_2$. We can now use Lemma~\ref{lem:useful} to conclude that $|\mathcal{L}_{\mathcal{R}}(T-\{w,w'\})|=1$. Let $\mathcal{L}_{\mathcal{R}}(T-\{w,w'\})=\{\ell\}$. It is clear that for each vertex of $T-\{w,w'\}$, and hence also for $v$, the only stab line that intersects the rectangle corresponding to it is $\ell$. If $r_w$ also intersects $\ell$, then the collection $\{\ell\cap r_x\}_{x\in V(T-\{w'\})}$ would form an interval representation of $G_2$, which contradicts the fact that $G_2$ is non-interval. This completes the proof of the claim.

\medskip

We shall now construct the desired block graph $G$ that satisfies the requirements in the statement of Theorem~\ref{thm:blocknot3sig}. Let $T_3$ and $T_4$ be trees that are isomorphic to $G_3$ (defined in Section~\ref{sec:gl}). Let $G'$ be the graph formed by taking the disjoint union of $T_3$ and $T_4$ and then doing the following: add a new vertex $b$ and connect it to a vertex of $T_3$ using a path of length 2 and a vertex of $T_4$ using a path of length 2. The graph $G$ is constructed by taking the disjoint union of $H$ and $G'$ and then adding an edge between $b$ and $w'$ (see Figure~\ref{fig:drawtree}(c) for a schematic diagram of $G$).
\medskip

\noindent\textit{Claim 3. $G$ is not 3-SRIG.}
\smallskip

\noindent\textit{Proof.} Suppose for the sake of contradiction that $G$ is 3-SRIG. Let $\mathcal{R}$ be a 3-stabbed rectangle intersection representation of $G$. Since $T_3$ and $T_4$ are neighbour-disjoint connected induced subgraphs of $G$ that are non-(2-SRIG) (recall that $T_3$ and $T_4$ are isomorphic to $G_3$ and that $G_3$ is non-(2-SRIG) by Lemma~\ref{lem:construct_G_l}\ref{it:not_l-1-SIG}), we have that $|\mathcal{L}_{\mathcal{R}}(T_3)|=|\mathcal{L}_{\mathcal{R}}(T_4)|=3$. Let $P$ be the path between $T_3$ and $T_4$ in $G$. Notice that $H-\{w'\}$ is a connected induced subgraph of $G$ that is neighbour-disjoint from $T_3$, $T_4$ and $P$. Moreover, there is a path from the vertex $w$ of $H-\{w'\}$ to the vertex $b$ of $P$ that misses $T_3$ and $T_4$. We can now use Lemma~\ref{lem:useful} to conclude that $|\mathcal{L}_{\mathcal{R}}(H-\{w'\})|=2$.
This means that in $\mathcal{R}$, the rectangles corresponding to $H-\{w'\}$ form a 2-stabbed rectangle intersection representation of $H-\{w'\}$. Then, by Claim 2, we know that neither of the two stab lines in $\mathcal{L}_{\mathcal{R}}(H-\{w'\})$ intersects both $r_w$ and $r_v$. Since $w'$ is adjacent to both $w$ and $v$, this implies that $r_{w'}$ intersects at least one of the two stab lines in $\mathcal{L}_{\mathcal{R}}(H-\{w'\})$. But then, the rectangles corresponding to the vertices of $H$, together with the stab lines in $\mathcal{L}_{\mathcal{R}}(H-\{w'\})$, form a 2-stabbed rectangle intersection representation of $H$. This contradicts Claim 1.
\medskip

To complete the proof of the theorem, we only need to show that $G$ is asteroidal-(non-2-SRIG)-free. Suppose for the sake of contradiction that there exist induced subgraphs $X_1,X_2,X_3$ that are asteroidal-(non-2-SRIG) in $G$. First we need the following claim, whose proof is left to the reader.

\begin{claim}
	In any block graph that contains three induced subgraphs that are asteroidal-$\mathcal{C}$ in it, for some graph class $\mathcal{C}$, there exists either a cutvertex that has no neighbour in each of the three subgraphs, or a triangle, whose removal results in a graph in which each of the three subgraphs is in a different component.
\end{claim}

From the above claim, we have that either there exists a vertex $x\in V(G)$ such that $G-\{x\}$ has three components $X'_1,X'_2,X'_3$ such that for each $i\in\{1,2,3\}$, $V(X_i)\subseteq V(X'_i)\setminus N[x]$, or $X_1,X_2,X_3$ are each contained in a different component of $G-\{w,w',v\}$ (since the only triangle in $G$ is formed by $w$, $w'$ and $v$). Let us first suppose that $X_1,X_2,X_3$ are each contained in a different component of the three components in $G-\{w,w',v\}$. It is easy to see that the component of $G-\{w,w',v\}$ that contains a neighbour of $v$ is a path and is therefore 1-SRIG, contradicting the fact that it contains one of the non-(2-SRIG) graphs $X_1,X_2,X_3$. So we can assume that there exists a vertex $x\in V(G)$ such that $G-x$ has three components $X'_1,X'_2,X'_3$ such that for each $i\in\{1,2,3\}$, $V(X_i)\subseteq V(X'_i)\setminus N[x]$. Note that since $G-\{x\}$ contains at least three components, degree of $x$ is at least 3 and $x\notin\{w,w',v\}$.

Let us first suppose that $x\in V(G')$. If $x=b$, one of the three components of $G-\{x\}$, say $X'_1$, is $H$. But now, $V(X'_1)\setminus N[x]=H-\{w'\}$, which is 2-SRIG by our earlier observation. This contradicts the fact that $V(X_1)\subseteq V(X'_1)\setminus N[x]$ as $X_1$ is non-(2-SRIG). If $x\neq b$, then $x\in V(T_3)$ or $x\in V(T_4)$. Suppose that $x\in V(T_3)$. As $G'$ is a tree, we know that $G'-\{x\}$ contains at least three components. Also, as $G'-V(T_3)$ has only one component, we can use Observation~\ref{obs:subtree}\ref{it:vertex} to conclude that all components of $G'-\{x\}$ except the component $Y$ that contains $b$ are proper subtrees of $T_3$. Since the only edge between $V(G)\setminus V(G')$ and $V(G')$ is $w'b$, we can see that every component of $G'-\{x\}$ other than $Y$ is also a component of $G-\{x\}$. This means that at least two components, say $X'_1,X'_2$, of $G-\{x\}$ are also components of $G'-\{x\}$. Since $V(X_1)\subseteq V(X'_1)$ and $V(X_2)\subseteq V(X'_2)$, we have that $X'_1$ and $X'_2$ are non-(2-SRIG) neighbour-disjoint induced subgraphs of $T_3$. As $T_3$ is isomorphic to $G_3$, this is a contradiction to Lemma~\ref{lem:construct_G_l}\ref{it:disjoint_nonexist}. For the same reason, we can also conclude that $x\notin V(T_4)$. This means that $x\in V(H)$.

But if $x\in V(H)$, then since $x\notin\{w,w',v\}$, it is clear from the construction of $G$ that at least one of the components, say $X'_1$, of $G-\{x\}$ is an induced subgraph of $H-\{w'\}$. As $H-\{w'\}$ is 2-SRIG by our earlier observation, this means that $X'_1$ is 2-SRIG, which contradicts the fact that it contains the non-(2-SRIG) graph $X_1$ as an induced subgraph. This shows that $G$ is asteroidal non-(2-SRIG)-free and hence completes the proof. \hfill\qed
\medskip

We shall now prove a general theorem that will later be used to prove Theorem~\ref{thm:treenotsuff}.

\begin{theorem}\label{thm:construct_not_k-SIG}
Let $k\geq 4$. For each $i\in\{k,k-1,k-2\}$, let $T_i,T'_i$ be two graphs that are $i$-SRIG but not $(i-1)$-SRIG and let $a_i\in V(T_i)$ and $a'_i\in V(T'_i)$. For $i\in\{k,k-1,k-2\}$, let $H_i$ be the graph obtained by adding a new vertex $b_i$ to the disjoint union of $T_i$ and $T'_i$ and connecting it to $a_i$ and $a'_i$ using paths of length at least two. Let $T$ be the graph obtained by adding a new vertex $c$ to the disjoint union of $H_k$, $H_{k-1}$ and $H_{k-2}$ and then connecting $c$ to each of $b_k$, $b_{k-1}$ and $b_{k-2}$ using paths of length at least two. Then $T$ is not $k$-SRIG.
\end{theorem}
\begin{proof}
Suppose for the sake of contradiction that $T$ is $k$-SRIG. Let $\mathcal{R}$ be a $(k+2)$-stabbed rectangle intersection representation of $T$ in which the top and bottom stab lines do not intersect any rectangle. Let $A$ be a good region that contains all the rectangles of $\mathcal{R}$, i.e., $T_A=T$ (note that such a region exists; we can consider a rectangle with top and bottom edges on the top and bottom stab lines such that it contains all the rectangles of $\mathcal{R}$). As the only two stab lines that are not intersected by any rectangle in $\mathcal{R}$ are the top and bottom stab lines (recall that $T_A$ is not $(k-1)$-SRIG as it contains $T_k$ and $T'_k$), it follows that $|\mathcal{L}_{\mathcal{R}}(A)|=k$. As $T_k$ and $T'_k$ are $k$-SRIG but not $(k-1)$-SRIG, we know that there are $\mathcal{L}_{\mathcal{R}}(A)$-spanning paths in each of them. Let $X_1$ and $X_2$ be minimal $\mathcal{L}_{\mathcal{R}}(A)$-spanning paths in $T_k$ and $T'_k$ respectively. It is easy to see that $X_1$ and $X_2$ are neighbour-disjoint. Let $X$ be an induced path in $T_A$ that connects some vertex of $X_1$ and some vertex of $X_2$ such that no internal vertex of $X$ belongs to either $X_1$ or $X_2$. Note that $X$ is a subgraph of $H_k$ that contains $b_k$. Let $(A_t,A_b)=\Delta(\mathcal{R},A,X_1,X_2,X)$.
	
	(+) By Observation~\ref{obs:hitsP}\ref{it:xhitsP}, if for $x\in V(T)$, the rectangle $r_x$ intersects $\mathbf{bottom}(A_t)$, then $x$ has a neighbour on $X$.
	
	Since there is a path in $T_A$ from $c\in V(T_A)$ to a vertex in $X$ (in this case, $b_k$) that misses both $X_1$ and $X_2$, we know by Lemma~\ref{lem:rtorrb} that $r_c$ is contained in $A_t$ or $A_b$. We shall assume without loss of generality that $r_c$ is contained in $A_t$ (see Figure~\ref{fig:recursiveregion}(a)). Let us define $B=A_t$. Let $T^*$ be the graph obtained by removing the vertices in $V(H_k)$ and their neighbours from $T$, or in other words, $T^*=T-(V(H_k)\cup N[b_k])$. Note that there is a path in $T_A$ from $c$ to each vertex of $T^*$ that misses $X_1$, $X_2$ and $X$. We can now infer using Lemma~\ref{lem:winrtorrb} that the rectangles corresponding to the vertices in $T^*$ are all contained in $A_t=B$. In other words, $T^*$ is a connected induced subgraph of $T_B$.

\begin{figure}
	\centering
	\begin{tabular}{ccc}
		\includegraphics[page=17]{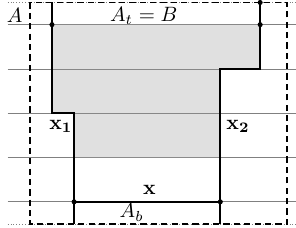}&\includegraphics[page=18]{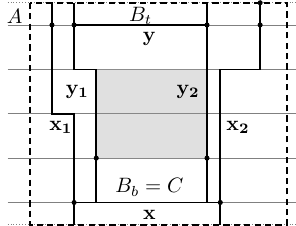}&\includegraphics[page=19]{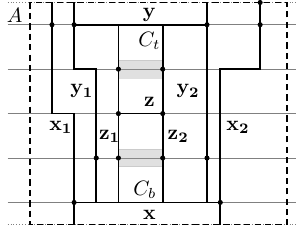}\\
		(a)&(b)&(c)
	\end{tabular}
	\caption{An illustration of various stages of the proof of Theorem~\ref{thm:construct_not_k-SIG}. The region bounded by the dashed curve is $A$. The solid curves represent the rectilinear curves through paths chosen in the proof to split the regions. For example, the solid curve labelled $\mathbf{x_1}$ is the rectilinear curve through the path $X_1$, the solid curve labelled $\mathbf{y}$ is the rectilinear curve through the path $Y$ and so on. The shaded region indicates the possible locations of the rectangle $r_c$ as the proof proceeds.}
	\label{fig:recursiveregion}
\end{figure}
	
	Since $T^*$ contains $T_{k-1}$ and $T'_{k-1}$ as induced subgraphs, and is therefore not $(k-2)$-SRIG, we have $|\mathcal{L}_{\mathcal{R}}(B)|\geq k-1$. By Lemma~\ref{lem:goodregion}, this means that $B=A_t$ is a good region. Since $B$ does not contain the bottom stab line in $\mathcal{L}_{\mathcal{R}}(A)$ by Lemma~\ref{lem:bottomstab}\ref{it:decreasestab}, we can conclude that $|\mathcal{L}_{\mathcal{R}}(B)|=k-1$. Now, $T_{k-1}$ and $T'_{k-1}$ are two neighbour-disjoint subgraphs of $T^*$ that are $(k-1)$-SRIG but not $(k-2)$-SRIG. Since the rectangles corresponding to the vertices in them are all contained in $B$ (recall that $T^*$ is an induced subgraph of $T_B$), there is at least one vertex of $T_{k-1}$ and at least one vertex of $T'_{k-1}$ on every stab line in $\mathcal{L}_{\mathcal{R}}(B)$. This means that there exist minimal $\mathcal{L}_{\mathcal{R}}(B)$-spanning paths $Y_1$ in $T_{k-1}$ and $Y_2$ in $T'_{k-1}$, and it is clear that $Y_1$ and $Y_2$ are neighbour-disjoint. Let $Y$ be an induced path in $T^*$ that connects some vertex of $Y_1$ and some vertex of $Y_2$ such that no internal vertex of $Y$ belongs to either $Y_1$ or $Y_2$. Note that $Y$ is a subgraph of $H_{k-1}$ that contains $b_{k-1}$. Let $(B_t,B_b)=\Delta(\mathcal{R},B,Y_1,Y_2,Y)$.
	
	(++) By Observation~\ref{obs:hitsP}\ref{it:xhitsP}, if for $x\in V(T)$, the rectangle $r_x$ intersects $\mathbf{top}(B_b)$, then $x$ has a neighbour on $Y$.
	
	Since there is a path in $T^*$ from $c$ to a vertex in $Y$ (in this case, $b_{k-1}$) that misses both $Y_1$ and $Y_2$, we know by Lemma~\ref{lem:rtorrb} that $r_c$ is contained in $B_t$ or $B_b$. Suppose that $r_c$ is contained in $B_t$. Note that the path $Q$ in $T$ between $c$ and $b_k$ misses $Y_1$, $Y_2$ and $Y$. As $b_k$ lies on the path $X$, we know by Observation~\ref{obs:hitsP}\ref{it:xinP} that $r_{b_k}$ intersects $\mathbf{bottom}(B)$. This means that $r_{b_k}$ contains some points from outside $B$ and hence some points from outside $B_t$. Since $r_c$ is contained in $B_t$, this can only mean that there exists some vertex $x$ in $Q$ such that the rectangle $r_x$ intersects the boundary of $B_t$. Since $x$ has no neighbour on $Y_1$, $Y_2$ or $Y$, we know by Lemma~\ref{lem:crossboundary} that $r_x$ intersects $\mathbf{top}(B_t)$. Since $B=A_t$ and $A$ are good regions, we have by Observation~\ref{obs:hitsP}\ref{it:tt'bb'} that $\mathbf{top}(B_t)\subseteq\mathbf{top}(A_t)\subseteq\mathbf{top}(A)$. This implies that $r_x$ intersects the boundary of $A$, which is a contradiction to the fact that $T=T_A$ (or in other words, all rectangles corresponding to vertices of $T$ are contained in $A$). Thus, we can conclude that $r_c$ is not contained in $B_t$, and hence is contained in $B_b$ (See Figure~\ref{fig:recursiveregion}(b)). Let us define $C=B_b$.
	
	Let $T^{**}$ be the graph obtained by removing the vertices in $V(H_{k-1})$ and their neighbours from $T^*$, or in other words, $T^{**}=T^*-(V(H_{k-1})\cup N[b_{k-1}])$. Note that $c\in V(T^{**})$ and that there is a path in $T^*$ from $c$ to each vertex of $T^{**}$ that misses $Y_1$, $Y_2$ and $Y$. We can now infer using Lemma~\ref{lem:winrtorrb} that the rectangles corresponding to the vertices in $T^{**}$ are all contained in $C$. In other words, $T^{**}$ is a connected induced subgraph of $T_C$.
	
	Since $T^{**}$ contains $T_{k-2}$ and $T'_{k-2}$ as induced subgraphs, and is therefore not $(k-3)$-SRIG, we have $|\mathcal{L}_{\mathcal{R}}(C)|\geq k-2$. By Lemma~\ref{lem:goodregion}, this means that $C$ is a good region. Since $C$ does not contain the top stab line in $\mathcal{L}_{\mathcal{R}}(B)$ by Lemma~\ref{lem:bottomstab}\ref{it:decreasestab}, we can conclude that $|\mathcal{L}_{\mathcal{R}}(C)|=k-2$. Now, $T_{k-2}$ and $T'_{k-2}$ are two neighbour-disjoint subgraphs of $T^{**}$ that are $(k-2)$-SRIG but not $(k-3)$-SRIG. Since $T^{**}$ is an induced subgraph of $T_C$, at least one vertex of $T_{k-2}$ and at least one vertex of $T'_{k-2}$ are on every stab line in $\mathcal{L}_{\mathcal{R}}(C)$. This means that there exist minimal $\mathcal{L}_{\mathcal{R}}(C)$-spanning paths $Z_1$ in $T_{k-2}$ and $Z_2$ in $T'_{k-2}$, which are neighbour-disjoint. Let $Z$ be an induced path in $T^{**}$ that connects some vertex of $Z_1$ and some vertex of $Z_2$ such that no internal vertex of $Z$ belongs to either $Z_1$ or $Z_2$. Note that $Z$ is a subgraph of $H_{k-2}$ that contains $b_{k-2}$. Let $(C_t,C_b)=\Delta(\mathcal{R},C,Z_1,Z_2,Z)$.
	
	Since there is a path in $T^{**}$ from $c$ to a vertex in $Z$ (in this case, $b_{k-2}$) that misses both $Z_1$ and $Z_2$, we know by Lemma~\ref{lem:rtorrb} that $r_c$ is contained in $C_t$ or $C_b$ (See Figure~\ref{fig:recursiveregion}(c)). Suppose that $r_c$ is contained in $C_t$. Note that the path $Q$ in $T$ between $c$ and $b_k$ misses $Z_1$, $Z_2$, $Z$ and $Y$. As $b_k$ lies on the path $X$, we know by Observation~\ref{obs:hitsP}\ref{it:xinP} that $r_{b_k}$ intersects $\mathbf{bottom}(B)$. This means that $r_{b_k}$ contains some points from outside $B$, and hence some points from outside $C_t$. Since $r_c$ is contained in $C_t$, this can only mean that there exists some vertex $x$ in $Q$ such that the rectangle $r_x$ intersects the boundary of $C_t$. Since $x$ has no neighbour on $Z_1$, $Z_2$ or $Z$, we know by Lemma~\ref{lem:crossboundary} that $r_x$ intersects $\mathbf{top}(C_t)$. Since $C=B_b$ is a good region, we have by Observation~\ref{obs:hitsP}\ref{it:tt'bb'} that $\mathbf{top}(C_t)\subseteq\mathbf{top}(B_b)$, implying that $r_x$ intersects $\mathbf{top}(B_b)$. By (++), we now have that $x$ has a neighbour on $Y$, which is a contradiction to the fact that $Q$ misses $Y$. This means that $r_c$ is contained in $C_b$.
	
	Now consider the path $Q$ in $T$ between $c$ and $b_{k-1}$. It is clear that $Q$ misses $Z_1$, $Z_2$, $Z$ and $X$. As $b_{k-1}$ lies on the path $Y$, we know by Observation~\ref{obs:hitsP}\ref{it:xinP} that $r_{b_{k-1}}$ intersects $\mathbf{top}(C)$. This means that $r_{b_{k-1}}$ contains some points from outside $C$, and hence some points from outside $C_b$. Since $r_c$ is contained in $C_b$, this can only mean that there exists some vertex $x$ in $Q$ such that the rectangle $r_x$ intersects the boundary of $C_b$. Since $x$ has no neighbour on $Z_1$, $Z_2$ or $Z$, we know by Lemma~\ref{lem:crossboundary} that $r_x$ intersects $\mathbf{bottom}(C_b)$. Since $C=B_b$ and $B=A_t$ are good regions, we have by Observation~\ref{obs:hitsP}\ref{it:tt'bb'} that $\mathbf{bottom}(C_b)\subseteq\mathbf{bottom}(B_b)\subseteq\mathbf{bottom}(A_t)$, implying that $r_x$ intersects $\mathbf{bottom}(A_t)$. By (+), we now have that $x$ has a neighbour on $X$, which is a contradiction to the fact that $Q$ misses $X$. This completes the proof.
\end{proof}

\medskip

\noindent\textbf{\textit{Proof of Theorem~\ref{thm:treenotsuff}.}}
\medskip

Let $k$ be any integer greater than or equal to 4. For each $i\in\{k,k-1,k-2\}$, let $T_i,T'_i$ be two rooted trees that are each isomorphic to $G_i$ (defined in Section~\ref{sec:gl}). From Lemma~\ref{lem:construct_G_l}\ref{it:not_l-1-SIG} and Lemma~\ref{lem:construct_G_l}\ref{it:l-SIG} we know that $T_i$ and $T'_i$ are $i$-SRIG but not $(i-1)$-SRIG. Let $a_i= root(T_i)$ and $a'_i=root(T'_i)$. For $i\in\{k,k-1,k-2\}$, let $H_i$ be the tree obtained by adding a new vertex $b_i$ to the disjoint union of $T_i$ and $T'_i$ and connecting it to $a_i$ and $a'_i$ using paths of length two. Note that $H_i$ is isomorphic to $F_i$ (also defined in Section~\ref{sec:gl}). Let $T$ be the tree obtained by adding a new vertex $c$ to the disjoint union of $H_k$, $H_{k-1}$ and $H_{k-2}$ and then connecting $c$ to each of $b_k$, $b_{k-1}$ and $b_{k-2}$ using paths of length at least two. See Figure~\ref{fig:treenotsuff} for a schematic diagram of $T$. From Theorem~\ref{thm:construct_not_k-SIG}, we know that $T$ is not $k$-SRIG. 

\begin{figure}
	\centering
	\includegraphics[page=21]{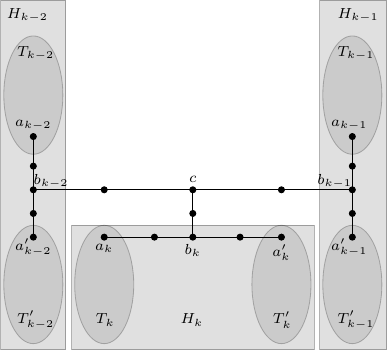}
	\caption{A schematic diagram of $T$. For each $i\in\{k,k-1,k-2\}$, let $T_i,T'_i$ be two rooted trees that are each isomorphic to $G_i$ (defined in Section~\ref{sec:gl}) and rooted at $a_i$ and $a'_i$ respectively.}\label{fig:treenotsuff}
\end{figure}

We now show that $T$ is asteroidal-(non-($k-1$)-ESRIG)-free. For the sake of contradiction, assume that there are three subtrees $X_1,X_2,X_3$ that are asteroidal-(non-$(k-1)$-ESRIG) in $T$. The following claim is easy to see.

\begin{claim}
	There is a vertex $v$ in $T$ of degree at least 3 such that $T-\{v\}$ contains three components $X'_1,X'_2,X'_3$ where for each $i\in\{1,2,3\}$, $X_i$ is an induced subtree of $X'_i-N[v]$. 
\end{claim}

Let $v$ be the vertex in $T$ of degree at least 3 such that $T-\{v\}$ contains three components $X'_1,X'_2,X'_3$ where for each $i\in\{1,2,3\}$, $X_i$ is an induced subtree of $X'_i-N[v]$. For each $i\in\{1,2,3\}$, since $X_i$ is non-$(k-1)$-ESRIG, we also have that $X'_i$ is non-$(k-1)$-ESRIG. Let us assume that $v$ is a vertex of $T_k$. Note that $T-V(T_k)$ has only one component. Then by Observation~\ref{obs:subtree}\ref{it:vertex}, all but one component of $T-\{v\}$ are proper subtrees of $T_k$. This implies that there exist distinct $X,Y\in\{X'_1,X'_2,X'_3\}$ such that $X,Y$ are proper subtrees of $T_k$. Therefore, $X$ and $Y$ are vertex-disjoint (in fact, neighbour-disjoint) subtrees of $T_k$ that are both non-($k-1$)-ESRIG. But since $T_k$ is isomorphic to $G_k$, this is a contradiction to Lemma~\ref{lem:construct_G_l}\ref{it:disjoint_nonexist}. Hence, $v$ is not a vertex of $T_k$ and for similar reasons, $v$ is not a vertex of $T'_k$. Let $T^*=T-(V(H_k)\cup N[b_k])$.

\begin{claim}
The tree $T^*$ is $(k-1)$-ESRIG.
\end{claim}

\noindent\textit{Proof.} From the definition of $T$, we know that $T^*$ is the union of $H_{k-1}$, $H_{k-2}$ and the path in $T$ between $b_{k-1}$ and $b_{k-2}$ (which contains the vertex $c$). Recall that $H_{k-1}$ is obtained by adding a new vertex $b_{k-1}$ to the disjoint union of $T_{k-1}$ and $T'_{k-1}$ and connecting their roots (i.e. $a_{k-1}$ and $a'_{k-1}$ respectively) to $b_{k-1}$ using paths of length two. Therefore, $H_{k-1}$ is isomorphic to $F_{k-1}$. Let $\mathcal{R}_1$ be the $(k-1)$-exactly stabbed rectangle intersection representation of $H_{k-1}$ that is given by Lemma~\ref{lem:construct_G_l}\ref{it:path_rep1}. Similarly, $H_{k-2}$ is isomorphic to $F_{k-2}$. Let $\mathcal{R}_2$ be the $(k-2)$-exactly stabbed rectangle intersection representation of $H_{k-2}$ that is given by Lemma~\ref{lem:construct_G_l}\ref{it:path_rep2}, in which the only vertices on the top stab line are those in $N[a_{k-2}]\cup N[a'_{k-2}]$. It can now be seen that the two representations $\mathcal{R}_1$ and $\mathcal{R}_2$ can be combined as shown in Figure~\ref{fig:treecombine} to obtain a $(k-1)$-exactly stabbed rectangle intersection representation $\mathcal{R}$ of $T^*[V(H_{k-1})\cup V(H_{k-2})]$ that satisfies the following properties: (i) all vertices of the path between $a_{k-1}$ and $a'_{k-1}$ are on the top stab line of $\mathcal{R}$, (ii) a vertex $u\in V(H_{k-2})$ is on the stab line just below the top stab line of $\mathcal{R}$ if and only if $u\in N[a_{k-2}]\cup N[a'_{k-2}]$, and (iii) for any vertex $u\in V(H_{k-2})$, we have that $span(u)\subset span(b_{k-1})$. We leave it to the reader to verify that $\mathcal{R}$ can be extended to a $(k-1)$-exactly stabbed rectangle intersection representation of $T^*$ by adding the rectangles corresponding to the three vertices in the path between $b_{k-1}$ and $b_{k-2}$ (refer to Figure~\ref{fig:treecombine}). Therefore we conclude that $T^*$ is $(k-1)$-ESRIG.

\begin{figure}
	\centering
	\includegraphics[page=13]{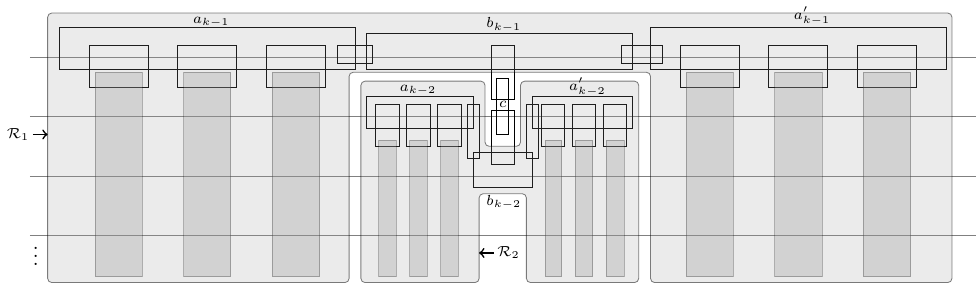}
	\caption{A schematic diagram of the $(k-1)$-stabbed rectangle intersection representation $\mathcal{R}$ of $T^*$.}\label{fig:treecombine}
\end{figure}
\medskip

Now suppose $v$ is a vertex of $H_k$. Since we have already concluded that $v\notin V(T_k)\cup V(T'_k)$, we can infer that $v$ must be the vertex $b_k$. Recalling the definition of $T$, we can infer that $T-\{b_k\}$ has exactly three components and since $b_k=v$ we know that they are $X'_1,X'_2,X'_3$. Also from the definition of $T$, it follows that there exists $i\in\{1,2,3\}$ such that $X'_i=T-V(H_k)$. We know from the definition of $v$ that $X_i$ is a subtree of $X'_i-N[v]=T^*$. But then by the above claim, we have that $X_i$ is $(k-1)$-ESRIG, which contradicts the fact that $X_i$ is non-$(k-1)$-ESRIG.

From the above arguments, we infer that $v$ must lie in the tree $T-V(H_k)$. Since $v$ has degree at least 3, we can infer from the construction of $T$ that $v\in V(T^*)$. Notice that $T-V(T^*)$ has only one component. Then by Observation~\ref{obs:subtree}\ref{it:vertex}, all but one component of $T-\{v\}$ are proper subtrees of $T^*$. This implies that there is a component $X\in\{X'_1,X'_2,X'_3\}$ such that $X$ is a proper subtree of $T^*$. But by the above claim, we now have that $X$ is $(k-1)$-ESRIG, contradicting our earlier observation that $X'_1,X'_2,X'_3$ are all non-$(k-1)$-ESRIG. This completes the proof.\hfill\qed

\subsection{Trees that are $k$-SRIG but not $k$-ESRIG}\label{sec:notesrig}
We define the tree $D_l$, for $l>1$, as follows. Let $T_1,T_2,\ldots,T_7$ be seven rooted trees, each isomorphic to $G_{l-1}$. Take a $K_{1,7}$ with vertex set $\{u,u_1,u_2,\ldots,u_7\}$, where $u_1,u_2,\ldots,u_7$ are the leaves, and add edges between $u_i$ and $root(T_i)$ for each $i\in\{1,2,\ldots,7\}$. The resulting graph is $D_l$ and we let $root(D_l)=u$.

\begin{lemma}\label{lem:construct_D_l}
Let $l>1$.
\begin{enumerate}
	\vspace{-0.075in}
	\itemsep 0in
	\renewcommand{\theenumi}{(\roman{enumi})}
	\renewcommand{\labelenumi}{\theenumi}
	\item\label{it:Dl-not_l-1-SRIG} $D_l$ is not $(l-1)$-SRIG.
	\item\label{it:Dl-SRIG} There is an $l$-exactly stabbed rectangle intersection representation $\mathcal{R}$ of $D_l$ such that for $v,w\in V(D_l)$, $span(v)\subseteq span(w)$ if $w$ is an ancestor of $v$ and the rectangles intersecting the top stab line of $\mathcal{R}$ are exactly the vertices in $N[root(D_l)]$.
	\item Let $T$ and $T'$ be two trees each isomorphic to $D_l$. Let $J_l$ be the tree obtained by taking a new vertex $u$ and joining it to the root vertices of $T,T'$ using paths of length two.
	\begin{enumerate}
		\vspace{-0.075in}
		\itemsep 0in
		\renewcommand{\theenumii}{(\alph{enumii})}
		\renewcommand{\labelenumii}{\theenumii}
		\item\label{it:Dl-pathjoin} There is an $l$-exactly stabbed rectangle intersection representation $\mathcal{R'}$ of $J_l$ such that for $v,w\in V(J_l)$, $span(v)\subseteq span(w)$ if $w$ is an ancestor of $v$ in $T$ or $T'$, and all vertices in the path between $root(T)$ and $root(T')$ are on the top stab line of $\mathcal{R'}$.
		\item\label{it:J_l-unique_rep} If $l\geq 6$, then in any $l$-exactly stabbed rectangle intersection representation of $J_l$, $root(T)$ and $root(T')$ are either both on the top stab line or both on the bottom stab line.
	\end{enumerate} 
	\item\label{it:unique_rep} In any $l$-exactly stabbed rectangle intersection representation $\mathcal{R}$ of $D_l$, $root(D_l)$ is on the top or bottom stab line of $\mathcal{R}$.
\end{enumerate}
\end{lemma}
\begin{proof}
For~\ref{it:Dl-not_l-1-SRIG}, it is easy to see that $G_l$ is an induced subgraph of $D_l$, and therefore by Lemma~\ref{lem:construct_G_l}\ref{it:not_l-1-SIG}, $D_l$ is not $(l-1)$-SRIG.
It is also easy to see that the constructions in the proofs of Lemma~\ref{lem:construct_G_l}\ref{it:l-SIG} and Lemma~\ref{lem:construct_G_l}\ref{it:path_rep1} can be easily extended to prove~\ref{it:Dl-SRIG} and~\ref{it:Dl-pathjoin} respectively.

We shall now prove~\ref{it:unique_rep}. Suppose for the sake of contradiction that there exists an $l$-exactly stabbed rectangle intersection representation $\mathcal{R}$ of $D_l$ in which $root(D_l)$ is not on the top or bottom stab lines. Recall that $D_l$ is constructed by taking a $K_{1,7}$ with vertex set $\{u,u_1,u_2,\ldots,u_7\}$ with leaves $u_1,u_2,\ldots,u_7$ and making each $u_i$ adjacent to the root of a tree $T_i$ that is isomorphic to $G_{l-1}$. For each $i\in\{1,2,\ldots,7\}$, let $T'_i=D_l[\{u_i\}\cup V(T_i)]$. Suppose that there exists $I\subseteq\{1,2,\ldots,7\}$ with $|I|=3$ such that for each $i\in I$, there is no vertex in $T'_i$ that is on the top stab line. Then, since $u=root(D_l)$ is not on the top stab line, the rectangles corresponding to the vertices of $\{u\}\cup\bigcup_{i\in I} V(T'_i)$ form an $(l-1)$-(exactly) stabbed rectangle intersection representation of a tree isomorphic to $G_l$. This contradicts Lemma~\ref{lem:construct_G_l}\ref{it:not_l-1-SIG}. Therefore, there are at most two trees in $\{T'_1,T'_2,\ldots,T'_7\}$ such that none of their vertices are on the top stab line. In similar fashion, we can conclude that there are at most two trees in $\{T'_1,T'_2,\ldots,T'_7\}$ such that none of their vertices are on the bottom stab line. This means that there are at least three trees in $\{T'_1,T'_2,\ldots,T'_7\}$, say $T'_1,T'_2,T'_3$, such that $|\mathcal{L}_{\mathcal{R}}(T'_1)|=|\mathcal{L}_{\mathcal{R}}(T'_2)|=|\mathcal{L}_{\mathcal{R}}(T'_3)|=l$. For $i\in\{1,2,3\}$, let $P_i$ be an $\mathcal{L}_{\mathcal{R}}(T'_i)$-spanning induced path in $T'_i$ starting at a vertex $x_i$ that is on the top stab line and ending at a vertex $y_i$ that is on the bottom stab line. Let $\mathbf{p_i}$ be a rectilinear curve through $P_i$ starting at some point on the top stab line in $r_{x_i}$ and ending at some point on the bottom stab line in $r_{y_i}$. As $T'_1,T'_2,T'_3$ are pairwise neighbour-disjoint, we know that $P_1,P_2,P_3$ are also pairwise neighbour-disjoint, implying that the curves $\mathbf{p_1},\mathbf{p_2},\mathbf{p_3}$ are pairwise disjoint. Therefore one of the curves, say $\mathbf{p_2}$, is between the other two. Then, it is easy to see that any path between a vertex of $T'_1$ and a vertex of $T'_3$ contains a vertex whose rectangle intersects $\mathbf{p_2}$, which means that this vertex has a neighbour on $P_2$. Now consider the path $u_1uu_3$. As the only vertex on this path that has a neighbour in $V(T_2)$ is $u=root(D_l)$, we can infer that $r_u$ intersects $\mathbf{p_2}$. It follows from the definition of rectilinear curves that there is a point $q\in r_u\cap\mathbf{p_2}$ that is also on a stab line, say $\ell$. As $u=root(D_l)$ is on $\ell$, we can conclude that $\ell$ is neither the top nor the bottom stab line of $\mathcal{R}$. Since the point $q\in\mathbf{p_2}$, it belongs to the rectangle corresponding to a vertex on $P_2$ that intersects $r_u$. Note that if $u$ has a neighbour on $P_2$, then it has to be $u_2$. This lets us conclude that $u_2$ is on $P$ and also that $q\in r_{u_2}$, which implies that $u_2$ is on $\ell$. As $\mathcal{R}$ is an $l$-exactly stabbed rectangle intersection representation, we infer that $u_2$ is neither on the top nor the bottom stab line. Then, $u_2\notin\{x_i,y_i\}$. But this means that $x_i,y_i\in V(T_i)$, implying that the path $P_2$ does not contain $u_2$. This contradicts our earlier observation that $u_2$ is on $P_2$.

It only remains to prove~\ref{it:J_l-unique_rep}. Let $l\geq 6$ and let $\mathcal{R}$ be any $l$-exactly stabbed rectangle intersection representation of $J_l$. Let $\ell,\ell'$ be the stab lines that intersect $r_{root(T)}$ and $r_{root(T')}$ respectively. By~\ref{it:unique_rep}, we know that each of $\ell,\ell'$ is either the top stab line or the bottom stab line. Since there is a path of length 4 between $root(T)$ and $root(T')$ in $J_l$, we can infer that $\ell$ and $\ell'$ have no more than 3 stab lines between them. Since $l\geq 6$, this means that it is not possible that one of $\ell,\ell'$ is the top stab line and the other the bottom stab line. So $\ell,\ell'$ are either both the top stab line or both the bottom stab line.
\end{proof}

\begin{lemma}\label{lem:rootinbottom}
Let $\mathcal{R}$ be a $k$-exactly stabbed rectangle intersection representation of a graph $G$ and let $R$ be a good region in this representation. Let $P_1$ and $P_2$ be minimal $\mathcal{L}_{\mathcal{R}}(R)$-spanning paths in $G_R$ that are neighbour-disjoint and let $P$ be an induced path in $G_R$ between some vertex in $V(P_1)$ and some vertex in $V(P_2)$ such that no internal vertex of $P$ is on $P_1$ or $P_2$. Let $(R_t,R_b)=\Delta(\mathcal{R},R,P_1,P_2,P)$. Suppose that there are two nonadjacent vertices $x_1,x_2\in V(P)$ that are on the top (bottom) stab line in $\mathcal{L}_{\mathcal{R}}(R)$ such that the subpath $P'$ of $P$ between $x_1$ and $x_2$ has length at most $d$, for some $d\geq 2$. Then there does not exist a connected induced subgraph $H$ of $G_{R_t}$ ($G_{R_b}$) which is neighbour-disjoint from $P_1,P_2,P$ and satisfies the following properties:
\begin{enumerate}
\vspace{-0.075in}
\itemsep 0in
\renewcommand{\theenumi}{(\roman{enumi})}
\renewcommand{\labelenumi}{\theenumi}
\item $|\mathcal{L}_{\mathcal{R}}(H)|>\left\lceil\frac{d-1}{2}\right\rceil$, and
\item $H$ contains a vertex $c$ such that there exists a path in $G_R$ from $c$ to some vertex in $P'$ that misses $x_1$, $x_2$, $P_1$, $P_2$ and $P-V(P')$.
\end{enumerate}
\end{lemma}
\begin{proof}
We shall prove the lemma only for the case when $x_1$ and $x_2$ are on the top stab line in $\mathcal{L}_{\mathcal{R}}(R)$, as the other case can be proved in similar fashion. 
Suppose there exists a connected component $H$ of $G_{R_t}$ that is neighbour-disjoint from $P_1$, $P_2$ and $P$ such that $|\mathcal{L}_{\mathcal{R}}(H)|>\left\lceil\frac{d-1}{2}\right\rceil$, and there exists $c\in V(H)$ from which there is a path $Q$ in $G_R$ to some vertex in $P'$ that misses $x_1$, $x_2$, $P_1$, $P_2$ and $P-V(P')$. For $i\in\{1,2\}$, let $u_i,v_i$ be the endvertices of $P_i$ on the top and bottom stab lines in $\mathcal{L}_{\mathcal{R}}(R)$ respectively, and let $V(P_i)\cap V(P)=\{w_i\}$. Let us assume without loss of generality that $x_1$ appears before $x_2$ when traversing the path $P$ from $w_1$ to $w_2$. For $i\in\{1,2\}$, define $P'_i$ to be the path obtained by the union of the subpath of $P_i$ between $v_i$ and $w_i$ and the subpath of $P$ between $w_i$ and $x_i$. It is clear that $P'_1$ and $P'_2$ are neighbour-disjoint $\mathcal{L}_{\mathcal{R}}(R)$-spanning paths in $G_R$ and that $P'$ is an induced path in $G_R$ between a vertex in $V(P'_1)$ and a vertex in $V(P'_2)$ none of whose internal vertices are on either $P'_1$ or $P'_2$. Let $(R'_t,R'_b)=\Delta(\mathcal{R},R,P'_1,P'_2,P')$. As $\mathcal{R}$ is a $k$-exactly stabbed rectangle intersection representation and $P'$ has length $d$, it follows that $|\mathcal{L}_{\mathcal{R}}(R'_t)|\leq\left\lceil\frac{d-1}{2}\right\rceil$.

Since $P'$ misses $c$ and there is the path $Q$ in $G_R$ between $c$ and a vertex of $P'$ that misses both $P'_1$ and $P'_2$, we can apply Lemma~\ref{lem:rtorrb} to conclude that $r_c$ is contained in $R'_t$ or $R'_b$. It is easy to see that for any vertex $z$ that misses $P_1$, $P_2$ and $P$, the rectangle $r_z$ is contained in $R'_b$ if and only if it is contained in $R_b$. As we know that $c\in V(G_{R_t})$, which implies that $r_c$ is contained in $R_t$ and therefore not in $R_b$, we can now conclude that $r_c$ is contained in $R'_t$. Since $H$ is neighbour-disjoint from $P_1$, $P_2$ and $P$, it is also neighbour-disjoint from $P'_1$, $P'_2$ and $P'$. As $H$ is connected, this means that there is a path in $G_R$ from $c$ to each vertex of $H$ that misses $P'_1$, $P'_2$ and $P'$. By Lemma~\ref{lem:winrtorrb}, we now have that $H$ is an induced subgraph of $G_{R'_t}$. This means that $|\mathcal{L}_{\mathcal{R}}(R'_t)|>\left\lceil\frac{d-1}{2}\right\rceil$, contradicting our earlier observation.
\end{proof}

\begin{theorem}\label{thm:srig-parameter-diff}
For every $k\geq 10$, there is a tree which is $k$-SRIG but not $k$-ESRIG.
\end{theorem}
\begin{proof}
Let $k$ be any integer greater than or equal to 10. For each $i\in\{k,k-1\}$, let $T_i,T'_i$ be two rooted trees that are each isomorphic to $D_i$ and let $T_{k-2}$ be a tree isomorphic to $D_{k-2}$. From Lemma~\ref{lem:construct_D_l}\ref{it:Dl-not_l-1-SRIG} and Lemma~\ref{lem:construct_D_l}\ref{it:Dl-SRIG}, we know that for $i\in\{k,k-1\}$, $T_i$ and $T'_i$ are $i$-SRIG but not $(i-1)$-SRIG. For $i\in\{k,k-1\}$, let $a_i= root(T_i)$ and $a'_i=root(T'_i)$. Further, let $H_i$ be the tree obtained by adding a new vertex $b_i$ to the disjoint union of $T_i$ and $T'_i$ and connecting it to $a_i$ and $a'_i$ using paths of length two. Let $a_{k-2}=root(T_{k-2})$. Let $T$ be the tree obtained by adding a new vertex $c$ to the disjoint union of $H_k$, $H_{k-1}$ and $T_{k-2}$ and then connecting $c$ to each of $b_k$, $b_{k-1}$ and $a_{k-2}$ using paths of length two. See Figure~\ref{fig:treenotksrig} for a schematic diagram of $T$. We claim that $T$ is $k$-SRIG but not $k$-ESRIG.

\begin{figure}
	\centering
	\includegraphics[page=23]{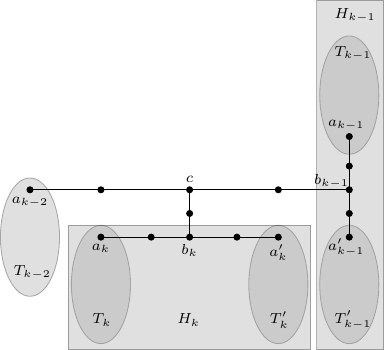}
	\caption{A schematic diagram of $T$. For each $i \in \{k, k-1\}$, let $T_i,T'_i$ two rooted trees that are each isomorphic to $D_i$ and rooted at $a_i$ and $a_i$ respectively. $T_{k-2}$ is isomorphic to $D_{k-2}$ and is rooted at $a_{k-2}$.}\label{fig:treenotksrig}
\end{figure}

We will first show that $T$ is $k$-SRIG. Let $\ell_1,\ell_2,\ldots,\ell_k$ be $k$ horizontal lines, ordered from bottom to top. Since $H_k$ is isomorphic to $J_k$, we know from Lemma~\ref{lem:construct_D_l}\ref{it:Dl-pathjoin} that there is a $k$-(exactly) stabbed rectangle intersection representation $\mathcal{R}_1$ of $H_k$ using stab lines $\ell_1,\ell_2,\ldots,\ell_k$ such that for $v,w\in V(H_k)$, $span(v)\subseteq span(w)$ if $w$ is an ancestor of $v$ in $T_k$ or $T'_k$, and all vertices in the path in $T$ between $a_k$ and $a'_k$ are on the bottom stab line $\ell_1$. Similarly, there is a $(k-1)$-(exactly) stabbed rectangle intersection representation $\mathcal{R}_2$ of $H_{k-1}$ using stab lines $\ell_2,\ell_3,\ldots,\ell_k$ such that for $v,w\in V(H_{k-1})$, $span(v)\subseteq span(w)$ if $w$ is an ancestor of $v$ in $T_{k-1}$ or $T'_{k-1}$, and all vertices in the path in $T$ between $a_{k-1}$ and $a'_{k-1}$ are on the top stab line $\ell_k$. By Lemma~\ref{lem:construct_D_l}\ref{it:Dl-SRIG}, there exists a $(k-2)$-(exactly) stabbed rectangle intersection representation $\mathcal{R}_3$ of $T_{k-2}$ using stab lines $\ell_2,\ell_3,\ldots,\ell_{k-1}$ such that for $v,w\in V(T_{k-2})$, $span(v)\subseteq span(w)$ if $w$ is an ancestor of $v$ in $T_{k-2}$, and the only vertices in $T_{k-2}$ that are on the stab line $\ell_{k-2}$ are the ones in $N[a_{k-2}]$. It can be seen as shown in Figure~\ref{fig:ksrigrep} that $\mathcal{R}_1$, $\mathcal{R}_2$ and $\mathcal{R}_3$ can be combined and rectangles for the vertices in $N[c]$ can be added to obtain a $k$-stabbed rectangle intersection representation of $T$ in which for any $x\in V(H_{k-1})$, $span(x)\subseteq span(b_k)$ and for any $x\in V(T_{k-2})$, $span(x)\subseteq span(b_{k-1})$.

\begin{figure}
	\centering
	\includegraphics[scale=0.75,page=25]{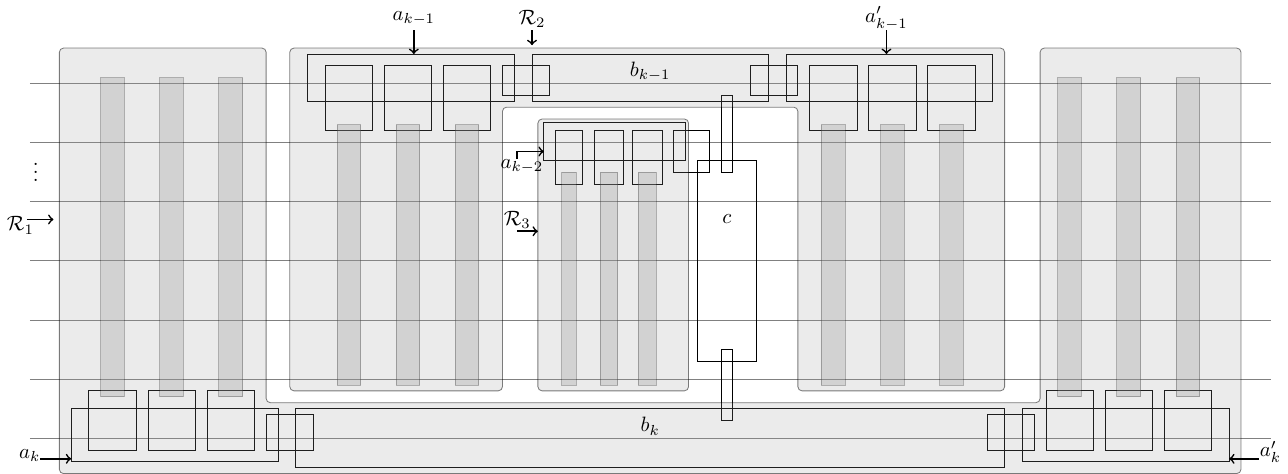}
	\caption{A schematic diagram of a $k$-stabbed rectangle intersection representation of $T$.}\label{fig:ksrigrep}
\end{figure}

Suppose for the sake of contradiction that $T$ is $k$-ESRIG. This part of the proof proceeds very similarly to the proof of Theorem~\ref{thm:construct_not_k-SIG}. As in that proof, we let $\mathcal{R}$ be a $(k+2)$-exactly stabbed rectangle intersection representation of $T$ in which the top and bottom stab lines do not intersect any rectangle and let $A$ be a good region that contains all the rectangles of $\mathcal{R}$. As $T_k$ and $T'_k$ are $k$-SRIG but not $(k-1)$-SRIG, we have $|\mathcal{L}_{\mathcal{R}}(A)|=k$ and there are $\mathcal{L}_{\mathcal{R}}(A)$-spanning paths in both $T_k$ and $T'_k$. Let $X_1$ and $X_2$ be minimal $\mathcal{L}_{\mathcal{R}}(A)$-spanning paths in $T_k$ and $T'_k$ respectively. Let $X$ be an induced path in $T$ that connects some vertex of $X_1$ and some vertex of $X_2$ such that no internal vertex of $X$ belongs to either $X_1$ or $X_2$. Note that $X$ is a subgraph of $H_k$ that contains $b_k$. Let $(A_t,A_b)=\Delta(\mathcal{R},A,X_1,X_2,X)$.


Since there is a path in $T_A=T$ from $c\in V(T_A)$ to a vertex in $X$ (in this case, $b_k$) that misses both $X_1$ and $X_2$, we know by Lemma~\ref{lem:rtorrb} that $r_c$ is contained in $A_t$ or $A_b$. We shall assume without loss of generality that $r_c$ is contained in $A_t$. Let $T^*=T-(V(H_k)\cup N[b_k])$. Since there is a path in $T_A$ from $c$ to each vertex of $T^*$ that misses $X_1$, $X_2$ and $X$, we can use Lemma~\ref{lem:winrtorrb} to infer that $T^*$ is a connected induced subgraph of $T_{A_t}$.

\begin{claim}
Both the vertices $a_k$ and $a'_k$ are on the bottom stab line in $\mathcal{L}_{\mathcal{R}}(A)$.
\end{claim} 

\noindent\textit{Proof.} Let $X'$ be the path in $T_A$ between $a_k$ and $a'_k$. Clearly, $X'$ has length 4 and is a subpath of $X$. The tree $T^*$ contains $T_{k-1}$ as an induced subgraph, and is therefore not $(k-2)$-SRIG by Lemma~\ref{lem:construct_D_l}\ref{it:Dl-not_l-1-SRIG}. Hence, $|\mathcal{L}_{\mathcal{R}}(T^*)|\geq k-1$. Since $T^*$ contains the vertex $c$ that has a path to a vertex in $X'$ which misses $a_k,a'_k,X_1,X_2$ and $X-V(X')$, we can use Lemma~\ref{lem:rootinbottom} to infer that at least one of $a_k$ and $a'_k$ is not on the top stab line in $\mathcal{L}_{\mathcal{R}}(A)$. Notice that the graph induced by $V(T_k)\cup V(T'_k)\cup V(X')$ in $T=T_A$ is isomorphic to $J_k$. This means that there is a $k$-exactly stabbed rectangle intersection representation of $J_k$ contained in the region $A$. Using Lemma~\ref{lem:construct_D_l}\ref{it:J_l-unique_rep}, we can now conclude that both $a_k$ and $a'_k$ are on the bottom stab line in $\mathcal{L}_{\mathcal{R}}(A)$. This completes the proof of the claim.
\medskip

From here onwards, we shall let $B=A_t$, for ease of notation. From the above arguments, we know that $|\mathcal{L}_{\mathcal{R}}(T^*)|\geq k-1$ and $T^*$ is a connected induced subgraph of $T_B$. Therefore, $|\mathcal{L}_{\mathcal{R}}(B)|\geq k-1$. By Lemma~\ref{lem:goodregion}, this means that $B$ is a good region and by Lemma~\ref{lem:bottomstab}\ref{it:decreasestab}, we can conclude that $|\mathcal{L}_{\mathcal{R}}(B)|=k-1$. Now, $T_{k-1}$ and $T'_{k-1}$ are two neighbour-disjoint subtrees of $T^*$ that are $(k-1)$-SRIG but not $(k-2)$-SRIG. This means that there exist minimal $\mathcal{L}_{\mathcal{R}}(B)$-spanning induced paths $Y_1$ in $T_{k-1}$ and $Y_2$ in $T'_{k-1}$. Let $Y$ be an induced path in $T^*$ that connects some vertex of $Y_1$ and some vertex of $Y_2$ such that no internal vertex of $Y$ belongs to either $Y_1$ or $Y_2$. Note that $Y$ is a subgraph of $H_{k-1}$ that contains $b_{k-1}$. Let $(B_t,B_b)=\Delta(\mathcal{R},B,Y_1,Y_2,Y)$.

Since there is a path in $T^*$ from $c$ to a vertex in $Y$ (in this case, $b_{k-1}$) that misses both $Y_1$ and $Y_2$, we know by Lemma~\ref{lem:rtorrb} that $r_c$ is contained in $B_t$ or $B_b$. As explained in the proof of Theorem~\ref{thm:construct_not_k-SIG}, it can be shown that $r_c$ is contained in $B_b$ (if $r_c$ is contained in $B_t$, then there could not have been a path in $T$ between $c$ and the vertex $b_k$ in $X$ that misses $Y_1$, $Y_2$ and $Y$). Let $T^{**}=T^*-(V(H_{k-1})\cup N[b_{k-1}])$. Since there is a path in $T_B$ from $c$ to each vertex of $T^{**}$ that misses $Y_1$, $Y_2$ and $Y$, we can use Lemma~\ref{lem:winrtorrb} to infer that $T^{**}$ is a connected induced subgraph of $T_{B_b}$. 

\begin{claim}
Both the vertices $a_{k-1}$ and $a'_{k-1}$ are on the top stab line in $\mathcal{L}_{\mathcal{R}}(A)$.
\end{claim} 

\noindent\textit{Proof.} Let $Y'$ be the path in $T$ between $a_{k-1}$ and $a'_{k-1}$. Clearly, $Y'$ has length 4 and is a subpath of $Y$. The tree $T^{**}$ contains $T_{k-2}$ as an induced subgraph, and is therefore not $(k-3)$-SRIG, implying that $|\mathcal{L}_{\mathcal{R}}(T^{**})|\geq k-2$. Since $T^{**}$ contains the vertex $c$ that has a path to a vertex in $Y'$ which misses $a_{k-1},$ $a'_{k-1},$ $Y_1,Y_2$ and $Y-V(Y')$, we can use Lemma~\ref{lem:rootinbottom} to infer that at least one of $a_{k-1}$ and $a'_{k-1}$ is not on the bottom stab line in $\mathcal{L}_{\mathcal{R}}(B)$. Notice that the graph induced by $V(T_{k-1})\cup V(T'_{k-1})\cup V(Y')$ in $T^{*}$ is isomorphic to $J_{k-1}$. This means that there is a $(k-1)$-exactly stabbed rectangle intersection representation contained in the region $B$. Using Lemma~\ref{lem:construct_D_l}\ref{it:J_l-unique_rep}, we can now conclude that both $a_{k-1}$ and $a'_{k-1}$ are on the top stab line in $\mathcal{L}_{\mathcal{R}}(B)$. Now since $B=A_t$ and $|\mathcal{L}_{\mathcal{R}}(B)|=|\mathcal{L}_{\mathcal{R}}(A)|-1$, we know by Lemma~\ref{lem:bottomstab}\ref{it:decreasestab} that the top stab line in $\mathcal{L}_{\mathcal{R}}(B)$ is also the top stab line in $\mathcal{L}_{\mathcal{R}}(A)$. This completes the proof of the claim.
\medskip

Let $\ell_1,\ell_2,\ldots,\ell_k$ be the stab lines in $\mathcal{L}_{\mathcal{R}}(A)$ in order from  bottom to top. Now, the fact that each rectangle in $\mathcal{R}$ intersects exactly one stab line gives us several observations. Since there is a path of length 2 between $b_k$ and $a_k$ in $T$, and because our first claim tells us that $a_k$ is on $\ell_1$, we can conclude that $b_k$ is not on any of the stab lines in $\{\ell_4,\ell_5,\ldots,\ell_k\}$. Similarly, our second claim tells us that $a_{k-1}$ is on $\ell_k$, and then the fact that there is a path of length 2 between $a_{k-1}$ and $b_{k-1}$ implies that $b_{k-1}$ cannot be on any stab line in $\{\ell_{k-3},\ell_{k-4},\ldots,\ell_2,\ell_1\}$. Now, since there is a path of length 4 between $b_k$ and $b_{k-1}$, there can be at most $3$ stab lines between $\ell_3$ and $\ell_{k-2}$. But this contradicts the fact that $k\geq 10$.
\end{proof}

\section{Conclusions}\label{sec:conclude}

A direction of further research could be to investigate the class of 2-SRIGs and try to characterize this class of graphs.

\begin{question}
Develop a forbidden structure characterization and/or a polynomial-time recognition algorithm for 2-SRIGs.
\end{question}

Note that Theorem~\ref{thm:block2sig} gives such a characterization of the 2-SRIGs within the class of block graphs. This theorem shows that within the class of block graphs, those graphs that do not contain asteroidal-(non-interval) subgraphs are exactly the 2-SRIGs. From the characterization of interval graphs by Lekkerkerker and Boland (Theorem~\ref{thm:lb}), we know that the absence of asteroidal triples characterizes the 1-SRIGs within chordal graphs. Therefore, a natural question is whether the absence of asteroidal-(non-interval) subgraphs is enough to characterize the 2-SRIGs within chordal graphs (note that block graphs are a subclass of chordal graphs). The answer to this question is negative, as we have shown in Theorem~\ref{thm:splitnot2-SIG} that there are split graphs that are not 2-SRIG. Split graphs are chordal and clearly, no split graph can contain asteroidal-(non-interval) subgraphs, as for any three connected induced subgraphs that are pairwise neighbour-disjoint in a split graph, at least two of them will contain just one vertex each. This gives rise to the following question.

\begin{question}
Find a forbidden structure characterization for chordal graphs (resp. split graphs) that are 2-SRIG. Can chordal graphs (resp. split graphs) that are 2-SRIG be recognized in polynomial-time?
\end{question}

We have shown that any split graph with boxicity at most 2 is 3-SRIG and that there exists a split graph which is 3-SRIG but not 2-SRIG. Therefore, following question is interesting.

\begin{question}
What is the complexity of recognizing split graphs that are 3-SRIG?
\end{question}

Note that by Theorem~\ref{thm:split3sig}, the above problem is equivalent to the problem of recognizing split graphs that have boxicity at most 2. This problem assumes significance in light of the fact that recognizing split graphs that have boxicity at most 3 is NP-complete~\cite{adiga2010}.

We constructed polynomial-time algorithms that check if $stab(G)\leq 2$ for any block graph $G$, and if $stab(T)\leq 3$ for any tree $T$. Therefore, the following are natural questions in this direction.

\begin{question}
For a given block graph $G$, is it possible to determine $stab(G)$ in polynomial-time?
\end{question} 

\begin{question}
For a given tree $T$, is it possible to determine $stab(T)$ in polynomial-time?
\end{question} 
We showed that $K_{4,4}$ is not $k$-ESRIG for any finite $k$, but is 4-SRIG. Here, the question arises as to how high the exact stab number of an exactly stabbable graph can be with respect to its stab number. Theorem~\ref{thm:blockstab} shows that trees are exactly stabbable and Theorem~\ref{thm:srig-parameter-diff} shows a tree $T$ such that $estab(T)>stab(T)$ (in fact, it is an easy exercise to show that $estab(T)=stab(T)+1$). The following questions are therefore of interest.

\begin{question}
Is there a constant $c$ such that for any tree $T$ we have, $estab(T)-stab(T)\leq c$ or $\frac{estab(T)}{stab(T)}\leq c$?
\end{question} 

\begin{question}
	For a given tree $T$, is it possible to determine $estab(T)$ in polynomial-time?
\end{question}

We constructed graphs on $n$ vertices ($(\sqrt{n},\sqrt{n})$-grids) which have stab number $\Omega(\sqrt{n})$. It can be asked if there are families of graphs which have asympotically larger stab number.

\begin{question}
Is there a class $\mathcal{C}$ of rectangle intersection graphs such that $stab(\mathcal{C},n)=\omega(\sqrt{n})$?
\end{question}

\bibliography{references}

\begin{thebibliography}{10}

\bibitem{adiga2010}
Abhijin Adiga, Diptendu Bhowmick, and L.~Sunil Chandran.
\newblock The hardness of approximating the boxicity, cubicity and threshold
  dimension of a graph.
\newblock {\em Discrete Applied Mathematics}, 158(16):1719--1726, 2010.

\bibitem{adiga2014}
Abhijin Adiga, L.~Sunil Chandran, and Naveen Sivadasan.
\newblock Lower bounds for boxicity.
\newblock {\em Combinatorica}, 34(6):631--655, 2014.

\bibitem{agarwal1998}
Pankaj~K. Agarwal, Marc Van~Kreveld, and Subhash Suri.
\newblock Label placement by maximum independent set in rectangles.
\newblock {\em Computational Geometry}, 11(3-4):209--218, 1998.

\bibitem{asplundgrunbaum}
Edgar Asplund and Branko Gr{\"u}nbaum.
\newblock On a coloring problem.
\newblock {\em Mathematica Scandinavica}, 8(1):181--188, 1960.

\bibitem{babu2014}
Jasine Babu, Manu Basavaraju, L.~Sunil Chandran, Deepak Rajendraprasad, and
  Naveen Sivadasan.
\newblock Approximating the cubicity of trees.
\newblock {\em arXiv:1402.6310}, 2014.

\bibitem{bhore2015}
Sujoy~Kumar Bhore, Dibyayan Chakraborty, Sandip Das, and Sagnik Sen.
\newblock On a special class of boxicity~2 graphs.
\newblock In {\em Algorithms and Discrete Applied Mathematics: First
  International Conference}, pages 157--168, 2015.

\bibitem{chan2004}
Timothy~M. Chan.
\newblock A note on maximum independent sets in rectangle intersection graphs.
\newblock {\em Information Processing Letters}, 89(1):19--23, 2004.

\bibitem{chandran2008}
L.~Sunil Chandran, Mathew~C. Francis, and Naveen Sivadasan.
\newblock Boxicity and maximum degree.
\newblock {\em J. Comb. Theory, Ser. {B}}, 98(2):443--445, 2008.

\bibitem{Chandran2016}
L.~Sunil Chandran, Rogers Mathew, and Deepak Rajendraprasad.
\newblock Upper bound on cubicity in terms of boxicity for graphs of low
  chromatic number.
\newblock {\em Discrete Mathematics}, 339(2):443 -- 446, 2016.

\bibitem{chandran2007}
L.~Sunil Chandran and Naveen Sivadasan.
\newblock Boxicity and treewidth.
\newblock {\em Journal of Combinatorial Theory, Series B}, 97(5):733 -- 744,
  2007.

\bibitem{corneil2009}
Derek~G. Corneil, Stephan Olariu, and Lorna Stewart.
\newblock The {LBFS} structure and recognition of interval graphs.
\newblock {\em SIAM Journal on Discrete Mathematics}, 23(4):1905--1953, 2009.

\bibitem{correa2014}
Jos{\'e}~R. Correa, Laurent Feuilloley, and Jos{\'e}~A. Soto.
\newblock Independent and hitting sets of rectangles intersecting a diagonal
  line.
\newblock In {\em Latin American Symposium on Theoretical Informatics}, pages
  35--46. Springer, 2014.

\bibitem{cozzens1983}
Margaret~B. Cozzens and Fred~S. Roberts.
\newblock Computing the boxicity of a graph by covering its complement by
  cointerval graphs.
\newblock {\em Discrete Applied Mathematics}, 6(3):217--228, 1983.

\bibitem{diestel}
Reinhard Diestel.
\newblock {\em Graph Theory}.
\newblock Electronic library of mathematics. Springer, 2006.

\bibitem{ellis2008}
John Ellis and Robert Warren.
\newblock Lower bounds on the pathwidth of some grid-like graphs.
\newblock {\em Discrete Applied Mathematics}, 156(5):545--555, 2008.

\bibitem{erlebach2010}
Thomas Erlebach and Erik~Jan Van~Leeuwen.
\newblock {PTAS} for weighted set cover on unit squares.
\newblock In {\em Approximation, Randomization, and Combinatorial Optimization.
  Algorithms and Techniques}, pages 166--177. Springer, 2010.

\bibitem{esperet2013}
Louis Esperet and Gwena{\"e}l Joret.
\newblock Boxicity of graphs on surfaces.
\newblock {\em Graphs and Combinatorics}, pages 1--11, 2013.

\bibitem{imai1983}
Hiroshi Imai and Takao Asano.
\newblock Finding the connected components and a maximum clique of an
  intersection graph of rectangles in the plane.
\newblock {\em Journal of Algorithms}, 4(4):310--323, 1983.

\bibitem{kratochvil1994}
Jan Kratochv{\'\i}l.
\newblock A special planar satisfiability problem and a consequence of its
  {NP}-completeness.
\newblock {\em Discrete Applied Mathematics}, 52(3):233--252, 1994.

\bibitem{kratochvil1990}
Jan Kratochv{\'\i}l and Jaroslav Ne{\v{s}}et{\v{r}}il.
\newblock Independent set and clique problems in intersection-defined classes
  of graphs.
\newblock {\em Commentationes Mathematicae Universitatis Carolinae},
  31(1):85--93, 1990.

\bibitem{lekkerkerkerboland}
C.~Lekkerkerker and J.~Boland.
\newblock Representation of a finite graph by a set of intervals on the real
  line.
\newblock {\em Fundamenta Mathematicae}, 51(1):45--64, 1962.

\bibitem{suderman2004}
Matthew Suderman.
\newblock Pathwidth and layered drawings of trees.
\newblock {\em International Journal of Computational Geometry \&
  Applications}, 14(03):203--225, 2004.

\bibitem{yannakakis1982}
Mihalis Yannakakis.
\newblock The complexity of the partial order dimension problem.
\newblock {\em SIAM Journal on Algebraic Discrete Methods}, 3(3):351--358,
  1982.

\end{thebibliography}
\end{document}